\documentclass[11pt]{article}
\usepackage[usenames,dvipsnames,svgnames,table]{xcolor}
\usepackage{enumitem}
\usepackage{graphicx}
\usepackage{fancybox}
\usepackage{comment}
\usepackage{xcolor}
\usepackage{nameref}

\definecolor{ForestGreen}{rgb}{0.1333,0.5451,0.1333}
\definecolor{DarkRed}{rgb}{0.8,0,0}
\definecolor{Red}{rgb}{1,0,0}
\usepackage[linktocpage=true,
pagebackref=true,colorlinks,
linkcolor=ForestGreen,citecolor=ForestGreen,
bookmarks,bookmarksopen,bookmarksnumbered]
{hyperref}

\usepackage{amsmath}
\usepackage{amssymb}
\usepackage{amsthm}

\usepackage[ruled, linesnumbered]{algorithm2e}

\usepackage{subfig}

\usepackage{thm-restate}

\usepackage{float}
\usepackage{cleveref}

\newtheorem{theorem}{Theorem}[section]

\newtheorem{corollary}[theorem]{Corollary}
\newtheorem{lemma}[theorem]{Lemma}
\newtheorem{observation}[theorem]{Observation}

\newtheorem{fact}[theorem]{Fact}

\newtheorem{definition}[theorem]{Definition}
\newtheorem{remark}[theorem]{Remark}

\newtheorem*{theorem*}{Theorem}
\newtheorem*{corollary*}{Corollary}
\newtheorem*{conjecture*}{Conjecture}
\newtheorem*{lemma*}{Lemma}
\newtheorem*{thm*}{Theorem}
\newtheorem*{prop*}{Proposition}
\newtheorem*{obs*}{Observation}
\newtheorem*{definition*}{Definition}

\newtheorem*{remark*}{Remark}
\newtheorem*{rec*}{Recommendation}

\newenvironment{fminipage}%
  {\begin{Sbox}\begin{minipage}}%
  {\end{minipage}\end{Sbox}\fbox{\TheSbox}}

\def\ceil#1{\left\lceil #1 \right\rceil}

\def\norm#1{\left\| #1 \right\|}

\def\aa{\pmb{\mathit{a}}}
\newcommand\bb{\boldsymbol{\mathit{b}}}

\newcommand\dd{\boldsymbol{\mathit{d}}}
\newcommand\ee{\boldsymbol{\mathit{e}}}

\newcommand\vv{\boldsymbol{\mathit{v}}}

\newcommand\yy{\boldsymbol{\mathit{y}}}

\newcommand\xx{\boldsymbol{\mathit{x}}}

\newcommand\veczero{\boldsymbol{0}}
\newcommand\vecone{\boldsymbol{1}}

\renewcommand\AA{\boldsymbol{\mathit{A}}}
\newcommand\BB{\boldsymbol{\mathit{B}}}
\newcommand\CC{\boldsymbol{\mathit{C}}}
\newcommand\DD{\boldsymbol{\mathit{D}}}
\newcommand\EE{\boldsymbol{\mathit{E}}}

\newcommand\HH{\boldsymbol{{H}}}
\newcommand\II{\boldsymbol{\mathit{I}}}

\newcommand\NN{\boldsymbol{\mathit{N}}}
\newcommand\MM{\boldsymbol{\mathit{M}}}
\newcommand\LL{\boldsymbol{\mathit{L}}}

\newcommand\RR{\boldsymbol{\mathit{R}}}

\newcommand\UU{\boldsymbol{\mathit{U}}}

\newcommand\VV{\boldsymbol{\mathit{V}}}

\newcommand\AAtil{\boldsymbol{\mathit{\tilde{A}}}}

\newcommand\AAn{\boldsymbol{\mathcal{A}}}
\newcommand\ZZ{\boldsymbol{\mathit{Z}}}

\newcommand\Otil{\widetilde{O}}

\newcommand\R{\mathbb{R}}

\newcommand{\schurto}[2]{\ensuremath{\textsc{Sc}\!\left[#1\right]_{#2}}}

\DeclareMathOperator{\nnz}{nnz}

\DeclareMathOperator*{\diag}{diag}

\newcommand{\vol}{\operatorname{vol}}

\DeclareMathOperator*{\im}{im}

\newcommand{\dir}{\protect\underrightarrow}

\newcommand{\Exp}{\operatorname{Exp}}

\usepackage{fullpage}

\newcommand\chain{\textsc{ChainConstruction}}
\newcommand\sparsify{\textsc{GlobalSparsification}}
\newcommand\sparsesquare{\textsc{SparseSquare}}
\newcommand\sparseproduct{\textsc{SparseProduct}}
\newcommand\sparsebipartite{\textsc{SparseBipartite}}
\newcommand\patchbipartite{\textsc{BipartitePatching}}
\newcommand\expdecomp{\textsc{ExpDecomp}}
\newcommand\sparsedir{\textsc{SparsifyDirected}}
\newcommand\specspardeg{\textsc{SpectralSparsifyDeg}}
\newcommand\squarechain{\textsc{SquareChain}}
\newcommand\solve{\textsc{SolveRecursive}}
\newcommand\solveEulerian{\textsc{SolveEulerian}}
\newcommand\peel{\textsc{Peel}}
\newcommand\richardson{\textsc{PreconRichardson}}

\title{Derandomizing Directed Random Walks in Almost-Linear Time}

\newcommand*\samethanks[1][\value{footnote}]{\footnotemark[#1]}
\author{Rasmus Kyng\thanks{The research leading to these results has received funding from grant no. 200021 204787 of the Swiss National Science Foundation.}, \\ ETH Zurich \\ kyng@inf.ethz.ch \\  \and Simon Meierhans\samethanks, \\ ETH Zurich \\ mesimon@inf.ethz.ch \and Maximilian Probst Gutenberg\samethanks, \\ ETH Zurich \\ maximilian.probst@inf.ethz.ch}

\begin{document}

\maketitle
\pagenumbering{gobble}

\begin{abstract}

    In this article, we present the first $\emph{deterministic}$ directed Laplacian $\LL$ systems solver that runs in  time almost-linear in the number of non-zero entries of $\LL$. Previous reductions imply the first deterministic almost-linear time algorithms for computing various fundamental quantities on directed graphs including stationary distributions, personalized PageRank, hitting times and escape probabilities. 
   
    We obtain these results by introducing \emph{partial symmetrization}, 
   a new technique that 
   makes the Laplacian of an Eulerian directed graph ``less directed''
   in a useful sense, which may be of independent interest.
   The usefulness of this technique comes from two key observations:
   Firstly, the partially symmetrized Laplacian preconditions the original Eulerian Laplacian well in Richardson iteration, enabling us to construct a solver for the original matrix from a solver for the partially symmetrized one.
   Secondly, the undirected structure in the partially symmetrized Laplacian makes it possible to sparsify the matrix \emph{very crudely}, i.e. with large 
   spectral error, and still show that Richardson iterations convergence
   when using the sparsified matrix as a preconditioner. 
   This allows us to develop deterministic sparsification tools for the partially symmetrized Laplacian.
   
   Together with previous reductions from directed Laplacians to Eulerian Laplacians, our technique results in the first deterministic almost-linear time algorithm for solving linear equations in directed Laplacians. To emphasize the generality of our new technique, we show that two prominent existing (randomized) frameworks for solving linear equations in Eulerian Laplacians can be derandomized in this way: the squaring-based framework of \cite{cohen2016almostlineartimeconference} and the sparsified Cholesky-based framework of \cite{peng2021sparsifiedconference}.

\end{abstract}

\clearpage
\pagenumbering{arabic}

\section{Introduction}

The development of spectral graph sparsification and nearly-linear time solvers for Laplacian linear equations initiated by the seminal article of Spielman and Teng \cite{storiginal}, that has since been split into three parts \cite{st_part_1, st_part_2, st_part_3}, is foundational for algorithmic spectral graph theory and one of the success stories in the design of graph algorithms. Until recently, spectral techniques were mainly used for analysing undirected graphs. While many of the techniques developed in algorithmic spectral graph theory heavily use that Laplacians of undirected graphs are symmetric, positive semi-definite matrices, Cohen-Kelner-Peebles-Peng-Sidford-Vladu \cite{cohen16fasterconference} first demonstrated that the structure of directed Laplacians can be used to accelerate linear equation solvers. Together with Rao they later introduced the first notion of spectral approximation for directed Laplacians \cite{cohen2016almostlineartimeconference}, which enabled them to develop sparsification tools in the directed setting. These tools were used to build the first almost-linear time directed Laplacian solver, operating in the framework of an undirected Laplacian solver by Peng and Spielman \cite{ps14}. The runtime was improved to nearly-linear time by Cohen-Kelner-Kyng-Peebles-Peng-Rao-Sidford  \cite{cohen2018solving} and further improved by Peng and Song \cite{peng2021sparsifiedconference} very recently. Both operate within sparsified-Cholesky frameworks for solving undirected Laplacian linear equations: \cite{cohen2018solving} uses the framework of Kyng and Sachdeva \cite{ks16} and \cite{peng2021sparsifiedconference} uses the framework of Kyng-Lee-Peng-Sachdeva-Spielman \cite{kyng2015sparsifiedconference}. 

All previous fast algorithms for solving directed Laplacian linear equations rely on sampling for globally sparsifying directed graphs while retaining a spectral $(1 \pm \epsilon)$-approximation guarantee for some  $\epsilon < 1$, according to the notion of approximation by $\cite{cohen2016almostlineartimeconference}$\footnote{When we refer to $(1 \pm \epsilon)$-approximations for Eulerian Laplacians in the introduction and overview we mean $\epsilon$-approximations as in \Cref{def:mat_approx}. We use this naming for convenience when comparing to undirected approximations.}. This suggests two main approaches to derandomizing these solvers: (1) developing a fast deterministic $(1 \pm \epsilon)$-approximate spectral sparsification routine or (2) introducing cruder forms of approximation and adapting the algorithms to cope with weaker guarantees. However, even for undirected graphs, no deterministic almost-linear time algorithms that achieve $(1 \pm \epsilon)$-approximate spectral sparsification are known, and this is a major obstacle to approach (1). 
We circumvent this issue by instead taking the route (2): we develop a new way of measuring approximation via the suitability as a preconditoner in Richardson, which allows cruder guarantees. 
We achieve this by introducing a ``robustification'' step before sparsification, which we call $\beta$-partial symmetrization.
Partial symmetrization counteracts the fragility of directed approximations and is at the core of our new crude deterministic sparsification procedure for Eulerian Laplacians. It lets us derandomize the frameworks of \cite{cohen2016almostlineartimeconference} and \cite{peng2021sparsifiedconference} with an almost-linear runtime. 

\subsection{Prior Work}

\paragraph{Undirected Laplacian Solvers. } The first nearly-linear time Laplacian solver \cite{st_part_1, st_part_2, st_part_3} sparked the development of a field producing a whole host of distinct algorithms for solving Laplacian linear equations \cite{kmp14_1, kmp11, kos13, CohenKPPR14, ps14, kyng2015sparsifiedconference, ks16, js21}. Of these, our algorithm is most similar to \cite{ps14} and its directed counterpart \cite{cohen2016almostlineartimeconference}, which works by repeatedly squaring the adjacency matrix. In \Cref{sec:spars_cholesky} we consider a Cholesky-factorisation based framework, which is motivated by \cite{kyng2015sparsifiedconference} in the undirected setting and was recently translated to the directed setting by \cite{peng2021sparsifiedconference}. 

\paragraph{Spectral Sparsification. } A crucial building block for spectral graph algorithms, including linear equation solvers, is the ability to sparsify undirected graphs while preserving their spectral properties. This is a stronger notion of sparsification than cut-sparsifiers, which only preserve the approximate size of cuts. Such spectral sparsifiers were first introduced in \cite{st_part_2}, and later strengthened and simplified by Spielman and Srivastava \cite{ss11}. It is known that for every $n$-vertex graph there exists a spectral sparsifier with $O(n)$-edges, and such sparsifiers can be constructed deterministically, as shown by Batson, Spielman, and  Srivastava \cite{bss12}. However, no deterministic almost-linear time algorithms are known for $(1 \pm \epsilon)$-spectral sparsification. Recently, Chuzhoy-Gao-Li-Nanongkai-Peng-Saranurak \cite{chuzhoy20detconf} presented an almost-linear time algorithm achieving $n^{o(1)}$-spectral sparsifiers for undirected graphs via deterministic expander decompositions. Our algorithm relies on both their sparsification and their expander decomposition results. 

\paragraph{Directed Laplacian Solvers. } Before \cite{cohen16fasterconference, cohen2016almostlineartimeconference} it was unclear that directed Laplacians, which are a natural generalization of undirected Laplacians to directed graphs, also exhibit properties that allow for  useful notions of sparsification and/or accelerated solving of linear equations. The reduction to strongly connected Eulerian Laplacians recovered some of the spectral properties of undirected Laplacians, and allowed for the development of a useful notion of sparsification. Current fast algorithms for solving Eulerian Laplacian linear equations either follow a squaring \cite{cohen2016almostlineartimeconference} or a sparsified-Cholesky approach \cite{cohen2018solving, peng2021sparsifiedconference}. Both rely on spectral sparsification techniques developed in \cite{cohen2016almostlineartimeconference}. 

\paragraph{Low Space Algorithms. } Recently, the first deterministic $\tilde{O}(\log N)$-space solver for Eulerian Laplacians was introduced by Ahmadinejad-Kelner-Murtagh-Peebles-Sidford-Vadhan \cite{akm20}. They show that the Rozeman and Vadhan \cite{rv05} deterministic squaring conserves approximation under squaring for a new, stronger measure of approximation. The directed to Eulerian reduction remains a major obstacle to solving directed Laplacian linear equations in small space. 

\paragraph{Applications of Directed Laplacian Solvers. }

There are numerous applications of directed Laplacian solvers given in Section 7 of \cite{cohen16faster}, most of which are deterministic reductions to directed Laplacian system solving. The deterministic ones include 
\begin{itemize}
    \item solving large classes of linear systems,
    \item computing personalised PageRank vectors,
    \item estimating the stationary distribution,
    \item and simulating random walks. 
\end{itemize}
Since the reductions are deterministic, we obtain deterministic almost-linear time algorithms for all these problems\footnote{Note that we require an upper bound $\kappa$ on the condition number, and hence mixing time.}.

\subsection{Our Contributions}

We introduce the first notion of crude approximation for Eulerian Laplacians $\LL$. It is defined via the suitability as a preconditioner in the Richardson iteration and sparse approximations are constructed using \emph{partial symmetrization} to increase robustness. This technique allows us to trade off additional Richardson iterations for a behaviour more akin to the undirected setting. The obtained \emph{deterministic} crude global sparsification routine ultimately allows us to derandomize two prominent frameworks for solving directed Laplacian linear equations. We summarize our main result assuming polynomially bounded edge weights and condition number.

\begin{theorem}[Informal Version of \Cref{thm:eul_square_solver}]\label{thm:mainTheoremSketchEulerian}
Given an Eulerian Laplacian $\LL_{\dir{G}}$ associated with a strongly connected Eulerian Graph $\dir{G}$ with $n$ vertices and $m$ edges, a vector $\bb \in \im(\LL_{\dir{G}})$, and a parameter $\epsilon \in (0, 1)$ the algorithm $\solveEulerian(\LL_{\dir{G}}, \epsilon)$ in time $m^{1 + o(1)} \log \epsilon^{-1}$ computes a vector $\xx$ satisfying 
\begin{align*}
    \norm{\xx - \LL_{\dir{G}}^+\bb}_{\UU_{\LL_{\dir{G}}}} \leq \epsilon \norm{\LL_{\dir{G}}^+\bb}_{\UU_{\LL_{\dir{G}}}} 
\end{align*}
where $\UU_{\AA} = (\AA + \AA^T)/2$ for any square matrix $\AA$.
\end{theorem}

Previous reductions allow us to reduce general directed Laplacian solvers to $\log^{O(1)}(n\kappa^{-1}\epsilon^{-1})$ Eulerian solves with polynomially bounded condition number and edge weights. We state the main theorem as a corollary. 

\begin{theorem}\label{thm:mainTheoremSketch}
Given a directed Laplacian $\LL_{\dir{G}} = \DD_{\dir{G}} - \AA_{\dir{G}}^T$ associated with a directed Graph $\dir{G}$ with $n$ vertices and $m$ edges, a vector $\bb \in \im(\LL_{\dir{G}})$, a parameter $\epsilon \in (0, 1)$ and an upper bound
$\kappa \geq \max(\kappa(\DD_{\dir{G}}), \kappa(\LL_{\dir{G}}))$ there is an algorithm $\textsc{SolveFull}(\LL_{\dir{G}}, \epsilon)$ that in time $m^{1 + o(1)} \log^{O(1)}(\kappa \epsilon^{-1})$ computes a vector $\xx$ satisfying 
\begin{align*}
    \norm{\xx - \LL_{\dir{G}}^+\bb}_2 \leq \epsilon \norm{\LL_{\dir{G}}^+\bb}_2
\end{align*}
where $\kappa(\AA) = \norm{\AA}_2 \cdot \norm{\AA^+}_2$ denotes the condition number of a given matrix $\AA$. 
\end{theorem}

A key application of these results is the deterministic simulation of random walks in almost-linear time as presented in Sections 7.4 and 7.5 of \cite{cohen16faster}. 

\paragraph{Crude Global Sparsification. } Our global sparsification routine is based on first increasing the robustness via $\beta$-partial symmetrization. Then we bucket by edge weight and layer partitions of undirected and unweighted graphs into expanders. First, the directed part in such an expander can be sparsified in a crude way that conserves the degrees. Then, the undirected part can be replaced with a sparse expander. Since the undirected part is scaled with a factor $\beta$ by $\beta$-partial symmetrization it dominates the directed part which allows us to bound the error of the first step. 

\paragraph{Open Questions. } Our techniques, in particular $\beta$-partial symmetrization, might be of interest for deterministically simulating random walks in small space. Further, the development of a composable notion of crude approximation for Eulerian Laplacians is an interesting task. However, there are Eulerian Laplcians that do not precondtion each other regardless of stepsize, which is a major obstruction (See \Cref{sec:cycle}).

\section{Overview}

Existing randomized algorithms for solving linear equations in directed Eulerian Laplacians can be classified as belonging to one of two general frameworks: squaring-solvers and sparsified-cholesky-solvers. Both heavily rely on $(1 \pm \epsilon)$-approximate graph sparsification techniques for two separate tasks: 1) \emph{globally sparsifying} a directed graph $\dir{G}$ with $m$-edges in time close to $m$ and 2) \emph{sparsifying the square graph} $\dir{G}^2$ in time close to $m$. The square graph $\dir{G}^2$ is the graph with adjacency matrix $\AA_{\dir{G}} \DD^{-1}_{\dir{G}}\AA_{\dir{G}}$ where $\AA_{\dir{G}}$ denotes the adjacency matrix of $\dir{G}$ and $\DD_{\dir{G}}$ denotes the diagonal matrix containing the degrees of the Eulerian graph $\dir{G}$. For some graphs the running time of 2) is sublinear in the number of edges of the sparsified graph $\dir{G}^2$ since squaring can drastically increase the density. We provide derandomized algorithms for both tasks, allowing us to derandomize two prominent frameworks. We focus on the squaring-solver \cite{cohen2016almostlineartimeconference}, and sketch the sparsified-Cholesky solver \cite{peng2021sparsifiedconference} in \Cref{sec:spars_cholesky}\footnote{\cite{peng2021sparsifiedconference} also relies on randomness for finding a linear sized almost independent subset. We provide a simple procedure to derandomize this step in \Cref{sec:rho_rcdd}.}. 

While we match the $(1 \pm \epsilon)$-approximation and roughly the runtime of randomized routines for 2), our algorithm for globally sparsifying directed graphs only achieves a crude $n^{o(1)}$-approximation guarantee. This is not surprising, since deterministic $(1 \pm \epsilon)$-approximate sparsification in almost linear time is an open problem even for undirected graphs. However, the current notions of spectral approximation for directed graphs due to \cite{cohen2016almostlineartimeconference} break down for approximation factor $\epsilon$ larger than $1$.
To address this problem, 
we directly define approximation via the suitability as a preconditioner in the Richardson iteration.
In particular, suppose we want to solve a linear equation in $\MM$, by preconditioning with a matrix $\NN$.
Roughly speaking, if, in an appropriate norm $\norm{\cdot}$, we have $\norm{\II_{\im(\MM)} - \NN^+ \MM} \leq 1-\alpha$, then we can solve the linear equation in $\MM$ to high accuracy using $\Otil(1/\alpha)$ preconditioned Richardson iterations where we apply $\MM$ and $\NN^+$ once per iteration.
A priori, it may be surprising that our approach relies on distinguishing the case ${1-\norm{\II_{\im(\MM)} - \NN^+ \MM} \geq n^{o(-1)}}$ from the case ${1-\norm{\II_{\im(\MM)} - \NN^+  \MM} \approx 0}$, however, our partial symmetrization technique and a careful choice of norm makes this possible.

Our notion of high-error approximation is not directly composable, but it does algorithmically compose: If $\AA$ is precondioned by $\BB$ and $\BB$ is preconditioned by $\CC$, then a solver for $\AA$ can be constructed with access to $\BB$ and $\CC^+$ using two layers of preconditioned Richardson.
The development of a directly composable notion of spectral approximation in directed graphs for approximation factors larger than $1$ remains an interesting open problem.

\subsection{Global Sparsification}

Given an Eulerian graph $\dir{G} = \dir{G}_0$ our sparsification routine does not directly produce a sparse graph $\tilde{\dir{G}} = \dir{G}_3$, but does so via a detour involving two other graphs $\dir{G}_1$ and $\dir{G}_2$:
\begin{itemize}
    \item $\dir{G}_0$ is the initial graph $\dir{G}$.
    \item $\dir{G}_1$ is obtained by $\beta$-partially-symmetrizing $\dir{G}_0$. It has a directed and a undirected part.
    \item $\dir{G}_2$ is obtained from $\dir{G}_1$ by sparsifying its directed part without touching the undirected part. 
    \item $\dir{G}_3$ is obtained by also sparsifying the undirected part. 
\end{itemize} 
We call such a bundle of four directed graphs a quadruple. Our construction ensures that $\LL^+_{\dir{G}_{i + 1}}$ is a suitable preconditioner for $\LL_{\dir{G}_i}$ with respect to the norm $\norm{\cdot}_{\UU_{\LL_{\dir{G}_{i + 1}}} \rightarrow \UU_{\LL_{\dir{G}_{i + 1}}}}$ for $i = 0, 1, 2$, meaning that $\LL_{\dir{G}_i}^+$ can be (approximately) applied by applying $\LL_{\dir{G}_{i+1}}^+$ for $N =  Exp(O((\log n)^{1/10}))$ times via the preconditioned Richardson iteration\footnote{Because not all approximations refer to the same norm, this does unfortunately not ensure that $\LL^+_{\dir{G}_3}$ is a suitable preconditioner for $\LL_{\dir{G}_0}$.}. Given an oracle for applying $\LL_{\dir{G}_3}^+$ we can use it $N$ times to  apply $\LL_{\dir{G}_2}^+$, and then in turn apply $\LL_{\dir{G}_1}^+$ via $N$ applications of $\LL_{\dir{G}_2}^+$, until we can finally apply $\LL_{\dir{G}_0}^+$. This yields a procedure for applying $\LL_{\dir{G}_0}^+$ that relies on $N^3 =  Exp(O((\log n)^{1/10}))$ applications of $\LL_{\dir{G}_3}^+$. Next we describe the way the graphs $\dir{G}_1, \dir{G}_2$ and $\dir{G}_3$ are constructed in more detail. Our full global sparsification results are presented in \Cref{sec:globspars}.

\begin{figure}[h]
    \centering
    \includegraphics[width=12cm]{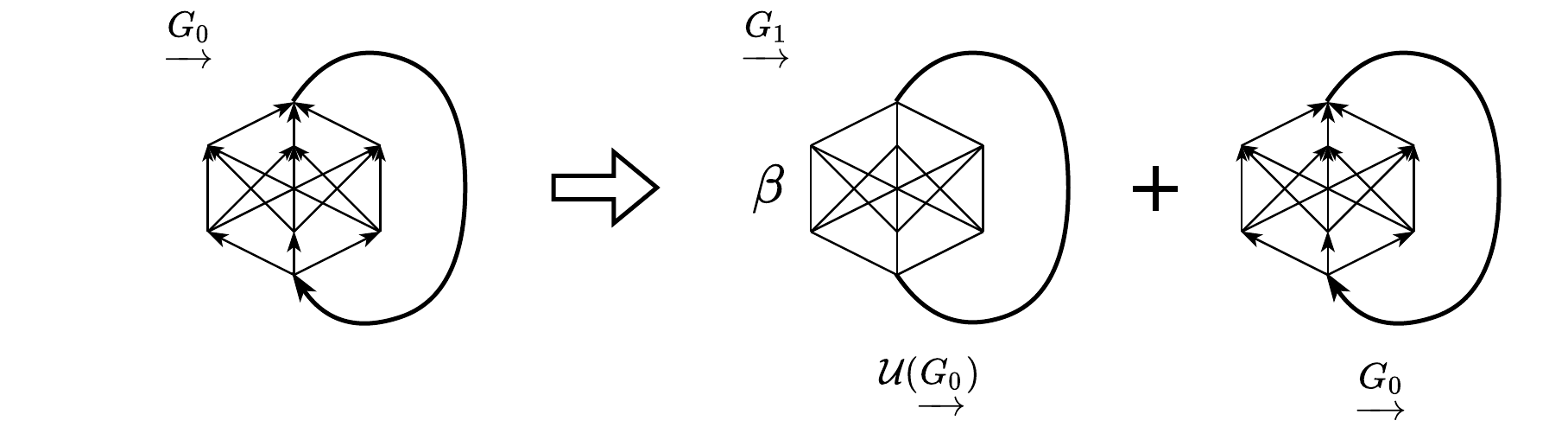}
    \caption{Obtaining $\dir{G}_1$ from $\dir{G}_0$ via $\beta$-partial symmetrization. Given the graph $\dir{G}_0$ as depicted on the left hand side, we add $\beta$ times its undirectification. This could double the amount of edges, but
    the graph $\dir{G}_1$ has much less directed structure, which we will use to sparsify it in the next steps. }
    \label{fig:g0g1}
\end{figure}

\paragraph{From $\dir{G}_0$ to $\dir{G}_1$: Robustness through Partial Symmetrization. }  Let $\mathcal{U}(\dir{G}_0)$ denote the undirectification of the graph $\dir{G}_0$, i.e. the graph obtained by replacing each directed edge with an undirected edge of half the weight. Then we simply obtain $\dir{G}_1 = \beta \cdot \mathcal{U}(\dir{G}_0) + \dir{G}_0$ where $\beta = n^{o(1)}$ is a sub-polynomial factor. Surprisingly, even though $\dir{G}_1$ removes a lot of directed structure, $\LL^+_{\dir{G}_0}$ can be applied via $O(\beta)$ applications of $\LL^+_{\dir{G}_1}$. Although $\dir{G}_1$ might seem similar to $\beta \cdot \mathcal{U}(\dir{G}_0)$, in fact, they behave very differently.
The key technical observation is that \begin{align*}
    \norm{\LL_{\mathcal{U}(\dir{G}_1)}^{+/2}(\LL_{\dir{G}_0} - \LL_{\dir{G}_1})\LL_{\mathcal{U}(\dir{G}_1)}^{+/2}}_2 = \norm{\LL_{(1 + \beta) \cdot \mathcal{U}(\dir{G}_0)}^{+/2}\LL_{\beta \cdot \mathcal{U}(\dir{G}_0)}\LL_{(1 + \beta) \cdot \mathcal{U}(\dir{G}_0)}^{+/2}}_2 = \frac{1}{\beta + 1}
\end{align*}
which relies on the fact that the directed graphs \emph{cancel out} additively, only leaving us with symmetric Laplacians. See \Cref{fig:g0g1} for an illustration.

\begin{figure}[h]
    \centering
    \includegraphics[width=14cm]{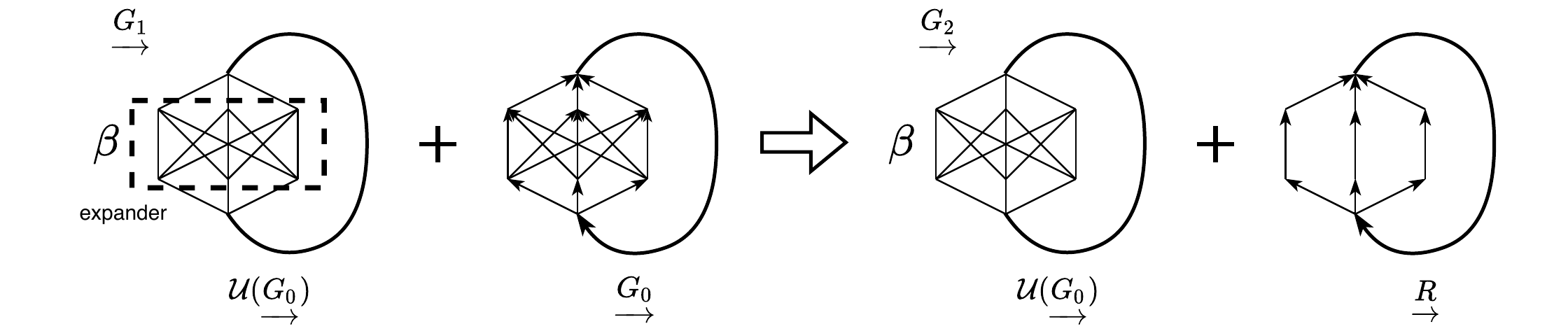}
    \caption{Obtaining $\dir{G}_2$ from $\dir{G}_1$ via patching. The dashed box contains the complete bipartite graph, which is a good expander. We use this information to drastically sparsify the directed part on the same vertex set $V'$, only ensuring that the degrees match up by increasing the weight of the edges. Since the induced subgraph on $V'$ remains a good expander when summing up, this does not alter the spectral structure of the graph by much.}
    \label{fig:g1g2}
\end{figure}

\paragraph{From $\dir{G}_1$ to $\dir{G}_2$: Patching Expander Parts. } The graph $\dir{G}_1 = \beta \cdot \mathcal{U}(\dir{G}_0) + \dir{G}_0$ is made up of two parts: an undirected graph $\beta \cdot \mathcal{U}(\dir{G}_0)$ and a directed graph $\dir{G}_0$. In this step we aim to sparsify the directed part by constructing a sparse graph $\dir{R}$ with the same in- and out-degrees as $\dir{G}_0$. Then $\dir{G}_2 = \beta \cdot \mathcal{U}(\dir{G}_0) + \dir{R}$. Our strategy for obtaining $\dir{R}$ relies on the structure of the undirected graph $\mathcal{U}(\dir{G}_0)$ and deterministic expander decompositions as presented in \cite{chuzhoy20detconf}. Our main technical observation here is that given a vertex set $V'$ so that $\mathcal{U}(\dir{G}_0)[V']$ is a good enough expander, we can replace the directed graph $\dir{G}_0[V']$ in $\dir{G}_1$ with \emph{any other} directed graph $\dir{R}'$, as long as $\dir{G}_0[V']$ and $\dir{R}'$ have exactly the same in- and out-degrees, retaining a close approximation between $\dir{G}_1 - \dir{G}_0[V'] + \dir{R}'$ and $\dir{G}_1$. Our algorithm for constructing $\dir{R}$ first removes the weighted structure from $\mathcal{U}(\dir{G}_0)$ by bucketing by edge weight, and then layers deterministic undirected and unweighted expander decompositions. These are partitions of undirected and unweighted graphs into expanding parts and a remainder, and the layering is achieved by recursing on the remainder. A single expanding part can be greedily sparsified using any sparse greedily constructed graph $\dir{R}'$ as introduced above. Summing up all of these yields $\dir{R}$. Crucially, while the error sums up over layers and buckets, the disjointness of the expander parts ensures this is not the case within a single decomposition. See \Cref{fig:g1g2} for an illustration. 

\begin{figure}[h]
    \centering
    \includegraphics[width=14cm]{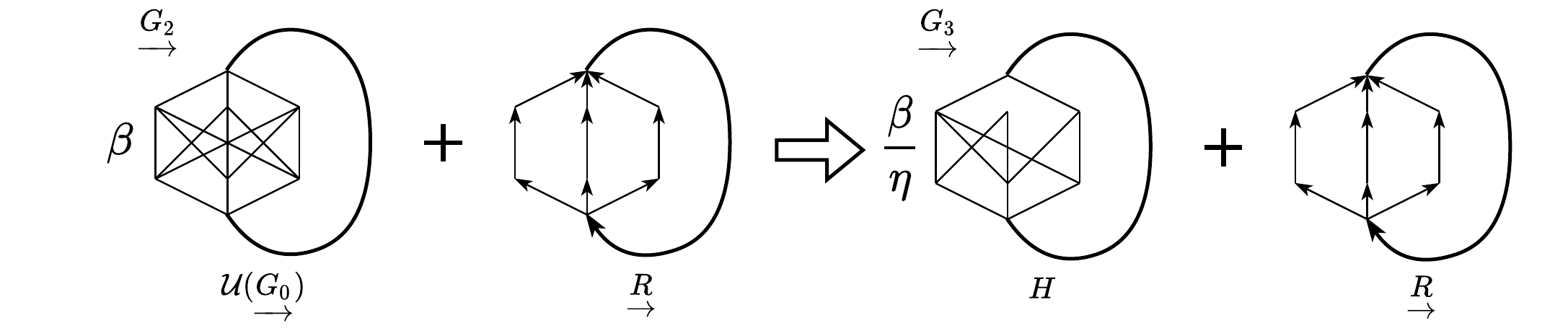}
    \caption{Obtaining $\dir{G}_3$ from $\dir{G}_2$ via undirected sparsification. The undirected part $\beta \cdot \mathcal{U}(\dir{G})$ of $\dir{G}_2$ is sparsified using previously known spectral sparsification routines for undirected graphs. Since both remaining parts are sparse, their sum is a sparse graph.}
    \label{fig:g2g3}
\end{figure}

\paragraph{From $\dir{G}_2$ to $\dir{G}_3$: Scaling an Undirected Sparsifier. } Finally we sparsify the undirected part $\beta \cdot \mathcal{U}(\dir{G}_0)$ obtaining a sparse undirected graph $H$ and let $\dir{G}_3 = H + \dir{R}$. This can be achieved via known theorems for deterministic low-accuracy undirected graph sparsification provided by \cite{chuzhoy2020deterministic}. To obtain a good preconditioner we have to scale $H$, but not $\dir{R}$, with the inverse of an appropriate rate $\eta$.
This relies crucially on our observation that when the approximation error is only on the undirected part, learning rates can be leveraged more effectively than for general Eulerian approximation.
We are left with a sum of two sparse graphs $H$ and $\dir{R}$, which is a sparse directed graph. See \Cref{fig:g2g3} for an illustration.  

\subsection{Sparsified Squaring}

Given a directed graph $\dir{G}$ with Laplacian $\LL_{\dir{G}} = \DD_{\dir{G}} - \AA^T_{\dir{G}} \in \R^{n \times n}$, the Laplacian of its square $\dir{G}^2$ is given by 
\begin{align*}
    \LL_{\dir{G}^2} = \DD_{\dir{G}} - \underbrace{\AA^T_{\dir{G}}\DD_{\dir{G}}^{-1}\AA^T_{\dir{G}}}_{\text{adjacency matrix}} = \DD_{\dir{G}} - \sum_{i = 1}^n \underbrace{\frac{1}{\DD_{\dir{G}}(i,i)} (\AA(i, :))^T \cdot (\AA(:,i))^T}_{\AA_i^T}.
\end{align*}
We consider the directed product graphs $\LL_i = \DD_i - \AA_i^T$ with adjacency matrix $\AA_i$. Consider the matrix 
\begin{align*}
    \LL = \begin{pmatrix}
    \diag(\AA_i \vecone) & \veczero \\
    \veczero & \diag(\AA_i^T \vecone) 
    \end{pmatrix}
    - 
    \begin{pmatrix}
        \veczero & \AA_i \\
        \AA_i^T & \veczero
    \end{pmatrix}
\end{align*}
which is the Laplacian of a bipartite product graph $G$. Such graphs are constant expanders, which allows for a simple trick. First, we sparsify the undirected bipartite product graph $\LL$ to high accuracy $\epsilon$, by adapting a procedure from \cite{kyng2015sparsified} to be degree preserving. We call the sparse bipartite graph we obtain $\tilde{\LL}$. Then we obtain a sparsified version $\tilde{\AA}_i^T$ of $\AA_i^T$ by simply taking the bottom left block of $\tilde{\LL}$. We show that $\tilde{\LL}_i = \DD_i - \tilde{\AA}_i^T$ is a $\epsilon/\Phi^2$-approximation of $\LL_i$, where $\Phi$ is the expansion of $G$. Using the fact that bipartite product graphs are constant expanders lets us directly translate the approximation guarantee, up to a constant overhead in runtime. We obtain an $(1 \pm \epsilon)$-approximation $\LL_{\tilde{\dir{G}}^2} = \sum_{i = 1}^n \tilde{\LL}_i$ of $\LL_{\dir{G}^2}$ with $\nnz(\LL_{\tilde{\dir{G}}^2}) \leq O(\nnz(\LL_{\dir{G}}) \epsilon^{-4})$ in almost linear time. Similar deterministic squaring techniques were developed for small space algorithms \cite{rv05, akm20}, albeit the slightly different guarantees. We believe its likely that their approach to be adapted to our setting, but for convenience we adapt a more direct approach.

\subsection{The Squaring Framework}

Consider an Eulerian Laplacian $\LL_{\dir{G}} = \DD_{\dir{G}} - \AA_{\dir{G}}^T$ where $\AA_{\dir{G}}$ is the adjacency matrix of an Eulerian graph $\dir{G}$ and $\DD_{\dir{G}}$ is the diagonal matrix containing the degrees. Then the normalised Laplacian is given by $\NN = \II - \AAn^T$ where $\AAn = \DD_{\dir{G}}^{+/2} \AA_{\dir{G}}^T \DD_{\dir{G}}^{+/2}$, having the property that $\norm{\AAn}_2 \leq 1$. The Neumann series expansion yields
\begin{align*}
    (\II - \AAn)^+\bb = \sum_{i = 0}^{\infty} \AAn^i \bb = \prod_{i = 0}^{\infty} \left(\II + \AAn^{2^k}\right) \bb
\end{align*}
for $\bb$ orthogonal to the kernel of $(\II - \AAn)^+$. Given a $1/\operatorname{poly}(n)$ lower bound on the smallest eigenvalue $\lambda_*$ of $\II - \AAn$, truncating the product expansion after $\Theta(\log n)$ factors (and hence squarings) yields a constant relative error. A very convenient equality in the same spirit is given by
\begin{align}
    \label{eq:main_id}
    (\II - \AAn)^+ \bb = (\II - \AAn^2)^+ (\II + \AAn) \bb
\end{align}
for $\bb$ orthogonal to the kernel of $(\II - \AAn)^+$ which is at the center of the squaring mechanism of \cite{cohen2016almostlineartime}, the algorithm our squaring solver resembles most. Their squaring scheme is in turn inspired by the squaring solver for symmetric Laplacians presented in \cite{ps14}. 

This leaves us with the task of solving linear equations in $\II - \AAn^2$, which is another normalised Laplacian. However, this is the normalised Laplacian of the square graph $\dir{G}^2$, and it can be shown that squaring drastically improves the condition number of the problem, such that after $k = \Theta(\log n)$ squaring steps linear equations can be solved to high accuracy quickly via a simple iterative scheme. To avoid periodic behaviours we consider the normalised adjacency matrix $\AAn_l^{(\alpha)} := \alpha \II + (1 - \alpha) \AAn_j$, which can be interpreted as adding self loops proportional to the out-degrees of $\dir{G}$. 

\paragraph{Sparsified Squaring. } Since squaring not only improves the condition number, but may also quickly increase the density of the graph, we let $\AAn_0 = \AAn$ and iteratively obtain $\AAn_{j + 1}$  by implicitly sparsifying $(\AAn_j^{(\alpha)})^2$ using our sparsified squaring technique with accuracy parameter $\epsilon$. We can conclude from \eqref{eq:main_id} that
\begin{align*}
    (\II - \AAn)^+ \bb \approx \underbrace{(1 - \alpha)^{d - 1} \left( \II - \AAn_d \right)^+ \left(\II + (\AA_{d - 1}^{(\alpha)})^2 \right) \cdots \left( \II - (\AA_0^{(\alpha)})^2 \right)}_{:= \ZZ} \bb 
\end{align*}
as we can ensure that $\bb$ is orthogonal to the known kernel of $(\II - \AAn)^+$. However, the repeated sparsification accumulates an error proportional to $\epsilon e^d$, and it is imperative that it stays below $1$ such that $\ZZ$ is an approximate pseudoinverse of $\II - \AAn$. Therefore, we have to choose $\epsilon$ proportional to $e^{-d}$. We conclude from the previous subsection that $\nnz(\AA_d) = O(\nnz(\AAn)e^{4d^2})$. Unfortunately, we cannot set $d = \Theta(\log n)$ without ending up with potentially dense matrices. Therefore, we set $d = \Theta((\log n)^{1/3})$ and have $e^{4d^2} = n^{o(1)}$.

\paragraph{Global Sparsification and Chains of Sparse Matrices. } Since $d = \Theta((\log n)^{1/3})$ squarings do not sufficiently decrease the condition number, we globally sparsify after $d$ sparsified squarings and repeat. Given $\AAn_0^{(0)} = \DD_{\dir{G}}^{+/2} \AA_{\dir{G}}\DD_{\dir{G}}^{+/2}$ for $i = 0, ..., \Theta((\log n)^{2/3})$ we iteratively construct:
\begin{itemize}
    \item Given $\AAn_0^{(i)}$, construct $\AAn_0^{(i)}, ..., \AAn_d^{(i)}$ by sparsified squaring as described in the previous paragraph. 
    \item Let $\dir{H}$ be the graph with adjacency matrix $\DD_{\dir{G}}^{1/2}\AAn_d^{(i)}\DD_{\dir{G}}^{1/2}$. Globally sparsify $\dir{H}$ obtaining $\tilde{\dir{H}}$. Then let $\AAn_0^{(i+1)} = \DD_{\dir{G}}^{+/2}\AA_{\tilde{\dir{H}}}\DD_{\dir{G}}^{+/2}$.
\end{itemize}

We discuss these collections of squaring chains linked by global sparsification in \Cref{sec:chains}. For our algorithm to run in almost-linear time, it is imperative that these chains are constructed once, and then our algorithm operates recursively on them.  

\paragraph{The Recursive Algorithm. } Our global sparsification routine allows us to solve linear equations in $\II - \AAn_0^{(i)}$ by solving $\Exp(O((\log n)^{1/10}))$ linear equations in $\II - \AAn_0^{(i+1)}$. Since linear equations in $\II - \AAn_0^{(\Theta(\log n)^{2/3})}$ are easy to solve using standard iterative procedures, this is the depth of our recursion. Therefore, the total amount of branches is $\Exp(O((\log n)^{2/3 + 1/10})) = n^{o(1)}$. Since all involved matrices contain an almost linear amount of entries, this gives an almost linear time deterministic algorithm for solving linear equations involving Eulerian Laplacians. That is, because preconditioned Richardson only does matrix vector multiplications with the matrix and the preconditioner.

\subsection{The Sparsified-Cholesky Framework}

Very recently \cite{peng2021sparsifiedconference} showed that the framework of \cite{kyng2015sparsifiedconference} directly works for Eulerian Laplacians by developing new tools for analysing the accumulation of error. Our sparsification tools can also be used to derandomize this algorithm. Unlike the squaring framework, which makes progress by improving the condition number, sparsified-Cholesky frameworks operate by eliminating rows and columns like Gaussian elimination. Such elimination steps can be directly interpreted as deleting a vertex from the graph and adding a weighted clique. 

Some algorithms eliminate one vertex at a time \cite{ks16, cohen2018solving}, but \cite{kyng2015sparsifiedconference} and \cite{peng2021sparsifiedconference} eliminate a large set of $\Omega(n)$ vertices together. To do so, a linear sized $\rho$-\emph{row-column-diagonally-dominant} ($\rho$-RCDD) subset $V'$ of the vertices is chosen for some constant $\rho$. A set of vertices is $\rho$-RCDD, if for each vertex a $\rho$-fraction of the weighted in-edges come from $V \setminus V'$ and a $\rho$-fraction of the weighted out-edges go to $V \setminus V'$\footnote{Notice that \cite{kyng2015sparsified} is concerned with the undirected case.}. Then the vertices belonging to this set can be eliminated using $O(\log \log n)$ sparsified squaring operations. While randomized algorithms using this paradigm can afford to globally sparsify after each squaring step, we have to allow for some build up of the edge count. Namely, we wait for a subpolynomial number of elimination rounds, and then globally sparsify and recurse. As previously, global sparsifications correspond to branching points when applying the inverse. We give a more detailed description in \Cref{sec:spars_cholesky}. 

\subsection{Reduction to the Eulerian Setting with bounded Condition Number}

Previous work \cite{cohen16faster, cohen2016almostlineartime} reduced solving linear equations in directed Laplacians $\LL = \DD - \AA^T$ to solving $\log^{O(1)}(n\kappa^{-1}\epsilon^{-1})$ systems involving Eulerian Laplacians with polynomially bounded condition number and edge weights to constant accuracy, where $\kappa$ is an upper bound on the maximum of $\kappa(\DD)$ and $\kappa(\LL)$ (See Appendix D and F of \cite{cohen2016almostlineartime} and Sections 5, 7.1 and 7.3 of \cite{cohen16faster}). They use that edge weights are polynomially bounded in the proof of Lemma C.3 in Appendix C  of \cite{cohen2016almostlineartime}. Different reductions to the Eulerian case were presented by Ahmadinejad-Jambulapati-Saberi-Sidford \cite{doi:10.1137/1.9781611975482.85} and in the thesis of Peebles \cite{Peebles19:thesis}.

\section{Preliminaries}

\subsection{Linear Algebra}

\paragraph{Matrices.} We denote matrices as bold upper case letters $\AA$. For a matrix $\AA \in \mathbb{R}^{n \times n}$, we let $\nnz(\AA)$ denote its number of non-zero entries and for $X \subseteq [n]$,  $Y \subseteq [n]$ we let $\AA(X, Y) = \AA_{XY}$ denote the $|X| \times |Y|$ submatrix containing the entries with  index in $X \times Y$. If $X = Y$, we write $\AA[X]$ as a shorthand for $\AA(X, X)$. When selecting submatrices, we let $l:u$ denote the set $\{l, l + 1, \dots , u\}$ for $u \geq l$ and $:$ the set of all columns/rows, e.g. $\AA(1:3, :)$ denotes the submatrix of $\AA$ consisting of the first $3$ rows of $\AA$. Further, we let $\AA^+$ denote the Moore-Penrose-Pseudoinverse of matrix $\AA$. Finally, $\II$ denotes the identity matrix. Sometimes we denote its dimension with $\II_n$.

\paragraph{Vectors. } We denote vectors as bold lower case letters $\vv$. Further, we let $\vecone$ denote the all ones vector, $\veczero$ the zero vector and $\ee_i$ denotes the $i$-th vector of the standard basis. Sometimes we indicate the dimension with $\vecone_n$ and $\veczero_n$.

\paragraph{The Loewner-order. } For symmetric matrices $\AA$ and $\BB$, we let $\AA \preceq \BB$ iff for all vectors $\xx$: $\xx^T \AA \xx \leq \xx^T \BB \xx$. We define $\prec, \succeq$ and $\succ$ analogously. If a symmetric matrix $\AA$ satisfies $\veczero \preceq \AA$ we call it \emph{PSD}. For a \emph{PSD} matrix $\AA$, we let $\AA^{1/2}$ denote the unique matrix so that $\AA^{1/2} \AA^{1/2} = \AA$.

\paragraph{Norms. } For every vector $\xx$, we let $\norm{\xx}_{\HH} := \sqrt{\xx^T \HH \xx}$ for a \emph{PSD} matrix $\HH$. We let $\norm{\MM}_{\HH \rightarrow \HH} := \max_{\xx \neq 0} \frac{\norm{\MM \xx}_{\HH}}{\norm{\xx}_{\HH}}$ for a PSD matrix $\HH$. Notice that $\norm{\MM}_{\HH \rightarrow \HH} = \norm{\HH^{1/2} \MM \HH^{+/2}}_2$. Further, we let $\norm{\MM}_1$ and $\norm{\MM}_{\infty}$  denote the maximum $\ell_1$ norm of a column and row of $\MM$ respectively.

\paragraph{Condition Number. } For a matrix $\AA \in \R^{n \times n}$ we let $\kappa(\AA) := \norm{\AA}_2 \norm{\AA^+}_2$ denote its condition number. Further, for PSD matrices $\AA$ and $\BB$ with the same kernel we let $\kappa(\AA, \BB) := \kappa(\AA^{+/2} \BB \AA^{+/2})$. 

\paragraph{Misc. } We let $\Exp(x) := e^x$. In this paper $\tilde{O}(\cdot)$ suppresses poly-logarithmic factors in $n$. For a PSD matrix $\AA$, we let $\lambda_*(\AA)$ denote its smallest nonzero eigenvalue.  

\subsection{Graphs}

\paragraph{General Notation. } We let $G = (V, E, \omega)$ denote an undirected graph where $\omega(e) = \omega(u, v)$ denotes the weight of edge $e = (u, v)$. Further, we let $\dir{G} = (V, E, \omega)$, where the edge weight $\omega(e) = \omega(u, v)$ of edge $e$ now depends on its direction. Sometimes we omit $\omega$ for unit weight (aka unweighted) graphs. When sometimes also write $V(\dir{G}), E(\dir{G})$ and $\omega_{\dir{G}}$ to avoid ambiguity. We let $\omega^{\max}$ and $\omega^{\min}$ denote the maximum and minimum edge weight respectively. 

\paragraph{Undirectification. } For a directed graph $\dir{G} = (V, E, \omega)$, we let $\mathcal{U}(\dir{G}) = (V, E', \omega)$, with $\{u, v \} \in E'$ iff $(u, v) \in E$ or $(v, u) \in E$ and $\omega{\{u, v \}} = \frac{1}{2}(\omega(u, v) + \omega(u, v))$, denote its undirectification (where we use the convention $\omega(u, v) = 0$ for $(u,v) \notin E$).

\paragraph{Induced Subgraphs. } For $G = (V, E, \omega)$ and $X \subseteq V$ we let $G[X]$ denote the induced subgraph on $X$. For a directed graph $\dir{G}$ we define $\dir{G}[X]$ analogously.

\subsection{Graph Laplacians}

\paragraph{General Notation.} For a undirected graph $G = (V, E, \omega)$ we denote its (graph) Laplacian as $\LL_G = \DD_G - \AA_G$ where $\DD_G$ is the diagonal matrices containing the degrees and $\AA_G(i,j) = \AA_G(j, i) = \omega((i, j))$. Generalizing this notation to directed graphs $\dir{G} = (V, E, \omega)$, we let $\LL_{\dir{G}} = \DD_{\dir{G}} - \AA^T_{\dir{G}}$ where the adjacency matrix is given by $\AA_{\dir{G}}(i, j) = \omega((i, j))$ and the diagonal matrix $\DD_{\dir{G}}(i, i) = \sum_j \AA_{\dir{G}}(i, j)$ contains the out degrees. Naturally $\vecone^T \LL_{\dir{G}} = \veczero$. For an undirected graph $G$, we let $\deg_G(v)$ be the (weighted) degree of vertex $v$. For directed graphs we let $\deg^-_G(v)$ and $\deg^+_G(v)$ denote in- and out-degree respectively.

\paragraph{Eulerian Laplacians. } If $\LL_{\dir{G}} \vecone = \veczero$ we call a Laplacian Eulerian. This correspond to the underlying directed graph $\dir{G}$ being Eulerian, i.e. having equal in- and out-degree for each vertex. 

\paragraph{Symmetrization. } For any matrix $\AA$ we denote $\UU_{\AA} = \UU(\AA) := \frac{1}{2}(\AA + \AA^T)$. Notice that for an Eulerian Laplacian $\LL_{\dir{G}}$ we have $\UU_{\LL_{\dir{G}}} = \LL_{\mathcal{U}(\dir{G})}$. This is a crucial fact exploited by all algorithms for directed Laplacians including ours. Symmetric Laplacians are \emph{PSD}.

\paragraph{Induced subgraphs and submatrices. } The reader should note that for a Laplacian $\LL_{\dir{G}}$ the matrices $\LL_{\dir{G}[X]}$ and $\LL_{\dir{G}}[X]$ are not equivalent unless there is no edge from $X$ to $V \setminus X$ or vice versa. Specifically, while the off-diagonal entries are equal, we have $\DD_{\dir{G}[X]}(i,i) \leq \DD_{\dir{G}}[X](i, i)$. 

\subsection{Directed Graph Approximation}

We will use the notions of approximation for directed graphs introduced in \cite{cohen2016almostlineartime}. 

\begin{definition}[Asymmetric Matrix Approximation, Definition 3.1 in \cite{cohen2016almostlineartime}]
A (possibly asymmetric) matrix $\AAtil$ is said to be an $\epsilon$-matrix-approximation of $\AA$ if
\begin{enumerate}
    \item $\UU_{\AA}$ is a symmetric PSD matrix, with $\ker(\UU_{\AA}) \subseteq \ker(\AAtil - \AA) \cap \ker((\AAtil - \AA)^T)$.
    \item $\norm{\UU_{\AA}^{+/2}(\AAtil - \AA)\UU_{\AA}^{+/2}}_2 \leq \epsilon$ 
\end{enumerate}
\label{def:mat_approx}
\end{definition}

\begin{remark}
Notice that \Cref{def:mat_approx} is not symmetric. 
\end{remark}

\subsection{Expanders}

Expander graphs, or expanders for short, play a central role in both the workings and analysis of our algorithms. Notice that while our algorithm operates on directed graphs, we only use expanders in the context of undirected graphs. Whenever we use the term expander, we mean expanders with respect to \emph{conductance}. We follow the notational conventions of \cite{chuzhoy2020deterministic}.

\begin{definition}[Conductance]
For a weighted but undirected graph $G= (V, E, \omega)$, given a set $\emptyset \subset S \subset V$ we let $\delta_G(S) := \sum_{(u,v) \in E: u \in S, v \notin S} \omega(u,v)$ and $\vol_G(S) := \sum_{v \in S} \sum_{u \in V} \omega(u, v)$. Then we define the conductance
\begin{align*}
    \Phi_G(S) = \frac{\delta_G(S)}{\min\{\vol_G(S), \vol_G(V \setminus S)\}}.
\end{align*}
\end{definition}

\begin{definition}[Expander]
We call a graph $G = (V, E, \omega)$ a $\Phi$-expander (with respect to conductance) if $\min_{S: \emptyset \subset S \subset V} \Phi_G(S) \geq \Phi$. We further let $\Phi_G = \min_{S: \emptyset \subset S \subset V} \Phi_G(S)$.
\end{definition}

Given a vector $\dd \in \mathbb{R}^n_{>0}$ we let $G(\dd)$ denote the weighted and undirected graph on $n$ vertices with $\omega(i,j) = \frac{\dd(i) \cdot \dd(j)}{\norm{\dd}_1}$. Note that $\LL_{G(\dd)} = \diag(\dd) - \frac{\dd \dd^T}{\norm{d}_1}$.

\begin{fact}[Observation 6.6 in \cite{chuzhoy2020deterministic}]
$G(\dd)$ is a $1/2$-expander.
\label{fact:exander}
\end{fact}

We will frequently use that Laplacians $\LL_H = \DD_H - \AA_H^T$ of good expanders $H$ are well approximated by $\LL_{G(\diag(\DD_H))}$. To establish this, we need the following lemma. 

\begin{lemma}[Lemma 6.7 in \cite{chuzhoy2020deterministic}]
Let $G$ and $H$ be two undirected weighted n-vertex graphs on the same vertex set, such that $\DD_G = \DD_H$. Assume further that $\Phi(G), \Phi(H) \geq \Phi$ for some threshold $\Phi$. Then for any real vector $\xx \in \R^n$: $ \frac{\Phi^2}{4} \xx^T \LL_G \xx \leq \xx^T \LL_H \xx^T \leq \frac{4}{\Phi^2} \xx^T \LL_G \xx$.
\label{lem:expand_approx}
\end{lemma}

We state a direct corollary in a form that turns out to be convenient for our arguments. 

\begin{corollary}[Non Uniform Degree Bound]
Given a $\Phi$-expander $G$ with Laplacian $\LL_G = \DD_G - \AA^T_G$ we have
\begin{align*}
    \frac{\Phi^2}{4} (\DD_G - \frac{\dd_G \dd_G^T}{\norm{\dd_G}_1}) \preceq \LL_G \preceq \frac{4}{\Phi^2} (\DD_G - \frac{\dd_G \dd_G^T}{\norm{\dd_G}_1})
\end{align*}
for $\dd_G = \diag(\DD_G)$ and $\Phi \leq 1/2$.
\label{col:approximation_via_expander}
\end{corollary}
\begin{proof}
Directly follows from \Cref{fact:exander} and \Cref{lem:expand_approx}.
\end{proof}

Next we state the definition of the expander decomposition for unweighted and undirected graphs in the notation of \cite{chuzhoy2020deterministic}. 

\begin{definition}[See Section 6 of \cite{chuzhoy2020deterministic}, Proposed by \cite{kvv2004clusterings, gr1999bipartiteness}]
A $(\epsilon, \Phi)$-expander decomposition of a undirected and unweighted graph $G = (V, E)$ is a partition $\mathcal{P} = \{V_1, ..., V_k\}$ of the vertex set $V$ such that for all $i \in [k]$ the conductance of $G[V_i]$ is at least $\Phi$ and $\sum_{i = 1}^k \delta_G(V_i) \leq \epsilon \vol_G(V)$.
\end{definition}

The next theorem shows that expander decompositions can be computed in almost linear time. 

\begin{theorem}[Corollary 7.7 in \cite{chuzhoy2020deterministic}]
There is a deterministic algorithm that, given an undirected and unweighted graph $G = (V, E)$ with $m$ edges and parameters $\epsilon \in (0,1]$ and $1 \leq r \leq O(\log m)$, computes a $(\epsilon, \Phi)$-expander decomposition of $G$ with $\Phi \geq \Omega(\epsilon/(\log m)^{O(r^2)})$ in time $O\left(m^{1 + O(1/r) + o(1)}\cdot (\log m)^{O(r^2)}\right)$. 
\label{thm:exp_dec_det}
\end{theorem}

Finally, we adapt the previous theorem to interface more conveniently with our statements.

\begin{corollary}[Expander Decomposition]
There is a deterministic algorithm \expdecomp($G, \gamma$) that, given an undirected and unweighted graph $G = (V, E)$ with $m$ edges and a constant $\gamma \in (0,1)$, computes a $(1/2, \frac{1}{\Exp((\log n)^{\gamma})})$-expander decomposition in time $m^{1 + o(1)}$.
\label{col:exp_dec_det}
\end{corollary}
\begin{proof}
Follows directly from \Cref{thm:exp_dec_det}. 
\end{proof}

\subsection{Preconditioned Richardson}

The next lemma analyses preconditioned Richardson (\Cref{alg:richardson}) for asymmetric matrices. 

\begin{lemma}[Preconditioned Richardson, Lemma 4.2 in \cite{cohen2016almostlineartime}]
Let $\bb \in \mathbb{R}^n$ and $\MM, \ZZ, \UU \in \mathbb{R}^{n \times n}$ such that $\UU$ is symmetric positive definite, $\ker(\UU) \subseteq \ker(\MM) = \ker(\MM^T) = \ker(\ZZ) = \ker(\ZZ^T)$, and $\bb \in \im(\MM)$. Then $N$ iterations of preconditioned Richardson with step size $\eta > 0$, result in a vector $\xx_N = \richardson(\MM, \ZZ, \bb, \eta, N)$ so that 
\begin{align*}
    \norm{\xx_N - \MM^+ \bb}_{\UU} \leq \norm{\II_{\im(\MM)} - \eta \ZZ \MM}^N_{\UU \rightarrow \UU}\norm{\MM^+ \bb}_{\UU}.
\end{align*}
Furthermore preconditioned Richardson implements a linear operator, in the sense that $\xx_N = \ZZ_N \bb$ for some matrix $\ZZ_N$ only depending on $\ZZ, \MM, \eta$ and $N$.
\label{lem:richardson}
\end{lemma}

The previous lemma leads to the notion of an approximate pseudoinverse by measuring the suitability of a matrix as a preconditioner.

\begin{definition}[Approximate Pseudoinverse, Definition 4.3 in \cite{cohen2016almostlineartime}]
Matrix $\ZZ$ is an $\epsilon$-approximate-pseudoinverse of matrix $\MM$ with respect to a \emph{PSD} matrix $\UU$, if $\ker(\UU) \subseteq \ker(\MM) = \ker(\MM)^T = \ker(\ZZ) = \ker(\ZZ^T)$, and
\begin{align*}
    \norm{\II_{\im(\MM)} - \ZZ \MM}_{\UU \rightarrow \UU} \leq \epsilon. 
\end{align*}
\end{definition}

\begin{algorithm}
\caption{\textsc{PreconRichardson}($\MM, \ZZ, \bb, \eta, N$)}
\label{alg:richardson}
$\xx_0 = \veczero$ \\
\For{i = 0, ..., N - 1}{
    $\xx_{i + 1} = \xx_i + \eta \ZZ(\bb - \MM \xx_i)$ \\
}
\Return{$\xx_N$}
\end{algorithm}

Finally, we state three more lemmas that are crucial for arguing about approximate pseudoinverses. The first two are often applied consecutively to upper bound $\norm{\II_{\im(\MM)} - \ZZ \MM}_{\UU_{\ZZ} \rightarrow \UU_{\ZZ}}$ with a variational form. 

\begin{lemma}[Part of Lemma B.9 in \cite{cohen2016almostlineartime}]
If $\LL$ is a matrix with $\ker(\LL) = \ker(\LL^T) = \ker(\UU_{\LL})$, and $\UU_{\LL}$ is positive semidefinite, then for any matrix $\AA$ with the same left and right kernels as $\LL$ we have
\begin{align*}
    \norm{\AA}_{\UU_{\LL} \rightarrow \UU_{\LL}} \leq  \norm{\UU_{\LL}^{+/2} \LL \AA \UU_{\LL}^{+/2}}_2
\end{align*}
\label{lem:b9}
\end{lemma}

\begin{lemma}[Part of Lemma B.2 in \cite{cohen2016almostlineartime}]
For all $\AA \in \mathbb{R}^{n \times n}$ and symmetric \emph{PSD} $\MM, \NN \in \mathbb{R}^{n \times n}$ such that $\ker(\MM) \subseteq \ker(\AA^T)$ and $\ker(\NN) \subseteq \ker(\AA)$ we have
\begin{align*}
    \norm{\MM^{+/2} \AA \NN^{+/2}}_2 = 2 \max_{\xx, \yy \neq 0} \frac{\xx^T \AA \yy}{\xx^T \MM \xx + \yy^T \NN \yy}
\end{align*}
where we define $0/0$ to be $0$.
\label{lem:b2}
\end{lemma}

\begin{lemma}[Part of Lemma B.4 in \cite{cohen2016almostlineartime}]
For a \emph{PSD} diagonal matrix $\DD$ and any matrix $\MM \in \mathbb{R}^{n \times n}$
\begin{align*}
    \norm{\DD^{-1/2} \MM \DD^{-1/2}}_2 \leq \max\{ \norm{\DD^{-1} \MM}_{\infty} , \norm{\DD^{-1} \MM^T}_{\infty} \} = \max\{ \norm{\MM^T \DD^{-1}}_{1} , \norm{\MM \DD^{-1}}_{1} \}.
\end{align*}
\label{lem:b4}
\end{lemma}

\section{Global Sparsification for Directed Laplacians}
\label{sec:globspars}
In this section we describe our low accuracy global sparsification routine. This constitutes the backbone of our algorithm. We first formally define the $\beta$-partial-symmetrization of an Eulerian graph $\dir{G}$. See \Cref{fig:almost_symmetrization} for an illustration. 

\begin{definition}
For an Eulerian directed graph $\dir{G}$ we call $\mathcal{U}^{(\beta)}(\dir{G}) := \beta \cdot \mathcal{U}(\dir{G}) +  \dir{G}$ its $\beta$-partial-symmetrization.
\end{definition}
\begin{remark}
$\LL_{\mathcal{U}^{(\beta)}(\dir{G})} = \beta \UU_{\LL_{\dir{G}}} + \LL_{\dir{G}}$. 
\label{rem:partial_sym_eul}
\end{remark}

There are three steps in our sparsification procedure. 

\begin{enumerate}
    \item The first step relies on what may be the most crucial observation. Given an Eulerian directed graph $\dir{G}= \dir{G}_0$, we let $\dir{G}_1 = \mathcal{U}^{(\beta)}(\dir{G}) = \beta \cdot \mathcal{U} (\dir{G}) + \dir{G}$ be the graph obtained from $\dir{G}_0$ by $\beta$-partial symmetrization. Then, surprisingly, $\LL_{\dir{G}_1}$ can be used as a preconditioner for solving linear equations in $\LL_{\dir{G}_0}$ in time $O(\beta)$ using Richardson. We call $\mathcal{U}^{(\beta)}(\dir{G})$ the $\beta$-partial-symmetrization of $\dir{G}$.
    \item A partial-symmetrization is naturally interpreted as the sum of an undirected graph $\beta \cdot \mathcal{U}(\dir{G})$ and a directed graph $\dir{G}$. We expander decompose the undirected graph $\beta \cdot \mathcal{U}(\dir{G})$ into parts $V_1, V_2, ..., V_k$. Then a simple greedy patching scheme can be used to sparsify the induced subgraphs $\dir{G}[V_i]$ leveraging the expander structure for error control. In our actual algorithm we additionally bucket by edge weight and layer expander decompositions. The former allows us to treat the graph as unweighted and the latter ensures that every edge is in an expander after $O(\log n)$ layers. The Laplacian of the obtained graph $\dir{G}_2 = \beta \cdot \mathcal{U} (\dir{G}) + \dir{R}$ can then be used as a preconditioner for $\LL_{\dir{G}_1}$
    \item Lastly, the undirected graph $\beta \cdot \mathcal{U} (\dir{G})$ can be sparsified via previously known deterministic algorithms presented in \cite{chuzhoy2020deterministic}.  We obtain $\dir{G}_3 = \frac{\beta}{\eta} \tilde{G} + \dir{R}$ which in turn is a preconditioner for $\dir{G}_2$. 
\end{enumerate}
 
\begin{figure}[h]
    \centering
    \includegraphics[width=13cm]{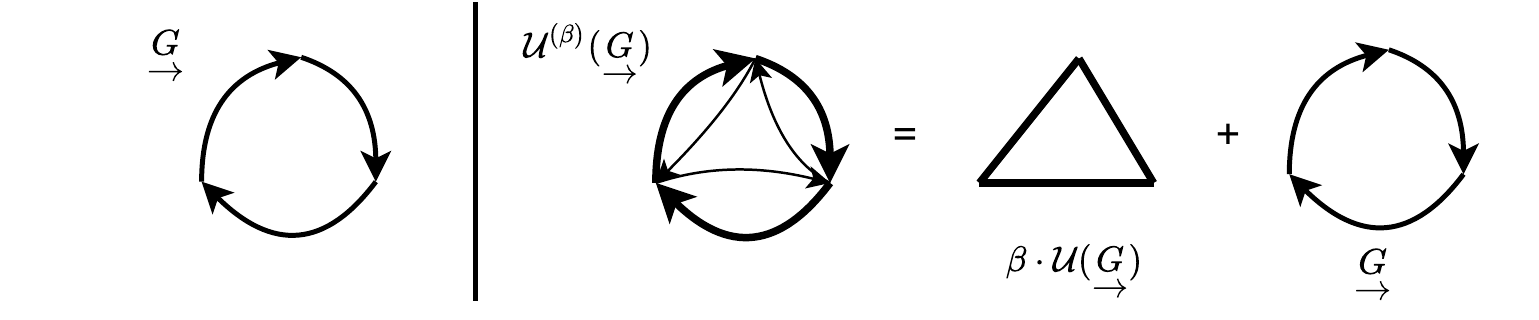}
    \caption{This drawing illustrates the concept of an $\beta$-partial-symmetrization. We use boldness to roughly indicate edge weights. On the left hand side an Eulerian graph $\dir{G}$ is depicted. On the right hand side its $\beta$-partial-symmetrization is drawn, alongside its natural split into an undirected and directed part.}
    \label{fig:almost_symmetrization}
\end{figure}

We define pseudoinverse sparsification quadruples. Constructing these is at the core of our global sparsification routine. 

\begin{definition}(Sparsification quadruple)
We call strongly connected $n$-vertex Eulerian graphs $\dir{G}_0, \dir{G}_1, \dir{G}_2, \dir{G}_3$ a $(\gamma, \beta, \eta)$-quadruple for some constant $\gamma \in (0,1)$ if 
\begin{enumerate}
    \item $\LL_{\dir{G}_i}^+$ is a $\left(1 - \frac{1}{\Exp(O((\log n)^{\gamma}))} \right)$-approximate pseudoinverse of $\LL_{\dir{G}_{i - 1}}$ with respect to $\UU_{\LL_{\dir{G}_{i}}}$ for $i = 1, 2, 3$. 
    \item $\frac{1}{\Exp(O((\log n)^{\gamma}))} \UU_{\LL_{\dir{G}_{i-1}}} \preceq \UU_{\LL_{\dir{G}_{i}}} \preceq \Exp(O((\log n)^{\gamma})) \UU_{\LL_{\dir{G}_{i- 1}}}$ for $i \in \{1, 2, 3\}$. 
    \item $|E(\dir{G}_3)| = \tilde{O}(n)$. $|E(\dir{G}_i)| \leq 2 |E(\dir{G}_0)| + \tilde{O}(n)$ for $i = 1,2$.
    \item For all vertices $v$: $\deg_{\dir{G}_0}^{+}(v) = \deg_{\dir{G}_0}^{-}(v) = (1 + \beta) \deg_{\dir{G}_i}^{+}(v) = (1 + \beta) \deg_{\dir{G}_i}^{-}(v)$ for $i = 1, 2$ and $\deg_{\dir{G}_0}^{+}(v) = (1 + \frac{\beta}{\eta})  \deg_{\dir{G}_3}^{+}(v) = (1 + \frac{\beta}{\eta}) \deg_{\dir{G}_3}^{-}(v)$. 
\end{enumerate}
\label{def:quadruple}
\end{definition}

We then state the main lemma of this section. It shows that it is possible to construct a ($\gamma, \beta, \eta$)-quadruple in almost linear time. 

\begin{lemma}[Global Sparsification]
For every $m$-edge, strongly connected Eulerian graph $\dir{G} = \dir{G}_0$ and a constant $\gamma \in (0,1)$ the routine $\dir{G}_1, \dir{G}_2, \dir{G}_3 = \sparsify(\dir{G}, \gamma)$ yields a $(\gamma, \beta = \Exp(O((\log n)^{\gamma})) , \eta = \Exp(-3(\log n)^{\gamma}))$-quadruple $\dir{G}_0, \dir{G}_1, \dir{G}_2, \dir{G}_3$. The runtime is $m^{1 + o(1)}$. 
\label{lem:globalspars}
\end{lemma}

\begin{remark}
While in the description of our algorithm, $\beta$ scales linearly in $O\left(\log \left( \frac{\omega^{max}_{\dir{G}}}{\omega^{min}_{\dir{G}}}\right)\right)$, we assume throughout the paper that $\beta = \tilde{O}(1) \cdot \Exp(2(\log n)^{\gamma})$ is fixed to a global upper bound as $\log \left(\frac{\omega^{max}_{\dir{G}}}{\omega^{min}_{\dir{G}}}\right) = \tilde{O}(1)$ for all graphs ${\dir{G}}$ we work with. This avoids clutter in the analysis. 
\end{remark}

\begin{algorithm}
\caption{\sparsify($\dir{G}, \gamma$)}
\label{alg:global_spars}
$\beta = L \cdot \Exp(2 \cdot (\log n)^{\gamma})$ for $L = 128 \cdot 20 \cdot P \cdot \log n$ and $P = \ceil{\log\left(\frac{\omega^{\max}_{\dir{G}}}{\omega^{\min}_{\dir{G}}}\right)}$. \\
$\eta = \Exp(-3\cdot (\log n)^\gamma)$ \\
$\dir{G}_1 = \beta \mathcal{U}(\dir{G}) + \dir{G}$ \tcp*{Note that $\dir{G}_1 = \mathcal{U}^{(\beta)}(\dir{G})$.}
$\dir{R} = \sparsedir(\dir{G} , \gamma)$ \\
$\dir{G}_2 =  \beta \mathcal{U}(\dir{G}) + \dir{R}$ \\
$\tilde{G} = \specspardeg(\mathcal{U}(\dir{G}), \gamma)$ \\
$\dir{G}_3 = \frac{\beta}{\eta} \tilde{G} + \dir{R}$\\
\Return{$\dir{G}_1, \dir{G}_2, \dir{G}_3$}
\end{algorithm}

\subsection{Preconditioning with the Partial-Symmetrization}

Our next lemma shows that $\LL_{\mathcal{U}^{(\beta)}(\dir{G})}$ is a good preconditioner in terms of \Cref{lem:richardson}.

\begin{lemma}
For every Eulerian Laplacian $\LL_{\dir{G}_0}$ the matrix $\LL_{\dir{G}_1}^+$ is an $(1 - \frac{1}{1 + \beta})$-approximate pseudoinverse of $\LL_{\dir{G}_0}$ with respect to $\UU_{\LL_{\dir{G}_1}}$ if $\dir{G}_1 = \mathcal{U}^{(\beta)}(\dir{G}_0)$.
\label{lem:partial_sym_approx}
\end{lemma} 
\begin{proof}
First recall that $\LL_{\dir{G}_1}^+ =\LL_{\mathcal{U}^{(\beta)}(\dir{G})}^+$. We have 
\begin{align*}
\norm{\II_{\im(\LL_{\dir{G}})} - \LL_{\mathcal{U}^{(\beta)}(\dir{G})}^+ \LL_{\dir{G}}}_{\UU_{\LL_{\mathcal{U}^{(\beta)}(\dir{G})}} \rightarrow \UU_{\LL_{\mathcal{U}^{(\beta)}(\dir{G})}}} 
    \leq \norm{\UU_{\LL_{\mathcal{U}^{(\beta)}(\dir{G})}}^{+/2} (\LL_{\mathcal{U}^{(\beta)}(\dir{G})} - \LL_{\dir{G}}) \UU_{\LL_{\mathcal{U}^{(\beta)}(\dir{G})}}^{+/2}}_2
\end{align*}
by \Cref{lem:b9}. With \Cref{rem:partial_sym_eul} we conclude
\begin{align*}
    \norm{\UU_{\LL_{\mathcal{U}^{(\beta)}(\dir{G})}}^{+/2} (\LL_{\mathcal{U}^{(\beta)}(\dir{G})} - \LL_{\dir{G}}) \UU_{\LL_{\mathcal{U}^{(\beta)}(\dir{G})}}^{+/2}}_2 &= \beta  \norm{\UU_{\LL_{\mathcal{U}^{(\beta)}(\dir{G})}}^{+/2} \UU_{\LL_{\dir{G}}}  \UU_{\LL_{\mathcal{U}^{(\beta)}(\dir{G})}}^{+/2}}_2 \\
    &= \frac{\beta}{1 + \beta} \norm{\UU_{\LL_{\dir{G}}}^{+/2} \UU_{\LL_{\dir{G}}}  \UU_{\LL_{\dir{G}}}^{+/2}}_2 \\
    &= \frac{\beta}{1 + \beta} = 1 - \frac{1}{1 + \beta}
\end{align*}
where we use that $\UU_{\LL_{\mathcal{U}^{(\beta)}(\dir{G})}} = \LL_{\mathcal{U}(\mathcal{U}^{(\beta)}(\dir{G}))} = (1 + \beta)\UU_{\LL_{\dir{G}}}$. The lemma follows from chaining the calculations.
\end{proof}

\subsection{Sparsifying the Directed Part}

Given $\LL_{\dir{G}_1} = \beta \UU_{\LL_{\dir{G}}} + \LL_{\dir{G}}$, we aim to obtain a sparse directed graph $\dir{R}$ with the same in- and  out-degrees as $\dir{G}$ so that the directed Laplacian $\LL_{\dir{G}_2} = \beta \UU_{\LL_{\dir{G}}} + \LL_{\dir{R}}$ preconditions $\LL_{\dir{G}}$. Our strategy closely follows common strategies for sparsifying undirected graphs via expander decompositions. First we get rid of most of the weighted structure by bucketing by edge weight. We obtain $\tilde{O}(1)$ graphs $\dir{G}^{(i)}$ with close to uniform edge weight such that $\sum_i \dir{G}^{(i)} = \dir{G}$. 

We let $H^{(i)}$ denote the unweighted and undirected graph with the same edges as $\dir{G}^{(i)}$. Then we layer $j = 1, ...,O(\log n)$ undirected and unweighted expander decompositions on this graph, where each of them peels of at least $1/2$ of the remaining edges $E_r^{(i,j)}$. This procedure computes $O(\log n)$ partitions $V^{(i, j)}_1, ..., V^{(i, j)}_{k(i,j)}$ of the vertex set, such that for each component $V^{(i, j)}_p$ the graph $H^{(i)}[V^{(i, j)}_p]$ is an expander. In the $j$-th layer, we put the remaining edges of the directed graph $\dir{G}^{(i)}$ that do not go across sets in the partition $V^{(i, j)}_1, ..., V^{(i, j)}_{k(i,j)}$ into the graph $\dir{G}^{(i,j)}$ and remove them from the set of remaining edges. 

The expander structure allows us to sparsify $\dir{G}^{(i,j)}$ via a greedy patching scheme obtaining $\tilde{\dir{G}}^{(i,j)}$. Finally, we sum up across layers and buckets and obtain $\dir{R} = \sum_{i, j} \tilde{\dir{G}}^{(i,j)}$. We leverage the robustness introduced by partial symmetrization to bound the error. See Algorithm \ref{alg:sparsepatchglobal} for detailed pseudocode. We first state the main lemma of this subsection, which analyses this algorithm. 

\begin{lemma}
Let $\dir{R} = \sparsedir(\dir{G}, \gamma)$ for $\gamma \in (0,1)$ constant. Then $\LL^+_{\dir{G}_2} = (\beta \UU_{\LL_{\dir{G}}} + \LL_{\dir{R}})^+$ is a $1/2$-approximate pseudoinverse of $\LL_{\dir{G}_1} = \LL_{\mathcal{U}^{\beta}(\dir{G})}$ with respect to $\UU_{\LL_{\dir{G}_2}}$ for $\beta = \tilde{O}(1) \cdot \Exp(2 \cdot (\log n)^{\gamma})$. Further, the graph $\dir{R}$ has $\tilde{O}(n)$ edges and the same in- and out-degrees as $\dir{G}$. 
\label{lem:global_patching}
\end{lemma}

\begin{algorithm}
\caption{\sparsedir($\dir{G} , \gamma$) and subroutines \textsc{Sparsify}() and \textsc{Patch}()}
\label{alg:sparsepatchglobal}
\SetKwFunction{algo}{\sparsedir}\SetKwFunction{procsp}{\textsc{Sparsify}}\SetKwFunction{procpa}{\textsc{Patch}} 
\SetKwProg{myalg}{Algorithm}{}{}
\SetKwProg{myproc}{Procedure}{}{}
\myalg{\algo{$\dir{G}, \gamma$}}{
$P = \ceil{\log\left(\frac{\omega^{\max}_{\dir{G}}}{\omega^{\min}_{\dir{G}}}\right)}$ \\
\For{$i = 1, ..., P$}{
    $E^{(i)} = \{e \in E(\dir{G}): \omega_{\dir{G}}^{min} \cdot 2^{i-1} \leq \omega_{\dir{G}}(e)  < \omega_{\dir{G}}^{min} \cdot 2^{i} \}$ \\
    $\tilde{\dir{G}}^{(i)} = \textsc{Sparsify}(\dir{G}^{(i)} = (V(\dir{G}), E^{(i)}, \omega_{\dir{G}}))$\\
}
\Return{$\dir{R} = \sum_{i = 1}^P \tilde{\dir{G}}^{(i)}$}
}

\myproc{\procsp{$\dir{G}^{(i)}$}}{
Let $H^{(i)}$ denote the unweighted and undirected graph with the same edges as $\mathcal{U}(\dir{G}^{(i)})$ \\
\For{$j = 1, ..., 10\log n$}{
    $E_r^{(i,j)} = E({H}^{(i)}) - \bigcup_{l = 1}^{j - 1} E({H}^{(i,j)})$; $\dir{E}_r^{(i,j)} = E(\dir{G}^{(i)}) - \bigcup_{l = 1}^{j - 1} E(\dir{G}^{(i,j)})$ \\
    $V^{(i,j)}_1, ..., V_{k(i,j)}^{(i,j)} = \expdecomp((\VV({H}^{(i)}), E_r^{(i,j)}), \gamma)$ \\
    $E({H}^{(i,j)}) = \bigcup_{p = 1}^{k(i,j)} \{(u,v) \in E_r^{(i,j)}: u \in V^{(i,j)}_p  \land v \in V^{(i,j)}_p \}$ \\
    $E(\dir{G}^{(i,j)}) = \bigcup_{p = 1}^{k(i,j)} \{(u,v) \in \dir{E}_r^{(i,j)}: u \in V^{(i,j)}_p  \land v \in V^{(i,j)}_p \}$\\ $H^{(i,j)} = (V({H}^{(i)}), E({H}^{(i,j)}))$; 
    $\dir{G}^{(i,j)} = (V(\dir{G}^{(i)}), E(\dir{G}^{(i,j)}), \omega_{\dir{G}^{(i)}})$ \\
    $\tilde{\dir{G}}^{(i,j)} = \sum_{p = 1}^{k(i,j)} \textsc{Patch}(\dir{G}^{(i,j)}[V^{(i,j)}_p])$
}
\Return{$\tilde{\dir{G}}^{(i)} = \sum_{j = 1}^{10 \log n} \tilde{\dir{G}}^{(i,j)}$}
}

\myproc{\procpa{$\dir{H}$}}{
Let $\aa,\bb \in \RR^n_{\geq 0}$ so that $\aa(v) = \deg^+_{\dir{H}}(v)$ and $\bb(v) = \deg^-_{\dir{H}}(v)$. \tcp{Note $\norm{\aa}_1 = \norm{\bb}_1$.} 
$E(\tilde{\dir{H}}) = \emptyset$; $\omega_{\tilde{\dir{H}}}(e) = 0$ for all $e$. \\
\While{$\norm{\aa} \neq 0$}{
Let $i$ and $j$ be arbitrary such that ${\aa}(i) > 0$ and ${\bb}(j) > 0$. \\
$w = \min\{\aa(i), \bb(j)\}$;  $\aa(i) = \aa(i) - w$; $\bb(j) = \bb(j) - w$\\
$E(\tilde{\dir{H}}) = E(\tilde{\dir{H}}) \cup \{(i, j)\}$; $\omega(i , j) = w$
}
\Return{$\tilde{\dir{H}}$}
}
\end{algorithm}

The proof of \Cref{lem:global_patching} relies on analysing the error incurred by sparsifying each individual expander decomposition, i.e. the cost of replacing $\dir{G}^{(i,j)}$ with $\tilde{\dir{G}}^{(i,j)}$. Then we conclude by just summing up the error. The next lemma carefully analyses the amount of error sparsifying such an expander decomposition creates. It crucially relies on the fact that the expander parts form a disjoint partition, and therefore the error does not scale in the number of expanders. 

\begin{lemma}
In the context of \textsc{Sparsify}() in \Cref{alg:sparsepatchglobal} we have
\begin{align*}
    \norm{\UU_{L_{\dir{G}}}^{+/2}(\dir{G}^{(i,j)} - \tilde{\dir{G}}^{(i,j)}) \UU_{L_{\dir{G}}}^{+/2}}_2 \leq 128 \cdot \Exp(2 \cdot (\log n)^{\gamma})
\end{align*}
for every edge weight bucket $i$ and expander decomposition layer $j$. 
\label{lem:expander_patching_error}
\end{lemma}
\begin{proof} $\dir{G}^{(i,j)}$ and $\tilde{\dir{G}}^{(i,j)}$ are directed graphs with the same in and out degrees since \textsc{Patch}() in \Cref{alg:sparsepatchglobal} preserves degrees exactly. Therefore $\vecone$ is in both the left and right kernel of $\dir{G}^{(i,j)} - \tilde{\dir{G}}^{(i,j)}$. We apply \Cref{lem:b2} twice and obtain 
\begin{align}
    \norm{\UU_{L_{\dir{G}}}^{+/2}(\LL_{\dir{G}^{(i,j)}} - \LL_{\tilde{\dir{G}}^{(i,j)}}) \UU_{L_{\dir{G}}}^{+/2}}_2 &= 2 \max_{\xx, \yy \neq \veczero} \frac{\xx^T(\LL_{\dir{G}^{(i,j)}} - \LL_{\tilde{\dir{G}}^{(i,j)}})\yy}{ \xx^T \UU_{L_{\dir{G}}} \xx +  \yy \UU_{L_{\dir{G}}} \yy^T} \nonumber \\
    &\stackrel{i)}{\leq}  2 \max_{\xx, \yy \neq \veczero} \frac{\xx^T(\LL_{\dir{G}^{(i,j)}} - \LL_{\tilde{\dir{G}}^{(i,j)}})\yy}{ \xx^T \LL_{\mathcal{U}(\dir{G}^{(i,j)})} \xx +  \yy  \LL_{\mathcal{U}(\dir{G}^{(i,j)})} \yy^T} \nonumber \\
    &\stackrel{ii)}{\leq} \frac{2}{\omega^{\min}_{\dir{G}} \cdot 2^{i - 2}} \max_{\xx, \yy \neq \veczero} \frac{\xx^T(\LL_{\dir{G}^{(i,j)}} - \LL_{\tilde{\dir{G}}^{(i,j)}})\yy}{ \xx^T \LL_{H^{(i,j)}} \xx +  \yy  \LL_{H^{(i,j)}} \yy^T} \nonumber \\
    &= \frac{1}{\omega^{\min}_{\dir{G}} \cdot 2^{i - 2}} \norm{\LL_{H^{(i,j)}}^{+/2}(\LL_{\dir{G}^{(i,j)}} - \LL_{\tilde{\dir{G}}^{(i,j)}})\LL_{H^{(i,j)}}^{+/2}}_2 \label{eq:sparsepatch1}
\end{align}
where i) follows since $\LL_{\mathcal{U}(\dir{G}^{(i,j)})} \preceq \UU_{L_{\dir{G}}}$ because $\mathcal{U}(\dir{G}^{(i,j)})$ is a subgraph of $\mathcal{U}(\dir{G})$ and ii) uses that all edge weights in $\mathcal{U}(\dir{G}^{i,j})$ are in $[\omega^{min}_{\dir{G}} \cdot 2^{i - 2}, \omega^{min}_{\dir{G}} \cdot 2^i]$. Next we can use that the expander parts $V^{(i,j)}_p$ are disjoint in both $\dir{G}^{(i,j)}$ and $\tilde{\dir{G}}^{(i,j)}$ to bound
\begin{align}
    \norm{\LL_{H^{(i,j)}}^{+/2}(\LL_{\dir{G}^{(i,j)}} - \LL_{\tilde{\dir{G}}^{(i,j)}})\LL_{H^{(i,j)}}^{+/2}}_2 &\leq  \max_{p} \norm{\LL_{H^{(i,j)}[V^{(i,j)}_p]}^{+/2}(\LL_{\dir{G}^{(i,j)}[V^{(i,j)}_p]} - \LL_{\tilde{\dir{G}}^{(i,j)}[V^{(i,j)}_p]})\LL_{H^{(i,j)}[V^{(i,j)}_p]}^{+/2}}_2
    \label{eq:sparsepatch2}
\end{align}
since the spectral norm of a block diagonal matrix is upper bounded by the maximum spectral norm of a block (See \Cref{fact:blockdiag}). Then, for every $p \in \{1, ... , k(i,j)\}$ we have 
\begin{align}
    \norm{\LL_{H^{(i,j)}[V^{(i,j)}_p]}^{+/2}(\LL_{\dir{G}^{(i,j)}[V^{(i,j)}_p]} - \LL_{\tilde{\dir{G}}^{(i,j)}[V^{(i,j)}_p]})\LL_{H^{(i,j)}[V^{(i,j)}_p]}^{+/2}}_2 \stackrel{i)}{=} 2\max_{\xx, \yy \neq \veczero}  \frac{\xx^T(\LL_{\dir{G}^{(i,j)}[V^{(i,j)}_p]} - \LL_{\tilde{\dir{G}}^{(i,j)}[V^{(i,j)}_p]})\yy}{ \xx^T \LL_{H^{(i,j)}[V^{(i,j)}_p]} \xx +  \yy  \LL_{H^{(i,j)}[V^{(i,j)}_p]} \yy^T} \nonumber \\
    \stackrel{ii)}{\leq} 8 \cdot 2^{2 (\log n)^{\gamma}} \max_{\xx, \yy \neq \veczero}  \frac{\xx^T(\LL_{\dir{G}^{(i,j)}[V^{(i,j)}_p]} - \LL_{\tilde{\dir{G}}^{(i,j)}[V^{(i,j)}_p]})\yy}{ \xx^T \LL_{G(\dd^{(i,j)}_p)} \xx +  \yy  \LL_{G(\dd^{(i,j)}_p)} \yy^T} 
    \label{eq:sparsepatch3}
\end{align} 
for $\dd^{(i,j)}_p$ being the degree vector of $H^{(i,j)}[V^{(i,j)}_p]$, where i) is by \Cref{lem:b2} and ii) is by the expansion of $H^{(i,j)}[V^{(i,j)}_p]$ and \Cref{col:approximation_via_expander}. Let $\xx, \yy \perp \vecone$ be maximising the right hand side of the previous inequality. Then, also $\xx' = \xx - \frac{\xx^T \dd^{(i,j)}_p}{\norm{\dd^{(i,j)}_p}_2}\vecone$ and $\yy' = \yy - \frac{\yy^T \dd^{(i,j)}_p}{\norm{\dd^{(i,j)}_p}_2}\vecone$ are maximizing. We obtain
\begin{align}
     2 \frac{\xx^T(\LL_{\dir{G}^{(i,j)}[V^{(i,j)}_p]} - \LL_{\tilde{\dir{G}}^{(i,j)}[V^{(i,j)}_p]})\yy}{ \xx^T \LL_{G(\dd^{(i,j)}_p)} \xx +  \yy  \LL_{G(\dd^{(i,j)}_p)} \yy^T} &= 2\frac{\xx'^T(\LL_{\dir{G}^{(i,j)}[V^{(i,j)}_p]} - \LL_{\tilde{\dir{G}}^{(i,j)}[V^{(i,j)}_p]})\yy'}{ \xx'^T \LL_{G(\dd^{(i,j)}_p)} \xx' +  \yy'  \LL_{G(\dd^{(i,j)}_p)} \yy'^T} \nonumber \\
     &= 2\frac{\xx'^T(\LL_{\dir{G}^{(i,j)}[V^{(i,j)}_p]} - \LL_{\tilde{\dir{G}}^{(i,j)}[V^{(i,j)}_p]})\yy'}{ \xx'^T \DD_{H^{(i,j)}[V^{(i,j)}_p]} \xx' +  \yy'  \DD_{H^{(i,j)}[V^{(i,j)}_p]} \yy'^T} \nonumber \\
     &= \norm{\DD_{H^{(i,j)}[V^{(i,j)}_p]}^{+/2}(\LL_{\dir{G}^{(i,j)}[V^{(i,j)}_p]} - \LL_{\tilde{\dir{G}}^{(i,j)}[V^{(i,j)}_p]})\DD_{H^{(i,j)}[V^{(i,j)}_p]}^{+/2}} 
     \label{eq:sparsepatch4}
\end{align}
where the last equality is by \Cref{lem:b2}. Next we upper bound
\begin{align*}
    \norm{\DD_{H^{(i,j)}[V^{(i,j)}_p]}^{+/2}\LL_{\dir{G}^{(i,j)}[V^{(i,j)}_p]}\DD_{H^{(i,j)}[V^{(i,j)}_p]}^{+/2}} 
\end{align*}
and 
\begin{align*}
    \norm{\DD_{H^{(i,j)}[V^{(i,j)}_p]}^{+/2}\LL_{\tilde{\dir{G}}^{(i,j)}[V^{(i,j)}_p]}\DD_{H^{(i,j)}[V^{(i,j)}_p]}^{+/2}}.
\end{align*}
By \Cref{lem:b4} we have
\begin{align*}
    \norm{\DD_{H^{(i,j)}[V^{(i,j)}_p]}^{+/2}\LL_{\dir{G}^{(i,j)}[V^{(i,j)}_p]}\DD_{H^{(i,j)}[V^{(i,j)}_p]}^{+/2}} \leq \max\left \{\norm{\LL_{\dir{G}^{(i,j)}[V^{(i,j)}_p]} \DD_{H^{(i,j)}[V^{(i,j)}_p]}^{+}}_{1}, \norm{\LL_{\dir{G}^{(i,j)}[V^{(i,j)}_p]}^T \DD_{H^{(i,j)}[V^{(i,j)}_p]}^{+}}_{1} \right \}.
\end{align*}
Since for every $v$, the undirected graph $\omega^{\min}_{\dir{G}} \cdot 2^{i + 1} \cdot H^{(i,j)}[V^{(i,j)}_p]$ satisfies
\begin{align*}
    \deg_{\omega^{\min}_{\dir{G}} \cdot 2^{i + 1} \cdot H^{(i,j)}[V^{(i,j)}_p]}(v) \geq \max\left\{\deg^+_{\dir{G}^{(i,j)}[V^{(i,j)}_p]}(v), \deg^-_{\dir{G}^{(i,j)}[V^{(i,j)}_p]}(v) \right\}
\end{align*}
we have 
\begin{align}
    \norm{\DD_{H^{(i,j)}[V^{(i,j)}_p]}^{+/2}\LL_{\dir{G}^{(i,j)}[V^{(i,j)}_p]}\DD_{H^{(i,j)}[V^{(i,j)}_p]}^{+/2}} \leq \omega^{\min}_{\dir{G}} \cdot 2^{i + 2}.  
    \label{eq:sparsepatch5}
\end{align}
Since we only used the in and out degrees of $\dir{G}^{(i,j)}[V^{(i,j)}_p]$ in the above, and $\tilde{\dir{G}}^{(i,j)}[V^{(i,j)}_p]$ has exactly the same degrees, we analogously conclude
\begin{align}
     \norm{\DD_{H^{(i,j)}[V^{(i,j)}_p]}^{+/2}\LL_{\tilde{\dir{G}}^{(i,j)}[V^{(i,j)}_p]}\DD_{H^{(i,j)}[V^{(i,j)}_p]}^{+/2}} \leq \omega^{\min}_{\dir{G}} \cdot 2^{i + 2}.  
    \label{eq:sparsepatch6}
\end{align}
Chaining inequalities \eqref{eq:sparsepatch1},  \eqref{eq:sparsepatch2}, \eqref{eq:sparsepatch3},  \eqref{eq:sparsepatch4},  \eqref{eq:sparsepatch5} and  \eqref{eq:sparsepatch6} yields 
\begin{align*}
    \norm{(\UU_{L_{\dir{G}}})^{+/2}(\dir{G}^{(i,j)} - \tilde{\dir{G}}^{(i,j)})( \UU_{L_{\dir{G}}})^{+/2}}_2 \leq 128 \cdot \Exp(2 \cdot (\log n)^{\gamma})
\end{align*}
which concludes our proof. 
\end{proof} 

Next we analyse the number of edges of $\tilde{\dir{G}}^{(i, j)}$

\begin{lemma}
$|E(\tilde{\dir{G}}^{(i, j)})| = O(n)$.
\label{patch_edges}
\end{lemma}
\begin{proof}
It is easy to see that the patching routine adds at most $2|V^{(i, j)}_p|$ edges to graph $\tilde{\dir{G}}^{(i, j)}[V^{(i, j)}_p]$, since each added edge repairs either the desired in-degree or the desired out-degree of a vertex. The result follows since $\sum_{p}|V^{(i, j)}_p| = n$.  
\end{proof}

We show the main lemma of this subsection by summing up the parts. 

\begin{proof}[Proof of \Cref{lem:global_patching}]
Notice that each expander decomposition peels of half of the edges, and thus every edge is part of an unique expander part by the end of the procedure \textsc{Sparsify}() in Algorithm \ref{alg:sparsepatchglobal}. Thus our algorithm exactly preserves the in- and out-degrees of $\dir{G}_1 = \beta \cdot \mathcal{U}(\dir{G}) + \dir{G}$, since each individual patching exactly preserves degrees. Since the graph $\dir{G}_2 = \beta \cdot \mathcal{U}(\dir{G}) + \dir{R}$ remains connected the null-spaces are unaltered. Further, since $\dir{R}$ is the sum of $\tilde{O(1)}$ graphs with $O(n)$ edges the total amount of edges of $\dir{R}$ is bounded by $\tilde{O}(n)$. 

Finally, we show the approximation bound. By \Cref{lem:b9} we have
\begin{align*}
        \norm{\II_{\im(\mathcal{U}^{(\beta)}(\dir{G}}) - \LL_{\dir{G}_2}^+\LL_{\dir{G}_1}}_{\UU_{\LL_{\dir{G}_2}} \rightarrow \UU_{\LL_{\dir{G}_2}}} \leq  \norm{\UU_{\LL_{\dir{G}_2}}^{+/2} (\LL_{\dir{R}} - \LL_{\dir{G}}) \UU_{\LL_{\dir{G}_2}}^{+/2}}. 
\end{align*}
We use \Cref{lem:b2} to obtain 
\begin{align*}
    \norm{\UU_{\LL_{\dir{G}_2}}^{+/2} (\LL_{\dir{R}} - \LL_{\dir{G}}) \UU_{\LL_{\dir{G}_2}}^{+/2}} &= 2 \max_{\xx, \yy \neq 0} \frac{\xx^T (\LL_{\dir{R}} - \LL_{\dir{G}}) \yy}{\xx^T (\beta \UU_{\LL_{\dir{G}}} + \UU_{\LL_{\dir{R}}})\xx + \yy^T(\beta \UU_{\LL_{\dir{G}}} + \UU_{\LL_{\dir{R}}}) \yy} \\ 
    &\leq 2\max_{\xx, \yy \neq 0} \frac{\xx^T (\LL_{\dir{R}} - \LL_{\dir{G}}) \yy}{\beta \xx^T \UU_{\LL_{\dir{G}}}\xx + \beta \yy^T \UU_{\LL_{\dir{G}}}\yy}. 
\end{align*}
Using \Cref{lem:b2} again we have 
\begin{align*}
     2\max_{\xx, \yy \neq 0} \frac{\xx^T (\LL_{\dir{R}} - \LL_{\dir{G}}) \yy}{\beta \xx^T \UU_{\LL_{\dir{G}}}\xx + \beta \yy^T \UU_{\LL_{\dir{G}}}\yy} &= \frac{1}{\beta} \norm{\UU_{\LL_{\dir{G}}}^{+/2}(\LL_{\dir{R}} - \LL_{\dir{G}})\UU_{\LL_{\dir{G}}}^{+/2}}_2 \\
     &= \frac{1}{\beta} \norm{\UU_{\LL_{\dir{G}}}^{+/2}\sum_{i,j}(\LL_{\dir{\tilde{G}}^{(i,j)}} - \LL_{\dir{G}^{(i,j)}})\UU_{\LL_{\dir{G}}}^{+/2}}_2 \\
     &\leq  \frac{1}{\beta} \sum_{i,j} \norm{\UU_{\LL_{\dir{G}}}^{+/2}(\LL_{\dir{\tilde{G}}^{(i,j)}} - \LL_{\dir{G}^{(i,j)}})\UU_{\LL_{\dir{G}}}^{+/2}}_2 \\
     & \stackrel{i)}{\leq} \frac{10 \cdot P \cdot \log n \cdot 128 \cdot \Exp(2 \cdot (\log n)^{\gamma})}{\beta} \stackrel{ii)}{\leq} \frac{1}{2}
\end{align*}
where i) follows from \Cref{lem:expander_patching_error} and ii) is by $\beta \geq 128 \cdot 20 \cdot P \cdot \log n \cdot \Exp(2 \cdot (\log n)^{\gamma})$. Chaining the inequalities shows the approximation statement and concludes our proof. 
\end{proof} 

\subsection{Sparsifying the Undirected Part}

The final task we have left is to sparsify the undirected graph $\mathcal{U}(\dir{G})$. This can be more or less directly achieved by employing a sparsification theorem presented in \cite{chuzhoy2020deterministic}. Mainly for notational convenience in our algorithm, we adapt this sparsifcation technique to be degree preserving in \Cref{sec:deg_pres_undir_spars} and state the resulting lemma here. See \Cref{lem:sparsify_undir_app} for the proof.

\begin{lemma}(Degree Preserving Sparsification)
There is a deterministic algorithm \\ $\specspardeg(G, \gamma)$ (\Cref{alg:specsparsdeg}) that given a parameter $\gamma \in (0,1)$ and an undirected graph $G = (V, E, \omega)$ with $n$ vertices and $m$ edges such that that $P := \frac{\max_{e \in E} \omega(e)}{\min_{e \in E} \omega(e)} = poly(n)$ computes $\tilde{G}$ satsifying
\begin{enumerate}
    \item $\Exp(-(\log n)^{\gamma}) \LL_G \preceq \LL_{\tilde{G}} \preceq \Exp((\log n)^{\gamma}) \LL_G$
    \item $\nnz(\AA) = \tilde{O}(n)$
\end{enumerate}
in time
\begin{align*}
    \tilde{O}(m^{1 + O(1/(\log n)^{\gamma/2})}\cdot (\log m)^{O((\log n)^\gamma)}) = m^{1 + o(1)}.    
\end{align*}
The graph $\tilde{G}$ has self loops and exactly the same degrees as $G$. 
\label{lem:sparsify_undir}
\end{lemma}

Next we apply degree preserving sparsification together with an appropriate scaling to sparsify the undirected part. 

\begin{lemma}
There exists a routine $\tilde{G} = \textsc{SpectralSparsifyDeg}(\mathcal{U}(\dir{G}), \gamma)$ (\Cref{alg:specsparsdeg}) that given the undirected graph $\mathcal{U}(\dir{G})$ computes an undirected graph $\tilde{G}$ with $\tilde{O}(n)$ edges so that $\LL^+_{\dir{G}_3} = (\frac{\beta}{\eta} \LL_{\tilde{G}} + \LL_{\dir{R}})^+$ is an $\left(1 - \frac{1}{2 \cdot \Exp(4 \cdot (\log n)^{\gamma})}\right)$-approximate pseudoinverse of $\LL_{\dir{G}_2} = \beta \UU_{\LL_{\dir{G}}} + \LL_{\dir{R}}$ with respect to $\UU_{\LL_{\dir{G}_3}}$. The degrees of $\tilde{G}$ and $\mathcal{U}(\dir{G})$ are the same.
\label{lem:undir_spars_part}
\end{lemma}
\begin{proof}
The sparsity and degree preservation follows directly from \Cref{lem:sparsify_undir}. To show the approximation property, we use \Cref{lem:b9} and obtain
\begin{align}
\norm{\II_{\im(\LL_{\dir{G}})} - \LL_{\dir{G}_3}^+ \LL_{\dir{G}_2}}_{\UU_{\LL_{\dir{G}_3}} \rightarrow \UU_{\LL_{\dir{G}_3}}} 
\leq \norm{\UU_{\LL_{\dir{G}_3}}^{+/2}\left(\frac{\beta}{\eta} \LL_{\tilde{G}} - \beta \LL_{\mathcal{U}(\dir{G})}\right)\UU_{\LL_{\dir{G}_3}}^{+/2}}_2. 
\label{eq:undir1}
\end{align} 
Then we use \Cref{lem:b2} twice with $\frac{\beta}{\eta} \LL_{\tilde{G}} \preceq \frac{\beta}{\eta} \LL_{\tilde{G}} + \UU_{\LL_{\dir{R}}}$ to obtain 
\begin{align}
    \norm{\UU_{\LL_{\dir{G}_3}}^{+/2}\left(\frac{\beta}{\eta} \LL_{\tilde{G}} - \beta \LL_{\mathcal{U}(\dir{G})}\right)\UU_{\LL_{\dir{G}_3}}^{+/2}}_2 &= 2 \max_{\xx, \yy \neq 0} \frac{\xx^T \left(\frac{\beta}{\eta} \LL_{\tilde{G}} - \beta \LL_{\mathcal{U}(\dir{G})}\right) \yy}{ \xx^T (\frac{\beta}{\eta} \LL_{\tilde{G}} + \UU_{\LL_{\dir{R}}}) \xx +  \yy^T (\frac{\beta}{\eta} \LL_{\tilde{G}} + \UU_{\LL_{\dir{R}}}) \yy} \nonumber \\
    &\leq 2 \max_{\xx, \yy \neq 0} \frac{\xx^T \left(\frac{\beta}{\eta} \LL_{\tilde{G}} - \beta \LL_{\mathcal{U}(\dir{G})}\right) \yy}{ \xx^T (\frac{\beta}{\eta} \LL_{\tilde{G}}) \xx +  \yy^T (\frac{\beta}{\eta} \LL_{\tilde{G}}) \yy} \nonumber \\
    &= \norm{ \LL_{\tilde{G}}^{+/2}( \LL_{\tilde{G}} -  \eta \LL_{\mathcal{U}(\dir{G})}) \LL_{\tilde{G}}^{+/2}}. \label{eq:undir2}
\end{align}
Next we compute 
\begin{align}
    \norm{ \LL_{\tilde{G}}^{+/2}( \LL_{\tilde{G}} - \eta \LL_{\mathcal{U}(\dir{G})}) \LL_{\tilde{G}}^{+/2}} 
    = \norm{\II_{\im(\LL_{\tilde{G}})} - \eta \LL_{\tilde{G}}^{+/2}\LL_{\mathcal{U}(\dir{G})}\LL_{\tilde{G}}^{+/2}}_2. \label{eq:undir3}
\end{align}
We use the standard strategy of bounding the square. Let $\MM := \LL_{\tilde{G}}^{+/2}\LL_{\mathcal{U}(\dir{G})}\LL_{\tilde{G}}^{+/2}$ be a shorthand. Then we have
\begin{align*}
    \norm{\II_{\im(\LL_{\tilde{G}})} - \eta  \MM}_2 &= \max_{\xx \in \im(\LL_{\tilde{G})}: \norm{\xx}_2 = 1} \xx^T (\II_{\im{\mathcal{U}(\dir{G})}} - \eta  \MM)^T(\II_{\im{\mathcal{U}(\dir{G})}} - \eta  \MM) \xx \\
    &= 1 + \max_{\xx \in \im{\LL_{\tilde{G}}}: \norm{\xx}_2 = 1} -2\eta \xx^T \MM \xx + \eta^2 \xx^T \MM^2 \xx \\
    &\leq 1 -2\eta \lambda_{*}(\MM) + \eta^2 \norm{\MM}_2^2 
\end{align*}
where $\lambda_{*}(\MM)$ denotes the smallest non-zero eigenvalue of $\MM$. We first lower bound $\lambda_{*}(\MM)$. By \Cref{lem:sparsify_undir} we have
\begin{align*}
    \frac{1}{2^{(\log n)^{\gamma}}}\LL_{\tilde{G}} &\preceq \LL_{\mathcal{U}(\dir{G})} \\
    \frac{1}{2^{(\log n)^{\gamma}}}\II_{\im(\tilde{G})} &\preceq \LL_{\tilde{G}}^{+/2}\LL_{\mathcal{U}(\dir{G})}\LL_{\tilde{G}}^{+/2} 
\end{align*}
and thus $\Exp(-(\log n)^{\gamma}) \leq \lambda_{*}(\MM)$. We obtain $\norm{\MM}_2 \leq \Exp(2 \cdot (\log n)^{\gamma})$ in an analogous way from \Cref{lem:sparsify_undir}. Then, we use $\eta = \Exp(-3(\log n)^{\gamma}) \leq \frac{\lambda_{*}(\MM)}{\norm{M}_2^2}$ to conclude. 
\begin{align}
    \norm{\II_{\im(\LL_{\tilde{G}})} - \eta  \MM}_2 \leq 1 - \frac{1}{2 \cdot \Exp(4 \cdot (\log n)^{\gamma})} \label{eq:undir4}
\end{align}
where we use $\sqrt{1 - \epsilon} \leq 1 - \epsilon/2$ for $\epsilon \in (0,1)$. Chaining inequalities \eqref{eq:undir1}, \eqref{eq:undir2}, \eqref{eq:undir3} and \eqref{eq:undir4} shows the desired approximate pseudoinverse property and finishes the proof. 
\end{proof}

\subsection{Loss in Condition Number}

The overarching goal of the recursion of a squaring solver is to reduce the condition number of the undirectification of the graph. When a constant condition number is reached, the recursion terminates.  However, since our sparsification has relatively low accuracy, we have to ensure that our losses do not outweigh our gains when we call it. In this subsection we bound the loss in condition number incurred by our global sparsification routine. We address the loss in condition number incurred by all of the three steps. To do so we consider spectral bounds.

\paragraph{By Partial Symmetrization.} The first part is the easiest. Since $\UU_{\LL_{\mathcal{U}^{(\beta)}(\dir{G})}} = (1 + \beta) \UU_{\LL_{\dir{G}}}$ $\beta$-partial-symmetrization does not change the condition number. Further we observe

\begin{observation}
$\UU_{\LL_{\dir{G}}} \preceq \UU_{\LL_{\mathcal{U}^{(\beta)}(\dir{G})}} \preceq (1 + \beta) \UU_{\LL_{\dir{G}}}$
\label{obs:partial_sym_cond}
\end{observation}

\paragraph{By Sparsifying the Directed Part.}

Next we analyse the loss incurred by patching undirected expanders. This part crucially relies on a spectral upper bound on $\UU_{\LL_{\dir{R}}}$. We show the following lemma

\begin{lemma}
    $\beta \UU_{\LL_{\dir{G}}} \preceq \beta \UU_{\LL_{\dir{G}}} + \UU_{\LL_{\dir{R}}} \preceq  (\beta + \tilde{O}(1) \cdot \Exp(2(\log n)^{\gamma})) \UU_{\LL_{\dir{G}}}$ 
\label{lem:dirspars_cond}
\end{lemma}
\begin{proof}
The first inequality, $\beta \UU_{\LL_{\dir{G}}} \preceq \beta \UU_{\LL_{\dir{G}}} + \UU_{\LL_{\dir{R}}}$, is clear since $\veczero \preceq \LL_{\mathcal{U}(\dir{R})} = \UU_{\LL_{\dir{R}}}$. To upper bound $\UU_{\LL_{\dir{R}}}$ we compute
\begin{align*}
    \norm{\UU_{\LL_{\dir{G}}}^{+/2}(\UU_{\LL_{\dir{R}}} - \UU_{\LL_{\dir{G}}})\UU_{\LL_{\dir{G}}}^{+/2}}_2 &\leq 2 \norm{\UU_{\LL_{\dir{G}}}^{+/2}(\LL_{\dir{G}} - \LL_{\dir{R}})\UU_{\LL_{\dir{G}}}^{+/2}}_2 \\
    &= 2\norm{\UU_{\LL_{\dir{G}}}^{+/2}\sum_{i,j}(\LL_{\dir{G}^{(i, j)}} - \LL_{\dir{\tilde{G}}^{(i,j)}})\UU_{\LL_{\dir{G}}}^{+/2}}_2 \\
    &\leq 2 \sum_{i,j} \norm{\UU_{\LL_{\dir{G}}}^{+/2}(\LL_{\dir{G}^{(i, j)}} - \LL_{\dir{\tilde{G}}^{(i,j)}})\UU_{\LL_{\dir{G}}}^{+/2}}_2 \\
    &\leq \tilde{O}(1) \cdot \Exp(2(\log n)^{\gamma})
\end{align*}
where we use that the sum is over $\tilde{O}(1)$ edge weight buckets $i$ and $\tilde{O}(1)$ expander decomposition layers $j$. We conclude 
\begin{align*}
    \UU_{\LL_{\dir{G}}}^{+/2}(\UU_{\LL_{\dir{R}}} - \UU_{\LL_{\dir{G}}})\UU_{\LL_{\dir{G}}}^{+/2} &\preceq \tilde{O}(1) \cdot \Exp(2(\log n)^{\gamma}) \II_{\im(\LL_{\dir{G}})} \\
    \UU_{\LL_{\dir{G}}}^{+/2}\UU_{\LL_{\dir{R}}}\UU_{\LL_{\dir{G}}}^{+/2} - \II_{\im(\LL_{\dir{G}})} &\preceq \tilde{O}(1) \cdot \Exp(2(\log n)^{\gamma}) \II_{\im(\LL_{\dir{G}})} \\
    \UU_{\LL_{\dir{R}}} &\preceq \tilde{O}(1) \cdot \Exp(2(\log n)^{\gamma}) \UU_{\LL_{\dir{G}}}.
\end{align*}
This concludes our proof. 
\end{proof}

\paragraph{By Sparsifying the Undirected Part.} Finally, we analyse the loss incurred by undirected sparsification. This part directly follows from the approximation guarantee of \Cref{lem:sparsify_undir}. 

\begin{lemma}
For $\tilde{G} = \specspardeg(\mathcal{U}(\dir{G}), \gamma)$ we have 
\begin{align*}
    \Exp(-4(\log n)^{\gamma}) \UU_{\LL_{\dir{G}_2}}  \preceq \UU_{\LL_{\dir{G}_3}}  \preceq \Exp(4(\log n)^{\gamma}) \UU_{\LL_{\dir{G}_2}} 
\end{align*}
\label{lem:undirspars_cond}
\end{lemma}
\begin{proof}
Recall that $\UU_{\LL_{\dir{G}_2}} = \beta \UU_{\LL_{\dir{G}}} + \UU_{\LL_{\dir{R}}}$ and $\UU_{\LL_{\dir{G}_3}} = \frac{\beta}{\eta} \LL_{\tilde{G}} + \UU_{\LL_{\dir{R}}}$. We show the following stronger statement
\begin{align*}
    \Exp(-4(\log n)^{\gamma}) \beta \UU_{\LL_{\dir{G}}} + \UU_{\LL_{\dir{R}}}  \preceq \frac{\beta}{\eta} \LL_{\tilde{G}} + \UU_{\LL_{\dir{R}}} \preceq \Exp(4(\log n)^{\gamma}) \beta \UU_{\LL_{\dir{G}}} + \UU_{\LL_{\dir{R}}}.
\end{align*}
After substracting $\UU_{\LL_{\dir{R}}}$ and dividing by $\beta$ we have  
\begin{align*}
       \Exp(-4(\log n)^{\gamma})  \UU_{\LL_{\dir{G}}} \preceq  \frac{1}{\eta} \LL_{\tilde{G}} \preceq  \Exp(4(\log n)^{\gamma}) \UU_{\LL_{\dir{G}}}
\end{align*}
which directly follows from \Cref{lem:sparsify_undir}.
\end{proof}

\subsection{Proof of \Cref{lem:globalspars}}

Now that we have assembled the pieces we are ready to prove \Cref{lem:globalspars}, the main lemma of this section. 

\begin{proof}[Proof of \Cref{lem:globalspars}]
We show each point in the enumeration seperately. 
\begin{enumerate}
    \item The statement for $i = 1$ is by \Cref{lem:partial_sym_approx} and $\beta = \tilde{O}(1) \cdot \Exp(2(\log n)^{\gamma})$, for $i = 2$ it is by \Cref{lem:global_patching} and for $i = 3$ it is by \Cref{lem:undir_spars_part}.
    \item The statement for $i = 1$ is by \Cref{obs:partial_sym_cond}, for $i = 2$ it is by \Cref{lem:dirspars_cond} and for $i = 3$ it is by \Cref{lem:undirspars_cond}. 

    \item Clearly, $\dir{G}_1$ has at most twice as many edges as $\dir{G} = \dir{G}_0$. Further, $|E(\dir{G}_2)| \leq |E(\dir{G}_1)| + |E(\dir{R})|$ and $|E(\dir{R})| = \tilde{O}(n)$ by \Cref{lem:global_patching}. Finally $|E(\dir{G}_3)| \leq |E(\tilde{G})| + |E(\dir{R})| = \tilde{O}(n)$ by \Cref{lem:undir_spars_part} and \Cref{lem:global_patching}.  
.    \item Adding $\beta$ times the undirectifaction to an Eulerian graph increases the in- and out- degrees by exactly a factor of $(1 + \beta)$. Then the statement follows from the degree preservation of our sparsification routines shown in \Cref{lem:global_patching} and \Cref{lem:undir_spars_part} and the extra scaling by $\frac{1}{\eta}$ of the undirected part of $\dir{G}_3$
\end{enumerate}
Finally, the runtime follows from \Cref{col:exp_dec_det}. 
\end{proof}

\section{Approximately Squaring Directed Laplacians}
\label{sec:squaring}
Given a directed Eulerian Laplacian $\LL_{\dir{G}} = \DD_{\dir{G}} - \AA^T_{\dir{G}}$ we call $\LL_{\dir{G}^2} = \DD_{\dir{G}} - \AA^T_{\dir{G}}\DD_{\dir{G}}^{-1}\AA^T_{\dir{G}}$ its square. The squaring operation not only preserves degrees, but also Eulerianness. In this section we introduce a deterministic algorithm for high accuracy sparsification of the square of an Eulerian Laplacian in roughly $\nnz(\LL_{\dir{G}})$ time. Since squaring can drastically increase the density, the run-time of this routine sometimes has to be sub-linear in the size of the graph it is approximating. Thus it is clear that we can only afford work with implicit representations of $\LL_{\dir{G}^2}$. Our routines are based on and inspired by \cite{kyng2015sparsified} and \cite{peng2021sparsified}. A similar sparisification routine is developed in \cite{akm20}. We first state the main lemma of this section. 

\begin{lemma}
For every strongly connected Eulerian Laplacian $\LL_{\dir{G}} = \DD_{\dir{G}} - \AA_{\dir{G}}^T$ the algorithm \sparsesquare($\dir{G}, \epsilon$) computes a strongly connected Eulerian graph $\tilde{\dir{G}}^2$ in time $O(\nnz(\LL_{\dir{G}})^{1 + o(1)}/\epsilon^4)$ so that $\LL_{\tilde{\dir{G}}^2} = \DD_{\tilde{\dir{G}}^2} - \AA^T_{\tilde{\dir{G}}^2}$ is an $\epsilon$-graph-approximation of $\LL_{\dir{G}^2} = \DD_{\dir{G}} - \AA^T_{\dir{G}}\DD_{\dir{G}}^{-1}\AA^T_{\dir{G}}$, $\nnz(\LL_{\dir{G}^2}) = O(\nnz(\LL_{\tilde{\dir{G}}})/\epsilon^4)$ and degrees are preserved. \label{lem:squaring}
\end{lemma}

As evident from the pseudocode of \sparsesquare() (Algorithm \ref{alg:sparsesquare}) our approach to implicitly sparsifying $\LL_{\dir{G}^2}$ breaks down its adjacency matrix into rank one contributions
\begin{align*}
    \AA^T_{\dir{G}}\DD_{\dir{G}}^{-1}\AA^T_{\dir{G}} = \sum_{i = 1}^n \underbrace{\DD_{\dir{G}}^{-1}(i,i) (\AA_{\dir{G}}(i,:))^T \cdot (\AA_{\dir{G}}(:,i))^T}_{=: \AA_{\dir{G}_i}}
\end{align*}
as commonplace in randomized squaring sparsification algorithms \cite{peng2021sparsified}. Then \sparseproduct() (Algorithm \ref{alg:sparseproduct}) computes $\AA_{\dir{H}_i} $, a suitable sparse approximation of $\AA_{\dir{G}_i}$. The full adjacency matrix is obtained by simply summing up these contributions. Note that \sparseproduct() will preserve in- and out-degrees exactly, ensuring that the resulting approximation preserves degrees, which in turn ensures $\LL_{\dir{G}^2}$ being Eulerian. 

\begin{algorithm}
\caption{\sparsesquare($\dir{G}, \epsilon$)}
\label{alg:sparsesquare}
\For{$i = 1, ..., n$}{
    $\AA_{\dir{H}_i} = \sparseproduct(\AA_{\dir{G}}(i,:), \AA_{\dir{G}}(:,i), \epsilon)$ \\
}
$\AA_{\tilde{\dir{G}}^2} = \sum_{i = 1}^n \AA_{\dir{H}_i}$\\
\Return{$\tilde{\dir{G}}^2$}

\end{algorithm}

\subsection{Deterministically Sparsifying Bipartite Product Graphs}

High accuracy spectral sparsifiers for undirected bipartite product graphs are known to the literature \cite{kyng2015sparsified}. We will first adapt these to our setting by showing that a simple greedy patching scheme can be employed to ensure the conservation of degrees. Then we will directly use sparsified bipartite product graphs to build \sparseproduct() (Algorithm \ref{alg:sparseproduct}), our routine for sparsifying $\LL_{\dir{G}_i}$. We start by defining the central object of this subsection. 
\begin{definition}
For $\aa \in \R^n_{\geq 0}$ and $\bb \in \R^n_{\geq 0}$ positive vectors so that $\norm{\aa}_1 = \norm{\bb}_1 = d$, let 
\begin{align*}
    \LL_{G(\aa, \bb)} := 
    \begin{pmatrix} 
    \diag(\aa) & \veczero \\
    \veczero & \diag(\bb)
    \end{pmatrix}
    - 
    \begin{pmatrix} 
    \veczero &  \frac{ \bb \aa^T}{d} \\
    \frac{\aa \bb^T}{d} & \veczero
    \end{pmatrix}
\end{align*}
be the Laplacian of the undirected bipartite product graph $G(\aa, \bb)$ of $\aa$ and $\bb$.
\label{def:biprod}
\end{definition}

\begin{definition}
We call an undirected graph Laplacian 
\begin{align*}
    \LL_{\tilde{G}(\aa, \bb)} = \begin{pmatrix} 
    \diag(\aa) & \veczero \\
    \veczero & \diag(\bb)
    \end{pmatrix}
    - 
    \begin{pmatrix} 
    \veczero &  \CC \\
    \CC^T & \veczero
    \end{pmatrix}
\end{align*}
with $\CC \in \R^{n \times n}_{\geq 0}$ an $\epsilon$-bipartite-sparsifier of $\LL_{G(\aa, \bb)}$ if 
\begin{align*}
    (1 - \epsilon) \LL_{\tilde{G}(\aa, \bb)} \preceq \LL_{G(\aa, \bb)} \preceq (1 + \epsilon) \LL_{\tilde{G}(\aa, \bb)}
\end{align*}
and $\nnz(\LL_{\tilde{G}(\aa, \bb)}) = O((\nnz(\aa) + \nnz(\bb))/\epsilon^4)$.
\label{def:high_accuracy}
\end{definition}

Next we introduce the previously known spectral sparsification lemma we will use. 

\begin{lemma}[Lemma G.16 in \cite{kyng2015sparsified}]
There is a routine $\textsc{WeightedBipartiteExpander}(\aa, \bb, \epsilon)$ such that for any vectors $\aa, \bb \in \R^n_{\geq 0}$ and a parameter $\epsilon \in (0, 1/2]$, it returns in time $O((\nnz(\aa) + \nnz(\bb)) \epsilon^{-4})$ a bipartite graph $\tilde{G}(\aa, \bb)$ on the same bipartition as $G(\aa, \bb)$ with $O((\nnz(\aa) + \nnz(\bb))\epsilon^{-4})$ edges such that 
\begin{align*}
    (1 - \epsilon)\LL_{\tilde{G}(\aa, \bb)} \preceq \LL_{G(\aa,\bb)} \preceq (1 + \epsilon)\LL_{\tilde{G}(\aa, \bb)}
\end{align*}
\label{lem:g16}
\end{lemma}

The routine $\sparsebipartite()$ (Algorithm \ref{alg:sparsebipartite}) combines \textsc{WeightedBipartiteExpander}() from \Cref{lem:g16} with the greedy patching procedure \patchbipartite() (Algorithm \ref{alg:sparsebipartite}) to construct an $\epsilon$-bipartite-sparsifier. We first show the desired runtime and preservation of degrees. 

\begin{algorithm}
\caption{\sparsebipartite($\aa, \bb, \epsilon$) and \patchbipartite($\aa, \bb, \dd^A, \dd^B$)}
\label{alg:sparsebipartite}
\SetKwFunction{algo}{\sparsebipartite}\SetKwFunction{proc}{\patchbipartite}
\SetKwProg{myalg}{Algorithm}{}{}
\SetKwProg{myproc}{Procedure}{}{}
\myalg{\algo{$\aa, \bb, \epsilon$}}{
$\epsilon' = \epsilon/128$\\
$H_1 = \textsc{WeightedBipartiteExpander}(\aa, \bb, \epsilon')$\\
$H_2 = H_1/(1 + \epsilon')$ \\
$\dd = \diag(\DD_{H_2})$ \\
$\LL_R = \patchbipartite(\aa, \bb, \dd(1:n), \dd((n+1):2n))$\\
\Return{$\LL_{H_2} + \LL_R$}
}

\myproc{\proc{$\aa, \bb, \dd^A, \dd^B$}}{
$\tilde{\aa} = \aa - \dd^A$\\
$\tilde{\bb} = \bb - \dd^B$\\
$\LL_R = \veczero^{2n \times 2n}$\\
\While{$\tilde{\aa} \neq \veczero$}{
Let $i$ and $j$ be arbitrary such that $\tilde{\aa}(i) > 0$ and $\tilde{\bb}(j) > 0$. \\
$w = \min\{\tilde{\aa}(i), \tilde{\bb}(j)\}$\\
$\tilde{\aa}(i) = \tilde{\aa}(i) - w$\\
$\tilde{\bb}(j) = \tilde{\bb}(j) - w$\\
$\LL_R = \LL_R + w(\ee_{i} - \ee_{j + n})(\ee_{i} - \ee_{j + n})^T$
}
\Return{$\LL_R$}
}
\end{algorithm}

\begin{lemma}
The routine $\LL_{\tilde{G}(\aa, \bb)} = \sparsebipartite(\aa, \bb, \epsilon)$ with $\aa, \bb \in R^n_{\geq 0}$ so that $\norm{\aa}_1 = \norm{\bb}_1$, and $\epsilon \leq 1/2$ terminates in time $O((nnz(\aa) + \nnz(\bb))/\epsilon^4)$ and ensures $\DD_{\tilde{G}(\aa, \bb)} = \DD_{G(\aa, \bb)}$.
\label{lem:runtime_deg}
\end{lemma}
\begin{proof}
Since for any graph $G$ we have $\ee_v^T \LL_G \ee_v = \deg_G(v)$ for all $v$ it is easy to see that the spectral approximation of \textsc{WeightedBipartiteExpander}() from \Cref{lem:g16} preserves degrees in the sense that for all $v$
\begin{align*}
    (1 - \epsilon') \deg_{G(\aa, \bb)}(v) \leq \deg_{H_1} \leq (1 + \epsilon') \deg_{G(\aa, \bb)}.
\end{align*}
From this it is immediate that $\deg_{H_2}(v) \leq \deg_{G(\aa, \bb)}(v)$. Then the routine \patchbipartite() iteratively computes the Laplacian of a bipartite patching graph so that $\deg_R(v) = \deg_{G(\aa, \bb)}(v) - \deg_{H_2}(v)$ for all $v$. It is clear that both $H_2$ and $G(\aa, \bb)$ being bipartite ensures the invariant that $\norm{\tilde{\aa}}_1 = \norm{\tilde{\bb}}_1$ throughout the execution of \patchbipartite() and $\tilde{\aa} \neq \veczero \implies \tilde{\bb} \neq \veczero$. Since each edge $(u,v)$ added to $R$ either ensures $\tilde{\aa}(u) = 0$ or $\tilde{\bb}(v) = 0$ the routine terminates after at most $\nnz(\aa) + \nnz(\bb)$ edges are added since $\nnz(\tilde{\aa}) + \nnz(\tilde{\bb}) \leq \nnz(\aa) + \nnz(\bb)$ initially. The result follows from these observations and \Cref{lem:g16}.
\end{proof}

To conclude this subsection we would like to use \Cref{lem:g16} to show that $\sparsebipartite()$ (Algorithm \ref{alg:sparsebipartite}) computes an $\epsilon$-bipartite-sparsifier of $\LL_{G(\aa, \bb)}$. The remaining obstruction is the quantification of the error introduced by the \patchbipartite() (Algorithm \ref{alg:sparsebipartite}). Our analysis of this error relies on the high expansion of $G(\aa, \bb)$.

\begin{figure}[h]
    \centering
    \includegraphics[width=8cm]{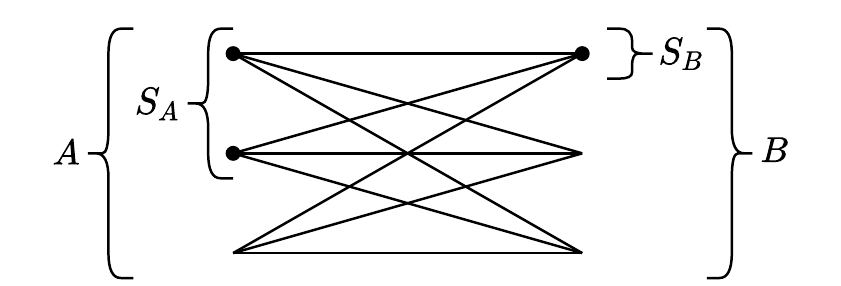}
    \caption{This figure depicts a bipartite graph $G(\aa, \bb)$, where the set $S$ as defined in \Cref{lem:bipart_prod_exp} consits of the vertices highlighted with black dots. Then the set is split into parts $S_A$ and $S_B$ as indicated.}
    \label{fig:bipartite_defs}
\end{figure}

\begin{lemma}
$G(\aa, \bb)$ is a $1/2$-expander.
\label{lem:bipart_prod_exp}
\end{lemma}
\begin{proof}
Consider any cut $\emptyset \subset S \subset V$. Without loss of generality, let $\vol_{G(\aa, \bb)}(S) \leq \vol_{G(\aa, \bb)}(V \setminus S)$. We distinguish between the two parts $A$ and $B$ of the bipartite graph and let $S_A \subseteq A$ and $S_B \subseteq B$ so that $S_A \cup S_B = S$. We first introduce some terminology that will prove useful:
\begin{align*}
    d := \norm{\aa}_1 (= \norm{\bb}_1) \hspace{20pt} t^A_{S_A} := \sum_{v \in S_A} \aa(v) \hspace{20pt} t^B_{S_B} := \sum_{v \in S_B} \bb(v).
\end{align*}
See \Cref{fig:bipartite_defs} for a drawing of the situation. Then we have 
\begin{align*}
    \delta_{G(\aa,\bb)}(S) &= \sum_{u \in S_A} \sum_{v \in B \setminus S_B} \frac{\aa(u)\bb(v)}{\norm{\aa}_1} + \sum_{v \in S_B} \sum_{u \in A \setminus S_a} \frac{\aa(u)\bb(v)}{\norm{\aa}_1} \\
    &= \frac{1}{d}(t^A_{S_A}(d - t^B_{S_B}) + t^B_{S_B}(d - t^A_{S_A}))
\end{align*}
and 
\begin{align*}
    \vol_{G(\aa,\bb)}(S) &= \sum_{u \in S_A} \sum_{v \in B}  \frac{\aa(u)\bb(v)}{\norm{\aa}_1} + \sum_{v \in S_B} \sum_{u \in A} \frac{\aa(u)\bb(v)}{\norm{\aa}_1} \\
    &= t^A_{S_A} + t^B_{S_B}.
\end{align*}
Finally we show
\begin{align*}
    \frac{\delta_{G(\aa,\bb)}(S)}{\vol_{G(\aa,\bb)}(S)} \geq 1/2 
    \iff (t^A_{S_A}(d - t^B_{S_B}) + t^B_{S_B}(d - t^A_{S_A})) & \geq \frac{d}{2} (t^A_{S_A} + t^B_{S_B}) \\
    \iff \frac{d}{2} (t^A_{S_A} + t^B_{S_B}) &\geq 2 t^B_{S_B} t^A_{S_A}.
\end{align*} 
Using $\vol_{G(\aa,\bb)}(S) \leq \vol_{G(\aa,\bb)}(V \setminus S)$ and $d = \vol_{G(\aa,\bb)}(V)$ we have $d \geq t^A_{S_A} + t^B_{S_B}$ and thus
\begin{align*}
    d(t^A_{S_A} + t^B_{S_B}) \geq (t^A_{S_A})^2 + (t^A_{S_A})^2 + 2 t^A_{S_A}t^B_{S_B} \geq 4 t^A_{S_A}t^B_{S_B}
\end{align*}
by since for any real numbers $a$ and $b$ we have $a^2 + b^2 \geq 2ab$. 
\end{proof}

Using the expansion of $\LL_{G(\aa, \bb)}$, we give a spectral upper bound on the error by patching next. 

\begin{lemma}
In the context of \sparsebipartite($\aa, \bb, \epsilon$) with $\aa, \bb \in \R^n_{\geq 0}$, $\norm{\aa}_1 = \norm{\bb}_1$, we have $\LL_R \preceq 128 \epsilon' \LL_{G(\aa, \bb)}$.
\label{lem:bip_patch_upper}
\end{lemma}
\begin{proof}
As observed in the proof of \Cref{lem:runtime_deg} we have $\deg_R(v) = \deg_{G(\aa, \bb)}(v) -  \deg_{H_2}(v)$ for all $v$. Since $\deg_{H_2}(v) \geq \frac{1 - \epsilon'}{1 + \epsilon'}\deg_{G(\aa, \bb)}(v)$ by \Cref{lem:g16} we have 
\begin{align}
    \deg_R(v) \leq \deg_{G(\aa, \bb)} - \frac{1 - \epsilon'}{1 + \epsilon'}\deg_{G(\aa, \bb)}(v) \leq 4\epsilon' \deg_{G(\aa, \bb)}(v).
    \label{eq:degbound}
\end{align}
In the following, let $\dd = \diag(\DD_{G(\aa, \bb)})$ be the degree vector of $G(\aa, \bb)$. We use \Cref{lem:b2} and arrive at
\begin{align}
    \norm{\LL_{G(\aa, \bb)}^{+/2} \LL_R \LL_{G(\aa, \bb)}^{+/2}}_2 = 2 \max_{\xx, \yy \neq \veczero} \frac{\xx^T \LL_R \yy}{\xx^T \LL_{G(\aa, \bb)} \xx + \yy^T \LL_{G(\aa, \bb)} \yy}
    \label{eq:upper_bound_norm_spectral_patch_bipartite}
\end{align}
Let $\xx$ and $\yy$ be maximizing the right hand side of  \eqref{eq:upper_bound_norm_spectral_patch_bipartite} and $\xx, \yy \perp \vecone$. Then, $\xx' = \xx - \frac{\xx^T \dd}{\norm{\dd}_2}\vecone$ and $\yy' = \yy - \frac{\yy^T \dd}{\norm{\dd}_2}\vecone$ are also maximizing since the all-ones vector is in the kernel of Laplacian matrices. By \Cref{col:approximation_via_expander} and \Cref{lem:bipart_prod_exp} we have $\frac{1}{16} (\diag(\dd) - \frac{\dd \dd^T}{\norm{\dd}_1}) \preceq \LL_{G(\aa, \bb)}$. We calculate
\begin{align*}
    \norm{\LL_{G(\aa, \bb)}^{+/2} \LL_R \LL_{G(\aa, \bb)}^{+/2}}_2 &= 2 \frac{\xx'^T \LL_R \yy'}{\xx'^T \LL_{G(\aa, \bb)} \xx' + \yy'^T \LL_{G(\aa, \bb)} \yy'} \\
    &\leq 32 \frac{\xx'^T \LL_R \yy'}{\xx'^T (\diag(\dd) - \frac{\dd \dd^T}{\norm{\dd}_1}) \xx' + \yy'^T (\diag(\dd) - \frac{\dd \dd^T}{\norm{\dd}_1}) \yy'} \\
    &= 32 \frac{\xx'^T \LL_R \yy'}{\xx'^T \diag(\dd) \xx' + \yy'^T \diag(\dd) \yy'} \\
    &\leq 16 \norm{\diag(\dd)^{+/2} \LL_R \diag(\dd)^{+/2}}
\end{align*}
where the last line follows from another application of \Cref{lem:b2}. By \Cref{lem:b4} and \eqref{eq:degbound} we have 
\begin{align*}
    \norm{\diag(\dd)^{+/2} \LL_R \diag(\dd)^{+/2}} \leq \norm{\LL_R \DD_{G(\aa, \bb)}}_1 \leq 8\epsilon'.
\end{align*}
We conclude 
\begin{align*}
    \LL_R \preceq 128 \epsilon' \LL_{G(\aa, \bb)}.
\end{align*}
\end{proof}

Finally, we assemble the pieces and show that \sparsebipartite() sparsifies $G(\aa, \bb)$ while preserving degrees. 

\begin{lemma}
There is a routine $\sparsebipartite(\aa, \bb, \epsilon)$ that given $\aa, \bb \in \R^n_{\geq 0}$ and a parameter $\epsilon \leq 1/2$ computes in $O((\nnz(a) + \nnz(\bb)) \epsilon^{-4})$ time an $\epsilon$-bipartite sparsifier $\LL_{\tilde{G}(\aa, \bb)}$ of $\LL_{G(\aa, \bb)}$.
\label{lem:full_bipartite_approx}
\end{lemma}
\begin{proof}
The runtime and degree preservation is clear from \Cref{lem:runtime_deg}. Further, we have 
\begin{align*}
    \frac{1 - \epsilon'}{1 + \epsilon'} \LL_{G(\aa, \bb)} \preceq \LL_{H_2} \preceq \LL_{H_2} + \LL_R = \LL_H
\end{align*}
and $\LL_{H_2} \preceq \LL_{G(\aa, \bb)}$ by \Cref{lem:g16}. From \Cref{lem:bip_patch_upper} we conclude
\begin{align*}
        \frac{1 - \epsilon'}{1 + \epsilon'} \LL_{G(\aa, \bb)} \preceq \LL_H \preceq (1 + 128\epsilon')\LL_{G(\aa, \bb)}.
\end{align*}
The desired approximation follows since $\epsilon' = \epsilon / 128$.
\end{proof}

\subsection{Sparsifying Directed Product Graphs via Bipartite Product Graphs}

The subroutine \sparseproduct() (Algorithm \ref{alg:sparseproduct}) exploits relationships between directed graphs and bipartite graphs, allowing us to reduce our problem to the sparsification problem on bipartite product graphs we solved in the previous subsection. Our object of interest in this subsection is defined as follows.

\begin{definition}
For $\aa \in \R^n_{\geq 0}$ and $\bb \in \R^n_{\geq 0}$ positive vectors so that $\norm{\aa}_1 = \norm{\bb}_1 = d$, let 
\begin{align*}
    \LL_{\dir{G}(\aa, \bb)} := \diag(\bb) - \frac{\aa \bb^T}{d}
\end{align*} 
denote the directed product graph of $\aa$ and $\bb$.
\label{def:dirprod}
\end{definition}

Inspecting $\LL_{G(\aa, \bb)}$ from \Cref{def:biprod} yields a natural approach to sparsifying $\LL_{\dir{G}(\aa, \bb)}$: Compute a $\epsilon$-bipartite sparsifier $\LL_H$ and then simply retrieve the adjacency matrix from its bottom left block. This is exactly what \sparseproduct() (Algorithm \ref{alg:sparseproduct}) does. Somewhat surprisingly, this approach directly yields a high accuracy approximation of $\LL_{\dir{G}(\aa, \bb)}$ with a mere constant factor loss in approximation. The argument relies on the fact that $G(\aa, \bb)$ is a constant expander (\Cref{lem:bipart_prod_exp}), which would be a serious obstruction for extending this strategy to sparsifying general directed graphs. For a graphical explanation of the relationship between directed graphs and bipartite graphs see \Cref{fig:dir_bipartie}.

\begin{figure}[h]
    \centering
    \includegraphics[width=8cm]{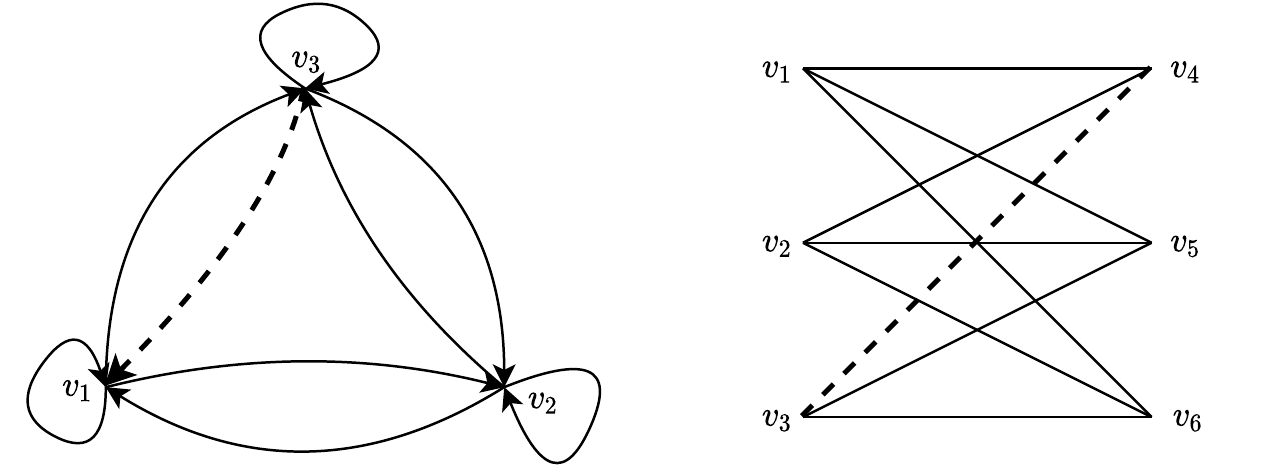}
    \caption{A directed graph naturally has an associated bipartite graph on twice as many vertices. The directed edge $(i, j)$ is then associated with the undirected edge $(i, j + |V|)$. For example, in the figure the two dashed edges are associated.}
    \label{fig:dir_bipartie}
\end{figure}

\begin{algorithm}
\caption{\sparseproduct($\aa, \bb, \epsilon$)}
$d = \norm{\aa}_1 (= \norm{\bb})$ \\
$\AA_{\tilde{G}(\aa, \bb)} = \sparsebipartite(\aa, \bb, \epsilon/128)$\\
$\AA^T_{\dir{\tilde{G}}(\aa, \bb)} = \AA_{\tilde{G}(\aa, \bb)}((n+1):2n, 1:n)$ \\
\Return{$\AA_{\dir{\tilde{G}}(\aa, \bb)}$}
\label{alg:sparseproduct}
\end{algorithm}

\begin{lemma}
For $\aa, \bb \in \R^n_{\geq 0}$ so that $\norm{\aa}_1 = \norm{\bb}_1$ and $\epsilon \leq 1/2$ the routine \sparseproduct($\aa, \bb, \epsilon$) (Algorithm \ref{alg:sparseproduct})  yields an $\epsilon$-graph-approximation $\LL_{\dir{\tilde{G}}(\aa, \bb)}$ of $\LL_{\dir{G}(\aa, \bb)}$ with $\nnz(\LL_{\dir{\tilde{G}}(\aa, \bb)}) = O((\nnz(\aa) + \nnz(\bb))\epsilon^{-4})$ in time $O((\nnz(\aa) + \nnz(\bb))\epsilon^{-4})$. The in- and out-degrees of $\dir{\tilde{G}}(\aa, \bb)$ and $\dir{G}(\aa, \bb)$ are identical. 
\label{lem:product}
\end{lemma}
\begin{proof}
The runtime and sparsity is clear from \Cref{lem:full_bipartite_approx}. The preservation of in- and out-degrees follows from the preservation of degrees in the biaprtite graph. It remains to show that 
\begin{align*}
    \norm{\LL_{\mathcal{U}(\dir{G}(\aa, \bb))}^{+/2} (\LL_{\dir{G}(\aa, \bb)} - \LL_{\dir{\tilde{G}}(\aa, \bb)}) \LL_{\mathcal{U}(\dir{G}(\aa, \bb))}^{+/2}} \leq \epsilon.
\end{align*}
By \Cref{lem:b2} we have
\begin{align}
   \norm{\LL_{\mathcal{U}(\dir{G}(\aa, \bb))}^{+/2} (\LL_{\dir{G}(\aa, \bb)} - \LL_{\dir{\tilde{G}}(\aa, \bb)}) \LL_{\mathcal{U}(\dir{G}(\aa, \bb))}^{+/2}} = 2 \max_{\xx, \yy \neq \veczero} \frac{\xx^T (\LL_{\dir{G}(\aa, \bb)} - \LL_{\dir{\tilde{G}}(\aa, \bb)}) \yy}{\xx^T \LL_{\mathcal{U}(\dir{G}(\aa, \bb))} \xx + \yy^T \LL_{\mathcal{U}(\dir{G}(\aa, \bb))} \yy}.
   \label{eq:prod_dir_b2}
\end{align}
Since $\LL_{\dir{G}(\aa, \bb)}$ and $\LL_{\dir{\tilde{G}}(\aa, \bb)}$ have identical in- and out-degrees, it is clear that $\vecone$ is in both the left and right kernel of $(\LL_{\dir{G}(\aa, \bb)} - \LL_{\dir{\tilde{G}}(\aa, \bb)})$
We let $\xx, \yy \perp \vecone$ be maximizing the right hand side of \eqref{eq:prod_dir_b2}. Further, we let $\dd := \diag(\DD_{G(\aa, \bb)})$ and
\begin{align*}
 &\hat{\xx} := 
\begin{pmatrix}
    \xx \\
    \xx
\end{pmatrix}
\hspace{20pt}
\hat{\yy} := 
\begin{pmatrix}
    \yy \\
    \yy 
\end{pmatrix} 
\hspace{20pt}
\hat{\xx}' := \hat{\xx} - \frac{\hat{\xx}^T \dd}{\norm{\dd}_2}\vecone_{2n} 
\hspace{20pt}
\hat{\yy}' := \hat{\yy} - \frac{\hat{\yy}^T \dd}{\norm{\dd}_2}\vecone_{2n} 
\\
&\xx' := \xx - \frac{\hat{\xx}^T \dd}{\norm{\dd}_2}\vecone_{n}
\hspace{20pt}
\yy' := \yy - \frac{\hat{\yy}^T \dd}{\norm{\dd}_2}\vecone_{n} 
\hspace{20pt}
\tilde{\xx} = \begin{pmatrix}
    \veczero \\
    \xx'
\end{pmatrix}
\hspace{20pt}
\tilde{\yy} = \begin{pmatrix}
    \yy' \\
     \veczero 
\end{pmatrix}.
\end{align*}
Then we have
\begin{align*}
    \norm{\LL_{\mathcal{U}(\dir{G}(\aa, \bb))}^{+/2} (\LL_{\dir{G}(\aa, \bb)} - \LL_{\dir{\tilde{G}}(\aa, \bb)}) \LL_{\mathcal{U}(\dir{G}(\aa, \bb))}^{+/2}} &= 2 \frac{\xx^T (\LL_{\dir{G}(\aa, \bb)} - \LL_{\dir{\tilde{G}}(\aa, \bb)}) \yy}{\xx^T \LL_{\mathcal{U}(\dir{G}(\aa, \bb))} \xx + \yy^T \LL_{\mathcal{U}(\dir{G}(\aa, \bb))} \yy} \\
    & = 2 \frac{\xx'^T (\LL_{\dir{G}(\aa, \bb)} - \LL_{\dir{\tilde{G}}(\aa, \bb)}) \yy'}{\xx^T \LL_{\mathcal{U}(\dir{G}(\aa, \bb))} \xx + \yy^T \LL_{\mathcal{U}(\dir{G}(\aa, \bb))} \yy} \\
    &= 4 \frac{\tilde{\xx}^T (\LL_{G(\aa, \bb)} - \LL_{\tilde{G}(\aa, \bb)}) \tilde{\yy}}{\hat{\xx}'^T  \LL_{G(\aa, \bb)} \hat{\xx}' + \hat{\yy}'^T \LL_{G(\aa, \bb)}  \hat{\yy}'} \\
    &\stackrel{i)}{\leq} 64  \frac{\tilde{\xx}^T (\LL_{G(\aa, \bb)} - \LL_{\tilde{G}(\aa, \bb)}) \tilde{\yy}}{\hat{\xx}'^T  \diag(\dd) \hat{\xx}' + \hat{\yy}'^T \diag(\dd) \hat{\yy}'} \\
    &\stackrel{ii)}{\leq} 64  \frac{\tilde{\xx}^T (\LL_{G(\aa, \bb)} - \LL_{\tilde{G}(\aa, \bb)}) \tilde{\yy}}{\tilde{\xx}^T  \diag(\dd) \tilde{\xx} + \tilde{\yy}^T \diag(\dd) \tilde{\yy}}
\end{align*}
where i) holds since $\LL_{G(\dd)} \preceq 16 \LL_{G(\aa, \bb)}$ by \Cref{col:approximation_via_expander} and $\xx'^T\LL_{G(\dd)}\xx' = \xx'^T\DD_{G(\dd)}\xx'$ by the choice of $\xx'$. Analogously we have $\yy'^T \diag(\dd)\yy' \leq 16 \yy'\LL_{G(\aa, \bb)}\yy'$. Further, ii) holds since setting parts of the vectors $\xx'$ and $\yy'$ to zero only decreases the value of the quadratic form with a positive diagonal matrix.  
Using $\LL_{G(\aa,\bb)} \preceq 2 \DD_{G(\aa,\bb)} = 2\diag(\dd)$ and \Cref{lem:b2} we conclude
\begin{align*}
    64 \frac{\tilde{\xx}^T (\LL_{G(\aa, \bb)} - \LL_{\dir{\tilde{G}}(\aa, \bb)}) \tilde{\yy}}{\tilde{\xx}^T  \diag(\dd) \tilde{\xx} + \tilde{\yy}^T \diag(\dd) \tilde{\yy}} \leq 128 \frac{\tilde{\xx}^T (\LL_{G(\aa, \bb)} - \LL_{\dir{\tilde{G}}(\aa, \bb)}) \tilde{\yy}}{\tilde{\xx}^T  \LL_{G(\aa, \bb)} \tilde{\xx} + \tilde{\yy}^T \LL_{G(\aa, \bb)}  \tilde{\yy}} \\
    = 128 \norm{\LL_{G(\aa, \bb)}^{+/2}(\LL_{G(\aa, \bb)} - \LL_{\dir{\tilde{G}}(\aa, \bb)})\LL_{G(\aa, \bb)}^{+/2}}_2 \leq \epsilon
\end{align*}
where the last inequality follows from \Cref{lem:full_bipartite_approx} since we call \sparsebipartite() with accuracy $\epsilon/128$. This concludes the proof. 
\end{proof}

Finally, we use \Cref{lem:product} to prove \Cref{lem:squaring}.

\begin{proof}[Proof of \Cref{lem:squaring}]
By \Cref{lem:product} each individual call to \sparseproduct() conserves in- and out-degrees, and thus $\LL_{\tilde{\dir{G}}^2}$ and $\LL_{\dir{G}^2}$ are Eulerian Laplacians with the same in- and out-degrees. Further, by \Cref{lem:product}, we have that 
\begin{align*}
    \norm{\LL_{\mathcal{U}(\dir{G}_i)}^{+/2}(\LL_{\dir{G}_i} - \LL_{\dir{H}_i}) \LL_{\mathcal{U}(G_i)}^{+/2}}_2 \leq \epsilon
\end{align*}
By \Cref{lem:b2} we conclude that for all $\xx, \yy \neq \veczero$
\begin{align*}
    \xx^T(\LL_{\dir{G}_i} - \LL_{\dir{H}_i})\yy \leq \frac{\epsilon}{2} (\xx^T\LL_{\mathcal{U}(\dir{G}_i)}\xx + \yy^T \LL_{\mathcal{U}(\dir{G}_i)} \yy).
\end{align*}
Finally, using \Cref{lem:b2} again, we have for some vectors $\xx, \yy$
\begin{align*}
    \norm{\UU_{\LL_{\dir{G}^2}}^{+/2}(\LL_{\dir{G}^2} - \LL_{\tilde{\dir{G}}^2})\UU_{\LL_{\dir{G}^2}}^{+/2}}_2 &= \norm{\UU_{\LL_{\dir{G}}^2}^{+/2}(\sum_{i = 1}^n \LL_{\dir{G}_i} - \LL_{\dir{H}_i})\UU_{\LL_{\dir{G}^2}}^{+/2}}_2 \\
    &= 2 \frac{\xx^T(\sum_{i = 1}^n \LL_{\dir{G}_i} - \LL_{\dir{H}_i})\yy}{\xx^T \UU_{\LL_{\dir{G}^2}} \xx + \yy^T \UU_{\LL_{\dir{G}^2}} \yy} \\
    &\leq \epsilon \frac{\sum_{i = 1}^n \left( \xx^T\LL_{\mathcal{U}(\dir{G}_i)}\xx + \yy^T \LL_{\mathcal{U}(\dir{G}_i)} \yy \right)}{\xx^T \UU_{\LL_{\dir{G}^2}} \xx + \yy^T \UU_{\LL_{\dir{G}^2}} \yy}
    = \epsilon
\end{align*}
where we used that $\sum_{i = 1}^n \LL_{\mathcal{U}(\dir{G}_i)} = \UU_{\LL_{\dir{G}^2}}$. Clearly, this also shows that 

The runtime and sparsity directly follow from \Cref{lem:product}.
\end{proof}

\section{Deterministic Square Sparsifier Chains}
\label{sec:chains}

In this section we define the collection of $\tilde{O}(1)$ sparse matrices our algorithm operates on. This collection will be computed once, not recursively, and relies on square sparsifier chains.  

\begin{definition}[Square Sparsifier Chain, Definition 4.6 in \cite{cohen2016almostlineartime}]
    We call sequences of matrices $\mathcal{S} = \AAn_0, \AAn_1 ..., \AAn_d \in \R^{n \times n}$ a square sparsifier chain of length $d$ with parameter $0 < \alpha \leq \frac{1}{2}$ and error $\epsilon \leq \frac{1}{2}$ (or a $(d, \epsilon, \alpha)$-square chain for short) if under the definitions $\LL_i = \II - \AAn_i$ and $\AAn^{(\alpha)} = \alpha \II + (1 - \alpha)\AAn_i$ the following hold
    \begin{enumerate}
        \item $\norm{\AAn_i}_2 \leq 1$ for all $i$.
        \item $\II - \AAn_i$ is an $\epsilon$-approximation of $\II - (\AAn_i^{(\alpha)})^2$ for all $i \geq 1$.
        \item $\ker(\LL_i) = \ker(\LL_i^T) = \ker(\LL_j) = \ker(\LL_j^T) = \ker(\UU_{\LL_i}) = \ker(\UU_{\LL_j})$ for all $i, j$. 
    \end{enumerate}
    \label{def:sparifier_chain}
\end{definition}

Since we cannot afford to run our global sparsification routine after each squaring operation while still achieving an almost-linear runtime, we compute chains of depth $d = O((\log n)^{\delta})$ for $\delta = 1/3$ and then globally sparsify. These global sparsifications will naturally correspond to branching points of our recursive algorithm, since we have to rectify the loss due to global sparsification. We first define pseudoinverse chains. 

\begin{definition}[Pseudoinverse Chain]
\label{def:pseudoinverse_chain}
We call a collection $\mathcal{C} = \{\dir{G}_0^{(i)}, \dir{G}_1^{(i)}, \dir{G}_2^{(i)}, \dir{G}_3^{(i)}, \mathcal{S}^{(i)} \}_{i = 0}^{k}$ with $\mathcal{S}^{(i)} = \{\AAn^{(i)}_0, \AAn^{(i)}_1, ..., \AAn^{(i)}_d\}$ a $(k, d, \alpha, \epsilon, \gamma)$-pseudoinverse chain if
\begin{enumerate}
    \item $\mathcal{S}^{(i)}$ is a $(d, \epsilon, \alpha)$-square chain for all $i$. 
    \item $\dir{G}_0^{(i)}, \dir{G}_1^{(i)}, \dir{G}_2^{(i)}, \dir{G}_3^{(i)}$ is a $\gamma$-quadruple for all $i$. 
    \item $\frac{1}{1 + \frac{\beta}{\eta}} \DD_{\dir{G}}^{+/2} \AA^T_{\dir{G}_3^{(i)}} \DD_{\dir{G}}^{+/2} = \AAn_0^{(i)}$ for all $i$.
    \item $\DD_{\dir{G}}^{+/2} \AA^T_{\dir{G}_0^{(i+1)}} \DD_{\dir{G}}^{+/2} = \AAn_d^{(i)}$ for all $i < k$.
\end{enumerate}
where $\dir{G} = \dir{G}_0^{(0)}$ and $\frac{\eta}{\beta}$ is as in \Cref{alg:global_spars} and $\DD_{\dir{G}_0^{(i)}} = \DD_{\dir{G}}$ for all $i$.
\end{definition}

\begin{algorithm}
\caption{\chain($\dir{G}, k, d, \alpha, \epsilon, \gamma$) and \squarechain($\dir{\hat{G}}^{(i)}, d, \alpha, \epsilon$)}
\label{alg:square_chain}
\SetKwFunction{algo}{\chain}\SetKwFunction{proc}{\squarechain}
\SetKwProg{myalg}{Algorithm}{}{}
\SetKwProg{myproc}{Procedure}{}{}
\myalg{\algo{$\dir{G}_0^{(0)} = \dir{G}, k, d, \alpha, \epsilon, \gamma$}}{
\For{$i = 0, ..., k$}{
    $\dir{G}_1^{(i)}, \dir{G}_2^{(i)}, \dir{G}_3^{(i)} = \sparsify(\dir{G}_0^{(i)}, \gamma)$ \tcp*{See \Cref{alg:global_spars}.} 
    \tcp{Scale back degrees by $\frac{\eta}{\beta}$ before building the next chain}
    $\dir{G}_0^{(i + 1)}, \mathcal{S}^{(i)} = \squarechain(\frac{1}{1 + \frac{\beta}{\eta}}  \dir{G}^{(i)}_3, d, \alpha, \epsilon)$
}
\Return{$\mathcal{C} = \{\dir{G}_0^{(i)}, \dir{G}_1^{(i)}, \dir{G}_2^{(i)}, \dir{G}_3^{(i)}, \mathcal{S}^{(i)} \}_{i = 0}^{k}$} 
}

\myproc{\proc{$\dir{G}^{(i,0)} = \frac{1}{1 + \frac{\beta}{\eta}} \dir{G}^{(i)}_3, d, \alpha, \epsilon$}}{
\For{$j = 1, ..., d $}{
    Let $G_D$ denote the graph with adjacency matrix $\DD_{\dir{G}}$, only consisting of self-loops.\\
    $\dir{G}^{(i,j)} = \sparsesquare((1 - \alpha)\dir{G}^{(i,j - 1)} + \alpha G_D, \epsilon)$ \tcp*{See \Cref{alg:sparsesquare}.}
}
\tcp{We use $\DD_{\dir{G}} = \DD_{\dir{G}^{(i, 0)}} = \DD_{\dir{G}^{(i, j)}}$ here since we preserve degrees.}
\Return{$\dir{G}^{(i + 1)}_0 = \dir{G}^{(i, d)}, \mathcal{S}^{(i)} = \{\AAn^{(i)}_j = \DD_{\dir{G}}^{+/2}\AA_{\dir{G}^{(i,j)}}^T\DD_{\dir{G}}^{+/2}\}_{j = 0}^d$}
}
\end{algorithm}

We then state some previous work that will be useful for showing that \Cref{alg:square_chain} creates a pseudoinverse chain. The following lemma shows that $\epsilon$-approximations behave well under diagonal rescalings.

\begin{lemma}[Lemma 3.7 in \cite{cohen2016almostlineartime}]
\label{lem:37}
If $\tilde{\AA} \in \R^{n \times n}$ is an $\epsilon$-approximation of $\AA \in \R^{n \times n}$ and $\MM \in \R^{n \times n}$ satisfies $\ker(\MM^T) \subseteq \ker(\UU_A)$, then $\MM^T \tilde{\AA} \MM$ is an $\epsilon$-approximation of $\MM^T \AA \MM$.
\end{lemma}

And the next lemma will be used to show the first point in \Cref{def:sparifier_chain}. 

\begin{lemma}[Lemma 4.7 in \cite{cohen2016almostlineartime}]
If for $\AA \in \R^{n \times n}$ and $\DD = \diag(\AA\vecone)$ the matrix $\LL = \DD - \AA^T$ is an Eulerian Laplacian associated with a strongly connected graph then, $\norm{\DD^{-1/2} \AA^T \DD^{-1/2}}_2 \leq 1$ and $\ker(\LL) = \ker(\LL^T) = \ker(\UU_{\LL})$
\label{lem:lapprop47}
\end{lemma}

Given these results, we first show a lemma about the subroutine \squarechain() in \Cref{alg:square_chain}. 

\begin{lemma}
Given a strongly connected Eulerian graph $\frac{1}{1 + \frac{\beta}{\eta}} \dir{G}^{(i)}_3$ with $m = \tilde{O}(n)$ edges and degree matrix $\DD_{\dir{G}}$, as well as parameters $\alpha \in (0,1)$ and $\delta \in (0, 1/2)$ the subroutine $\dir{G}^{(i + 1)}_0, \mathcal{S}^{(i)} = \squarechain\left(\frac{1}{1 + \frac{\beta}{\eta}} \dir{G}^{(i)}_3, d = \Theta((\log n)^\delta), \alpha, \epsilon = \frac{1}{30 \cdot \Exp(5d)} \right)$ computes a $(d, \epsilon, \alpha)$-square chain $\mathcal{S}^{(i)}$ and a graph $\dir{G}^{(i + 1)}_0$ with degree matrix $\DD_{\dir{G}}$ so that 
\begin{enumerate}
    \item $\nnz(\AAn_j^{(i)}) = n^{1 + o(1)}$ for $j = 0, ..., d$.
    \item $\DD_{\dir{G}}^{+/2} \AA^T_{\dir{G}^{(i,0)}} \DD_{\dir{G}}^{+/2} = \AAn_0^{(i)}$ and $\frac{1}{1 + \frac{\beta}{\eta}} \DD_{\dir{G}}^{+/2} \AA^T_{\dir{G}^{(i)}_3} \DD_{\dir{G}}^{+/2} = \AAn_d^{(i)}$ 
\end{enumerate}
\label{lem:chain}
\end{lemma}
\begin{proof}
Notice that $\DD_{\dir{G}} = \DD_{\dir{G}^{(i,0)}} =  \DD_{\dir{G}^{(i,d)}}$ since degrees are exactly preserved by \sparsesquare(). We first show that $\mathcal{S}^{(i)}$ is a $(d, \epsilon, \alpha)$-square chain. Since the graph $(1 - \alpha)\dir{G}^{(i,j - 1)} + \alpha G_D$ has the adjacency matrix $(1 - \alpha)\AA_{\dir{G}^{(i,j - 1)}} + \alpha \DD_{\dir{G}^{(i,j - 1)}}$, its square contains a multiple of $\dir{G}^{(i,j - 1)}$ as a subgraph and is therefore strongly connected, if $\dir{G}^{(i,j - 1)}$ is strongly connected. Then by \Cref{lem:squaring} $\dir{G}^{(i,j)}$ is strongly connected. Together with \Cref{lem:37} and \Cref{lem:lapprop47} this yields the properties of a $(d, \epsilon, \alpha)$-square chain by induction. It remains to show sparsity. 

Initially, we have that $\nnz(\AAn^{(i)}_0) = \tilde{O}(n)$. By \Cref{lem:squaring} we have 
\begin{align*}
    \nnz(\AAn^{(i)}_{j + 1}) \leq C \nnz(\AAn^{(i)}_{j})/\epsilon^4 =  30C \cdot \Exp(20 (\log n)^\delta) \cdot \nnz(\AAn^{(i)}_{j})
\end{align*} 
for some constant $C$. By induction, we conclude
\begin{align*}
\nnz(\AAn^{(i)}_{j}) \leq (30C)^j \cdot \Exp(20j(\log n)^{\delta}) \cdot \nnz(\AAn^{(i)}_0)  \leq \Exp(O((\log n)^{2\delta})) \cdot \tilde{O}(n) = n^{1 + o(1)}
\end{align*}
where we used that $\delta < 1/2$ is a constant.
\end{proof}

Next we show that \chain() in \Cref{alg:square_chain} deterministically computes a pseudoinverse chain in almost linear time, using the machinery we developed. 

\begin{lemma}
For constants $\delta, \gamma \in (0, 1/2)$, $\gamma < \delta$, $k \geq 1$ , $\alpha = 1/4$ and a strongly connected Eulerian graph $\dir{G}$ with $m$ edges the algorithm $\mathcal{C} = \chain(\dir{G}, k, d = O((\log n)^{\delta}), \alpha, \epsilon = 1/(30 \cdot \Exp(5d)), \gamma)$ computes a $(k, d, \alpha, \epsilon, \gamma)$-pseudoinverse chain $\mathcal{C}$ in time $m^{1 + o(1)}$ so that 
\begin{enumerate}
    \item $|E(\dir{G}^{(0)}_j)| \leq 2m + \tilde{O}(n)$ for $j = 1, 2$, $|E(\dir{G}^{(0)}_3)| = \tilde{O}(n)$. 
    \item $|E(\dir{G}^{(i)}_j)| = n^{1+o(1)}$ for $i > 0$, $j = 0, 1, 2, 3$. 
    \item $\nnz(\AAn^{(i)}_{j}) = n^{1+o(1)}$ for all $i, j$. 
\end{enumerate}
\label{lem:pseudoinverse_chain}
\end{lemma}
\begin{proof}
Follows directly from \Cref{lem:chain} and \Cref{lem:globalspars}. 
\end{proof}

Before concluding with a lemma about the improvement in condition number of a pseudoinverse chain, we show a convenient corollary about the quadruplets introduced in \Cref{sec:globspars} and state a lemma about the improvement in condition number due to squaring. 

\begin{corollary}
Given a $(\gamma, \beta = \Exp(O((\log n)^{\gamma})), \eta = \Exp(-3(\log n)^{\gamma}))$-quadruple $\dir{G}_0, \dir{G}_1, \dir{G}_2, \dir{G}_3$ for $\gamma \in (0,1)$ constant, we have 
\begin{align*}
   \Exp(-O((\log n)^{\gamma})) \lambda_*(\DD_{\dir{G}_0}^{+/2}\LL_{\dir{G}_0}\DD_{\dir{G}_0}^{+/2}) \leq  \frac{\eta}{\beta} \lambda_*(\DD_{\dir{G}_0}^{+/2}\LL_{\dir{G}_3}\DD_{\dir{G}_0}^{+/2})
\end{align*}
\label{lem:quadruple_loss}
\end{corollary}
\begin{proof}
By the second property of \Cref{def:quadruple} we have 
\begin{align*}
    \Exp(-O((\log n)^{\gamma})) \DD_{\dir{G}_0}^{+/2}\LL_{\dir{G}_0}\DD_{\dir{G}_0}^{+/2} \preceq \DD_{\dir{G}_0}^{+/2}\LL_{\dir{G}_3}\DD_{\dir{G}_0}^{+/2}.
\end{align*}
Since $\frac{\eta}{\beta} \geq \Exp(-O((\log n)^{\gamma}))$ we conclude the lemma. 
\end{proof}

\begin{lemma}[Lemma 4.9 in \cite{cohen2016almostlineartime}]
\label{lem:chain_prop49}
For length $d \geq 1$, parameter $\alpha = 1/4$ and error $\epsilon \in (1, 1/2)$ and a $(d, \alpha, \epsilon)$-chain $\AAn_0, \AAn_1, ..., \AAn_d$ the following hold:
\begin{enumerate}
    \item $\kappa(\II - \AAn_i, \II = \AAn_{i - 1}) \leq 21$ for all $i = 1, ..., d$.
    \item $\lambda_*(\II - \AA_d) \geq \min\{1/4,\lambda_*(\II - \AA_0) \cdot ((1 - \epsilon)1.25)^d\}$.
\end{enumerate}
\end{lemma}

Finally, we show that a deep enough pseudoinverse chain reaches a constant condition number. 

\begin{lemma}
Given $\gamma, \delta, \epsilon, \alpha, \mathcal{C}$ and $\dir{G}$ as in \Cref{lem:pseudoinverse_chain} where we set $k = \Theta(\log(1/\hat{\lambda}_*)/(\log n)^\delta)$ for $1/\mathit{poly}(n) \leq \hat{\lambda}_* \leq \lambda_*(\DD_{\dir{G}}^{+/2} \UU_{\LL_{\dir{G}}} \DD_{\dir{G}}^{+/2})$ we have $\lambda_*(\DD_{\dir{G}}^{+/2} \UU_{\LL_{\dir{G}_0^{(k)}}} \DD_{\dir{G}}^{+/2}) \geq 1/4$
\label{lem:eigvals_const_lower}
\end{lemma}
\begin{proof}
By \Cref{lem:quadruple_loss} we have for every $i = 0, 1, ..., k$
\begin{align*}
    \Exp(-O((\log n)\gamma)) \cdot \lambda_*(\DD_{\dir{G}}^{+/2} \LL_{\dir{G}^{(i)}_0} \DD_{\dir{G}}^{+/2}) \leq \frac{\eta}{\beta} \cdot \lambda_*(\DD_{\dir{G}}^{+/2} \LL_{\dir{G}^{(i)}_3} \DD_{\dir{G}}^{+/2})
\end{align*}
and by \Cref{lem:chain_prop49} we have 
\begin{align*}
    \min\{ 1.125^d \lambda_*(\frac{\eta}{\beta}\DD_{\dir{G}}^{+/2} \LL_{\dir{G}^{(i)}_3} \DD_{\dir{G}}^{+/2}), 1/4\} \leq \lambda_*(\DD_{\dir{G}}^{+/2} \LL_{\dir{G}^{(i + 1)}_0} \DD_{\dir{G}}^{+/2}).
\end{align*}
Therefore each of the $k$ chains improves the minimum eigenvalue by $\frac{1.125^d}{\Exp(O((\log n)^\gamma))} = \Exp(\Omega((\log n)^\delta))$ for $\delta > \lambda$ and $d = \Theta((\log n)^\delta)$ big enough. After $k = \Theta(\log(\hat{\lambda}_*)/d)$ iterations this ensures the threshold $1/4$ is reached.  
\end{proof}

\section{A Deterministic Squaring Solver}
\label{sec:finalSec}

In this section, we build our recursive squaring solver on top of the pseudoinverse chain constructed in the previous section. It is inspired by and closely follows the solver presented in \cite{cohen2016almostlineartime}. In our algorithm, the global sparsifications of the pseudoinverse chain correspond to branching points of the recursion. This is necessary since we incur relatively high sub-polynomial losses when globally sparsifying. We state the main theorem of this section. 

\begin{theorem}[Deterministic Eulerian Solver]
\label{thm:eul_square_solver}
Given the Laplacian $\LL_{\dir{G}}$ of a strongly connected Eulerian graph $\dir{G}$ with $m$ edges and polynomially bounded edge weights, a parameter $\epsilon \in (0,1)$ and a lower bound $\hat{\lambda}_*$ on $\lambda_*(\DD_{\dir{G}}^{+/2} \LL_{\dir{G}} \DD_{\dir{G}}^{+/2})$ so that $\hat{\lambda}_* \geq 1/\mathit{poly}(n)$ the algorithm \solveEulerian($\LL_{\dir{G}}, \epsilon, \hat{\lambda}_*$) (\Cref{alg:solve_Eulerian}) deterministically computes a vector $\xx$ so that 
\begin{align*}
    \norm{\xx - \LL_{\dir{G}}^+ \bb}_{\UU_{\LL_{\dir{G}}}} \leq \epsilon \norm{\LL_{\dir{G}}^+ \bb}_{\UU_{\LL_{\dir{G}}}}
\end{align*}
in time $m^{1 + o(1)} \cdot \log \frac{1}{\epsilon}$ for a vector $\bb \in \im(\LL_{\dir{G}})$.
\end{theorem}

The following lemma re-interprets preconditioned Richardson in terms of augmenting the quality of a pseudoinverse. 

\begin{lemma}[Pseudoinverse Improvement, Lemma 4.4 in \cite{cohen2016almostlineartime}]
If $\ZZ$ is an $\epsilon$-approximate-pseudoinverse of $\MM$ with respect to $\UU$, for $\epsilon \in (0, 1)$, $\bb \in \im(\MM)$ and $N \geq 0$, then $\xx_N = \textsc{PreconRichardson}(\MM, \ZZ, \bb, 1, N)$ computes $\xx_N = \ZZ_N \bb$ for some matrix $\ZZ_N$ only depending on $\ZZ$, $\MM$ and $N$, such that $\ZZ_N$ is an $\epsilon^N$-approximate pseudoinverse of $\MM$ with respect to $\UU$. 
\label{lem:pseudoinverse_improvement}
\end{lemma} 

Our algorithm frequently uses preconditioned Richardson to boost the quality of pseudoinverses. Following \cite{cohen2016almostlineartime} we let $\ZZ_N = \richardson(\MM, \ZZ, \cdot ,\eta, N)$ denote the implicit matrix $\ZZ_N$ as in \Cref{lem:pseudoinverse_improvement}. To apply it to a vector $\bb$, we need to apply the (possibly implicit) matrices $\MM$ and $\ZZ$ a total of $N$ times. Before we show the main technical lemma of this section, we reference previous work that we use in the proof. We first state a lemma about solving well conditioned linear equations with Richardson which we will use to establish the base case of our recursion

\begin{lemma}[Lemma 4.5 in \cite{cohen2016almostlineartime}]
Let $\MM \in \R^{n \times n}$ such that $\UU_{\MM}$ is PSD with $\ker(\MM) = \ker(\MM^T)$. Then for parameters $\eta \leq \lambda_*(\UU_{\MM})/\norm{M}_2^2$ and $N > 0$ we have that the resulting matrix $\ZZ_N = \richardson(\MM, \eta \II_{\im(\MM)}, \cdot, 1, N)$ is an $\Exp(-N\eta\lambda_*(\UU_{\MM})/2)$-approximate pseudoinverse. 
\label{lem:buildpred45}
\end{lemma}

Then we state a lemma that formalises obtaining preconditioners via squaring as introduced in the overview.

\begin{lemma}[Lemma 4.13 in \cite{cohen2016almostlineartime}]
Let the sequence $\AAn_0, ..., \AAn_d$ be an $(d, \epsilon, \alpha)$-chain as specified in \Cref{def:sparifier_chain}, with $\epsilon \leq 1/2$ and $\alpha \leq 1/4$. Using the notation from \Cref{def:sparifier_chain}, consider the matrix 
\begin{align*}
    \bar{\ZZ}_{i, i + \Delta} = (1 - \alpha)^{\Delta} \left( \II - \AAn_{i + \Delta} \right)^+ \left(\II + \AAn_{i + \Delta - 1} \right) \cdots \left(\II + \AAn_{i} \right),
\end{align*}
for any i, $\Delta \geq 0$. Then $\bar{\ZZ}_{i, i + \Delta}$ is an $(\Exp(5 \Delta) \cdot \epsilon)$-approximate pseudoinverse of $\II - \AAn_i$ with respect to $\II - \UU_{\AAn_i}$.
\label{lem:squared_pseudoinverse}
\end{lemma}

Finally, we need two lemmas about composing approximate pseudoinverses and changing norms. 

\begin{lemma}[Lemma 4.10 in \cite{cohen2016almostlineartime}]
If matrix $\ZZ$ is an $\epsilon$-approximate pseudoinverse of $\MM$ with respect to $\UU$, and $\widetilde{\MM}^+$ is an $\epsilon'$-approximate pseudoinverse of $\ZZ^+$ with respect to $\UU$, and has the same left and right kernels as $\MM$ and $\ZZ$, then $\widetilde{\MM}^+$ is an $(\epsilon + \epsilon' + \epsilon \epsilon')$-approximate pseudoinverse of $\MM$ with respect to $\UU$.
\label{lem:comp410}
\end{lemma}

\begin{lemma}[Lemma 4.12 in \cite{cohen2016almostlineartime}]
\label{lem:approxpseudprop412}
Let $\ZZ, \MM, \UU \in \R^{n \times n}$ be matrices such that $\UU$ is PSD, and $\ker(\ZZ) = \ker(\ZZ^T) = \ker(\MM) = \ker(\MM^T) \supseteq \ker(\UU)$. Then the following holds:
\begin{enumerate}
    \item (Preserved under right multiplication) Let $\CC \in \R^{n \times n}$ such that both $\CC$ and $\CC^T$ are invariant on the kernel of $\MM$ in the sense that $\xx \in \ker(\MM)$ if and only if $\CC \xx \in \ker(\MM)$, and similarly for $\CC^{\perp}$. Then $\ZZ$ is an $\epsilon$-approximate pseudoinverse for $\CC \MM$ with respect to $\UU$ if and only if $\ZZ \CC$ is an $\epsilon$-approximate pseudoinverse  for $\MM$ with respect to $\UU$.
    \item (Approximately preserved under norm change) If $\ZZ$ is an $\epsilon$-approximate pseudoinverse for $\MM$ with respect to $\UU$, then for any PSD matrix $\widetilde{\UU}$, such that $\ker(\widetilde{\UU}) = \ker(\UU)$, $\ZZ$ is an $(\epsilon \sqrt{\kappa(\widetilde{\UU}, \UU)})$-approximate pseudoinverse of $\MM$ with respect to $\tilde{\UU}$.
\end{enumerate}
\end{lemma}

\begin{algorithm}
\caption{\solve($\mathcal{C} = \{\dir{G}_0^{(i)}, \dir{G}_1^{(i)}, \dir{G}_2^{(i)}, \dir{G}_3^{(i)}, \mathcal{S}^{(i)}\}_{i = l }^{k}, \epsilon$) and \peel()}
\label{alg:solve}
\tcp{We let $\dir{G}_0^{(l)} = \dir{G}$ in these algorithms.}
\SetKwFunction{algo}{\solve}\SetKwFunction{proc}{\peel}
\SetKwProg{myalg}{Algorithm}{}{}
\SetKwProg{myproc}{Procedure}{}{}
\myalg{\algo{$\mathcal{C} = \{\dir{G}_0^{(i)}, \dir{G}_1^{(i)}, \dir{G}_2^{(i)}, \dir{G}_3^{(i)}, \mathcal{S}^{(i)}\}_{i = l }^{k}, \epsilon$}}{
\uIf{$l = k$}{
    \tcp{Leaf of Recursion}
    $\AAn = \DD_{\dir{G}}^{+/2}\AA^T_{\dir{G}_0^{(k)}}\DD_{\dir{G}}^{+/2}$ \\
    $\ZZ_0^{(k)} = \richardson(\II - \AAn, \frac{1}{16}\II_{\im(\II - \AAn)}, \cdot, 1, 256 \log(1/\epsilon))$ \\
}\Else{
    $\ZZ_0^{(l)} = \peel(\{\dir{G}_j^{(l)}\}_{j = 0}^3, \mathcal{S}^{(l)}, \mathcal{C}' = \{\dir{G}_0^{(i)}, \dir{G}_1^{(i)}, \dir{G}_2^{(i)}, \dir{G}_3^{(i)}, \mathcal{S}^{(i)}\}_{i = l+1}^{k}, \epsilon)$ \\
}
\Return{$\ZZ_0^{(l)}$}
}

\myproc{\proc{$\{\dir{G}_j^{(l)}\}_{j = p}^3, \mathcal{S}^{(l)}, \mathcal{C}' = \{\dir{G}_0^{(i)}, \dir{G}_1^{(i)}, \dir{G}_2^{(i)}, \dir{G}_3^{(i)}, \mathcal{S}^{(i)}\}_{i = l+1}^{k}, \epsilon$}}{
\uIf{$p = 3$}{
\tcp{We have $\mathcal{S}^{(l)} = \AAn_0^{(l)}, ..., \AAn_d^{(l)}$ and we let $\AAn_t^{(l, 1/4)} = \frac{3}{4}\AAn_t^{(l)} + \frac{1}{4}\II $}
    $\ZZ^{(l + 1)}_0 = \solve(\mathcal{C}', \frac{1}{30 \cdot \Exp(5d)})$ \\
    $\ZZ = (3/4)^{d + 1}\ZZ^{(l + 1)}_0 (\II + \AAn_{d-1}^{(l, 1/4)}) \cdots (\II + \AAn_{0}^{(l, 1/4)})$\\ 
    $\ZZ^{(l)}_{3} = \richardson(\DD_{\dir{G}}^{+/2}\LL_{\dir{G}_p^{(l)}}\DD_{\dir{G}}^{+/2}, (1 + \frac{\beta}{\eta})\ZZ,\cdot, 1, \log(1/\epsilon))$
}\Else{
    $\ZZ^{(l)}_{p+1} = \peel(\{\dir{G}_j^{(l)}\}_{j = p + 1}^3, \mathcal{S}^{(l)}, \mathcal{C}', \frac{1}{\Exp(\Theta((\log n)^{\gamma}))})$ \\
    $\ZZ^{(l)}_{p} = \richardson(\DD_{\dir{G}}^{+/2}\LL_{\dir{G}_p^{(l)}}\DD_{\dir{G}}^{+/2}, \ZZ^{(l)}_{p+1},\cdot, 1, \Exp(\Theta((\log n)^{\gamma})) \cdot \log(1 / \epsilon))$
}
\Return{$\ZZ^{(l)}_{p}$}
}
\end{algorithm}

These results, together with the properties of pseudoinverse chains, allow us to show the following main technical lemma of this section, which establishes correctness of the algorithm \solve() (\Cref{alg:solve}). 

\begin{lemma}
Given the tail $\mathcal{C} = \{\dir{G}_0^{(i)}, \dir{G}_1^{(i)}, \dir{G}_2^{(i)}, \dir{G}_3^{(i)}, \mathcal{S}^{(i)}\}_{i = l }^{k}$ of a $(k , d = \Theta((\log n)^{1/3}), \alpha = 1/4, \epsilon_0 = \frac{1}{30 \Exp(5d)}, \gamma = 1/10)$-pseudoinverse chain as in \Cref{lem:eigvals_const_lower}, and a parameter $\epsilon \in (0, 1/2)$ the algorithm \solve($\mathcal{C}, \epsilon$) computes an $\epsilon$-approximate pseudoinverse $\ZZ_0^{(l)}$ of $\DD^{+/2}_{\dir{G}} \LL_{\dir{G}} \DD^{+/2}_{\dir{G}}$ with respect to $\DD^{+/2}_{\dir{G}} \UU_{\LL_{\dir{G}}} \DD^{+/2}_{\dir{G}}$ where $\DD_{\dir{G}} = \DD_{\dir{G}_0^{(l)}}$.
\label{lem:rec_pseudo}
\end{lemma}
\begin{proof}
We use that $\DD_{\dir{G}} = \DD_{\dir{G}_0^{(i)}}$ for all $i$ by \Cref{def:pseudoinverse_chain}. The proof is by induction on $k - l$. 
\begin{itemize}
    \item \emph{Base case:} $k - l = 0$. We land in the if case of the algorithm \solve() and denote $\II - \AAn = \DD_{\dir{G}}^{+/2} \LL_{\dir{G}_0^{(k)}} \DD_{\dir{G}}^{+/2}$ as in the algorithm. Then we have $\lambda_*(\II - \UU_{\AAn}) \geq 1/4$ by \Cref{lem:eigvals_const_lower}, which is exactly what we set out to achieve with squaring. By \Cref{lem:lapprop47} we further have $\norm{\II - \AAn}_2^2 \leq 4$. Then the base case follows directly from \Cref{lem:buildpred45}. 
    \item \emph{Step case:} $k - l > 0$. We land in the else case of the algorithm \solve() and let $\mathcal{C}' = \{\dir{G}_0^{(i)}, \dir{G}_1^{(i)}, \dir{G}_2^{(i)}, \dir{G}_3^{(i)}, \mathcal{S}^{(i)}\}_{i = l + 1}^{k}$. Then the induction hypothesis is that $\solve(\mathcal{C}', \epsilon)$ returns an $\epsilon$-approximate pseudoinverse $\ZZ_0^{(l + 1)}$ of $\DD^{+/2}_{\dir{G}} \LL_{\dir{G}_0^{(l+1)}} \DD^{+/2}_{\dir{G}}$ with respect to $\DD^{+/2}_{\dir{G}} \UU_{\LL_{\dir{G}_0^{(l+1)}}} \DD^{+/2}_{\dir{G}}$. To establish the step case we rely on a second, nested induction of constant depth, showing the correctness of the subroutine \peel(). Namely, we show that $\ZZ_p^{(l)} = \peel(\{\dir{G}_j^{(l)}\}_{j = p}^3, \mathcal{S}^{(l)}, \mathcal{C}', \epsilon)$ returns an $\epsilon$-approximate pseudoinverse $\ZZ_p^{(l)}$ of $\DD^{+/2}_{\dir{G}} \LL_{\dir{G}_p^{(l)}} \DD^{+/2}_{\dir{G}}$ with respect to $\DD^{+/2}_{\dir{G}} \UU_{\LL_{\dir{G}_p^{(l)}}}\DD^{+/2}_{\dir{G}}$ for $p = 0, 1, 2, 3$. The nested induction is on $3 - p$. 
    \begin{itemize}
        \item \emph{Nested base case:} $3 - p = 0$. We land in the if case of the procedure \peel(). Noting that $\II - \AAn_d^{(l)} = \DD^{+/2}_{\dir{G}} \LL_{\dir{G}_0^{(l+1)}} \DD^{+/2}_{\dir{G}}$ by \Cref{def:pseudoinverse_chain}, we have that $\ZZ^{(l + 1)}_0$ is an $\frac{1}{30 \cdot \Exp(5d)}$-approximate pseudoinverse of $\II - \AAn_d^{(l)}$ with respect to $\II - \UU_{\AAn_d^{(l)}}$ by the induction hypothesis of the main induction. Then, by \Cref{lem:squared_pseudoinverse} $\ZZ$ is an $1/30$ approximate pseudoinverse of $\II - \AAn_0^{(l)}$ with respect to $\II - \UU_{\AAn_0^{(l)}}$. Recall that $\II - \AAn_0^{(l)} = (1 + \frac{\beta}{\eta})^{-1} \DD^{+/2}_{\dir{G}} \LL_{\dir{G}_3^{(l)}} \DD^{+/2}_{\dir{G}}$. By \Cref{lem:pseudoinverse_improvement} the matrix $\ZZ_3^{(l)} = \peel(\{\dir{G}_3^{(l)}\}_{j = p}^3, \mathcal{S}^{(l)}, \mathcal{C}', \epsilon)$ is an $\epsilon$-approximate pseudoinverse of $\DD^{+/2}_{\dir{G}} \LL_{\dir{G}_3^{(l)}} \DD^{+/2}_{\dir{G}}$ with respect to $\DD^{+/2}_{\dir{G}} \UU_{\LL_{\dir{G}_3^{(l)}}}\DD^{+/2}_{\dir{G}}$ since the factor $(1 + \frac{\beta}{\eta})$ cancels. This establishes the base case.  
        \item \emph{Nested step case:} $p < 3$. We land in the else case of the procedure \peel(). By the induction hypothesis $\ZZ_{p+1}^{(l)}$ is a $\left(\epsilon_1 = \frac{1}{\Exp(\Theta((\log n)^{\gamma}))}\right)$-approximate pseudoinverse of $\DD^{+/2}_{\dir{G}} \LL_{\dir{G}_{p + 1}^{(l)}} \DD^{+/2}_{\dir{G}}$ with respect to $\DD^{+/2}_{\dir{G}} \UU_{\LL_{\dir{G}_{p + 1}^{(l)}}}\DD^{+/2}_{\dir{G}}$. Further, we have that $\DD_{\dir{G}}^{1/2}\LL^+_{\dir{G}_{p+1}^{(l)}}\DD_{\dir{G}}^{1/2}$ is an $\left(\epsilon_2 = 1 - \frac{1}{\Exp(O((\log n)^{\gamma}))} \right)$-approximate pseudoinverse of $\DD_{\dir{G}}^{1/2}\LL^+_{\dir{G}_{p}^{(l)}}\DD_{\dir{G}}^{1/2}$ with respect to $\DD_{\dir{G}}^{1/2}\UU_{\LL^+_{\dir{G}_{p + 1}^{(l)}}}\DD_{\dir{G}}^{1/2}$ by \Cref{def:quadruple}. Combining the two with \Cref{lem:comp410} yields that $\ZZ_{p + 1}^{(l)}$ is an $\left(\epsilon_1 + \epsilon_2 + \epsilon_1 \epsilon_2 = 1 - \frac{1}{\Exp(O((\log n)^{\gamma}))}\right)$-approximate pseudoinverse of $\DD_{\dir{G}}^{1/2}\LL^+_{\dir{G}_{p}^{(l)}}\DD_{\dir{G}}^{1/2}$ with respect to $\DD_{\dir{G}}^{1/2}\UU_{\LL^+_{\dir{G}_{p + 1}^{(l)}}}\DD_{\dir{G}}^{1/2}$, where we choose $\epsilon_{1}$ small enough. By \Cref{lem:pseudoinverse_improvement}, $\Exp(\Theta((\log n)^{\gamma}) \log(1/\epsilon))$ iterations of preconditioned Richardson suffice to reduce the error to $\epsilon$. However, the resulting approximate pseudoinverse is with respect to $\DD_{\dir{G}}^{1/2}\UU_{\LL^+_{\dir{G}_{p + 1}^{(l)}}}\DD_{\dir{G}}^{1/2}$ instead of $\DD_{\dir{G}}^{1/2}\UU_{\LL^+_{\dir{G}_{p}^{(l)}}}\DD_{\dir{G}}^{1/2}$. By the second point in \Cref{def:quadruple} and \Cref{lem:approxpseudprop412} another factor of $\tilde{O}(1)$ in the iteration count suffices to translate between norms. This factor can be subsumed in the iteration count $\Exp(\Theta((\log n)^{\gamma}) \log(1/\epsilon))$. Therefore, $\ZZ_{p}^{(l)}$ is an $\epsilon$-approximate pseudoinverse of $\DD_{\dir{G}}^{1/2}\LL^+_{\dir{G}_{p}^{(l)}}\DD_{\dir{G}}^{1/2}$ with respect to $\DD_{\dir{G}}^{1/2}\UU_{\LL^+_{\dir{G}_{p}^{(l)}}}\DD_{\dir{G}}^{1/2}$, which concludes the step case. 
    \end{itemize}
    The nested induction establishes that $\ZZ^{(l)}_0$ is an $\epsilon$-approximate pseudoinverse of $\DD^{+/2}_{\dir{G}} \LL_{\dir{G}_0^{(l)}} \DD^{+/2}_{\dir{G}}$ with respect to $\DD^{+/2}_{\dir{G}} \UU_{\LL_{\dir{G}_0^{(l)}}}\DD^{+/2}_{\dir{G}}$. This concludes the step case, since \solve($\mathcal{C} = \{\dir{G}_0^{(i)}, \dir{G}_1^{(i)}, \dir{G}_2^{(i)}, \dir{G}_3^{(i)}, \mathcal{S}^{(i)}\}_{i = l }^{k}, \epsilon$) returns $\ZZ^{(l)}_0$. The lemma follows by induction. 
\end{itemize}
\end{proof}

\begin{algorithm}
\caption{\solveEulerian($\LL_{\dir{G}}, \bb, \epsilon, \hat{\lambda}_*$)}
\label{alg:solve_Eulerian}
$d = \Theta((\log n)^{1/3})$; $k = C \log (1/\hat{\lambda}_*)/d \in O((\log n)^{2/3})$; $\epsilon_0 = \frac{1}{30 \cdot \Exp(5d)}$\\
$\mathcal{C} = \chain(\dir{G}, k, d, \delta = 1/3, \epsilon_0, \gamma = 1/10)$ \\
$\ZZ = \solve(\mathcal{C}, 1/10)$ \\
$\hat{\ZZ} = \richardson(\DD_{\dir{G}}^{+/2} \LL_{\dir{G}} \DD_{\dir{G}}^{+/2} , \ZZ, \cdot, 1, \log(1/\epsilon))$ \\
\Return{$\DD_{\dir{G}}^{+/2} \hat{\ZZ} \DD_{\dir{G}}^{+/2} \bb$}
\end{algorithm}

Before we conclude with the main theorem of this section, we analyse the runtime of \solveEulerian($\LL_{\dir{G}}, \bb, \epsilon, \hat{\lambda}_*$).

\begin{lemma}
\label{lem:runtime_solver}
Given $\hat{\lambda}_* \geq 1/poly(n)$, the Laplacian of an Eulerian graph $\dir{G}$ with $n$-vertices and $m$-edges, a vector $\bb \in \R^n$ so that $\bb \in \im(\LL_{\dir{G}})$ and a parameter $\epsilon \in (0,1)$, \solveEulerian($\LL_{\dir{G}}, \bb, \epsilon, \hat{\lambda}_*$) runs in time $m^{1 + o(1)} \log(1/\epsilon)$.  
\end{lemma}
\begin{proof}
The pseudoinverse chain is built in $m^{1+o(1)}$ time by \Cref{lem:pseudoinverse_chain}. By \Cref{def:pseudoinverse_chain} all matrices involved have at most $m^{1+o(1)}$ edges, and Richardson only ever multiplies with these matrices. Therefore it suffices to show that the total branching of our recursion \solve() is $n^{o(1)}$. There are $4$ branching points for each level $l \in [k]$ for $k = \Theta((\log n)^{2/3})$ in the procedure $\peel$() in \Cref{alg:solve}. Each of these branches by a factor at most $f = \Exp(\Theta((\log n)^{1/10}))$. Therefore the total branching is at most $f^k \leq O(1) \cdot \Exp((\log n)^{4/5}) = n^{o(1)}$. Finally \solveEulerian() contributes another factor of $O(\log(1/\epsilon))$. This concludes the proof.
\end{proof}

Finally, we put the pieces together and show \Cref{thm:eul_square_solver}.

\begin{proof}[Proof of \Cref{thm:eul_square_solver}]
The runtime follows from \Cref{lem:runtime_solver}. $\mathcal{C}$ is a pseudoinverse chain by \Cref{lem:pseudoinverse_chain}. Then, by \Cref{lem:rec_pseudo} we have that $\ZZ$ is a $1/10$ approximate pseudoinverse of $\DD^{+/2}_{\dir{G}} \LL_{\dir{G}} \DD^{+/2}_{\dir{G}}$ with respect to $\DD^{+/2}_{\dir{G}} \UU_{\LL_{\dir{G}}} \DD^{+/2}_{\dir{G}}$. Therefore $\hat{\ZZ}$ in \solveEulerian() is an $\epsilon$-approximate pseudoinverse of $\DD^{+/2}_{\dir{G}} \LL_{\dir{G}} \DD^{+/2}_{\dir{G}}$ with respect to $\DD^{+/2}_{\dir{G}} \UU_{\LL_{\dir{G}}} \DD^{+/2}_{\dir{G}}$ by \Cref{lem:pseudoinverse_improvement}. The result follows from \Cref{lem:richardson} and the diagonal rescaling. 
\end{proof}

\clearpage

\bibliographystyle{alpha}
\bibliography{refs}

\clearpage

\appendix

\section{Linear Algebra}

\begin{fact}
For a block diagonal matrix $\AA = \begin{pmatrix}
\AA_1 & &\\
& \ddots & \\
& & \AA_k
\end{pmatrix}$ we have $\norm{\AA}_2 = \max_i\{\norm{\AA_i}_2\}$
\label{fact:blockdiag}
\end{fact}
\begin{proof}
Consider the matrix $\AA = \begin{pmatrix}
\AA_1  &\\
& \AA_2
\end{pmatrix}$. Then $\norm{\AA}_2 = \max_{\vv: \norm{\vv} = 1} \norm{\AA \vv}_2$ Now every $\vv$ is decomposable into $a\xx$ and $b\yy$ so that $\AA \vv = \begin{pmatrix} a \AA_1 \xx \\ b \AA_2 \yy\end{pmatrix}$ and $\norm{x}_2 = 1$, $\norm{\yy}_2 = 1$ and $a^2 + b^2 = 1$. But then $\norm{\begin{pmatrix} a \AA_1 \xx \\ b \AA_2 \yy\end{pmatrix}}_2^2 = a^2 \norm{\AA_1 \xx}^2_2 + b^2 \norm{\AA_2 \yy} \leq (a^2 + b^2) \max\{\norm{\AA_1 \xx}^2_2 ,\norm{\AA_2 \yy}_2^2\}$. For $k > 2$ the fact follows by induction.
\end{proof}

\section{Degree Preserving Undirected Sparsification}
\label{sec:deg_pres_undir_spars}

We adapt the sparsification routine provided in Section 6.3 of \cite{chuzhoy2020deterministic} to be degree preserving. We first state their result.  

\begin{theorem}[Corollary 6.4 in \cite{chuzhoy2020deterministic}]
There is a deterministic algorithm called \\ $\textsc{SpectralSparsify}(G, r)$ that, given an undirected n-node m-edge graph $G = (V, E, \omega)$ with integer edge weights $\omega$ bounded by $U$, and a parameter $1 \leq r \leq O(\log m)$, computes a $(\log m)^{C \cdot r^2}$-spectral sparsifier $H$ for $G$, with $|E(H)| \leq O(n \log n \log U)$ for some constant $C$, in time
\begin{align*}
        O(m^{1 + O(1/r)}  \cdot (\log m)^{O(r^2)} \log U).
\end{align*}

\label{thm:sparsify_chuzhoy}
\end{theorem}

\begin{algorithm}
\caption{\specspardeg($G, \gamma$)}
\label{alg:specsparsdeg}
$r = ((\log n)^{\gamma/2}/\sqrt{2C})/(\log \log m)$ \\
let $\omega_{min} = \min_{e \in E(G)} \omega(e)$ denote the minimum weight \\
let $\hat{G}$ be  \\
$H = \textsc{SpectralSparsify}(G, r)$ \\
$\AA_{\tilde{G}} = \AA_H/\Exp(\frac{1}{2} \cdot (\log n)^{\gamma}) + \underbrace{\DD_G - \DD_H/\Exp(\frac{1}{2} \cdot (\log n)^{\gamma})}_{\text{self loops}}$ \\
\Return{$\tilde{G}$}
\end{algorithm} 

To make the sparsification degree preserving, we simply add self loops. 

\begin{lemma}(Degree Preserving Sparsification, \Cref{lem:sparsify_undir} restated)
There is a deterministic algorithm \\ $\specspardeg(G, \gamma)$ that given a parameter $\gamma \in (0,1)$ and an undirected graph $G = (V, E, \omega)$ with $n$ vertices and $m$ edges such that that $P := \frac{\max_{e \in E} \omega(e)}{\min_{e \in E} \omega(e)} = poly(n)$ computes $\tilde{G}$ satsifying
\begin{enumerate}
    \item $\Exp(-(\log n)^{\gamma}) \LL_G \preceq \LL_{\tilde{G}} \preceq \Exp((\log n)^{\gamma}) \LL_G$
    \item $\nnz(\AA) = \tilde{O}(n)$
\end{enumerate}
in time
\begin{align*}
    \tilde{O}(m^{1 + O(1/(\log n)^{\gamma/2})}\cdot (\log m)^{O((\log n)^\gamma)}) = m^{1 + o(1)}.    
\end{align*}
The graph $\tilde{G}$ has self loops and exactly the same degrees as $G$. 
\label{lem:sparsify_undir_app}
\end{lemma}
\begin{proof}
The runtime directly follows from the runtime of \Cref{thm:sparsify_chuzhoy}. Further, since for our choice of $r$ we have $(\log m)^{C \cdot r^2} \leq \Exp(\frac{1}{2} \cdot (\log n)^{\gamma})$, we obtain
\begin{align*}
    \Exp(-\frac{1}{2} \cdot (\log n)^{\gamma}) \LL_G \preceq \LL_H \preceq \Exp(\frac{1}{2} \cdot (\log n)^{\gamma}) \LL_G.
\end{align*}
This directly allows us to conclude that for all $v$: $\deg_G(v) \leq \deg_{\tilde{G}}$, and thus the self loops we add are valid and ensure the preservation of degrees. 
Finally, we have
\begin{align*}
    \Exp(-(\log n)^{\gamma}) \LL_G \preceq \LL_{\tilde{G}} \preceq \LL_G
\end{align*}
since self loops cancel and thus do not change the directed Laplacian. This proofs the desired approximation and concludes our proof. 
\end{proof}

\newcommand{\score}{\operatorname{score}}

\section{Sketching the Cholesky Solver}
\label{sec:spars_cholesky}

Very recently, new techniques for analysing the error accumulation in sparsified-cholesky-solvers for directed Laplacians lead to an algorithm with almost optimal dependence on the runtime of the sparsifiation routine \cite{peng2021sparsified}. In this section, we sketch that our deterministic sparsification routines can also be used to derandomize that framework. We first summarize the framework of \cite{kyng2015sparsified, peng2021sparsified}. 

\subsection{Sparsified-Cholesky for Directed Laplacians}

Given a bi-partition $(F, C)$-of the vertex set $V$, the block Cholesky decomposition of an Eulerian Laplacian $\LL = \LL_{\dir{G}}$ is given by
\begin{align*}
    \LL = \begin{pmatrix}
    \II & \veczero \\
    \LL_{C F} \LL_{F F}^{-1} & \II
    \end{pmatrix} 
    \cdot 
    \begin{pmatrix}
    \LL_{F F} & \veczero \\
    \veczero & \schurto{\LL}{F}
    \end{pmatrix} 
    \cdot 
    \begin{pmatrix}
    \II & \LL_{F F}^{-1} \LL_{F C}\\
    \veczero & \II
    \end{pmatrix} 
\end{align*}
where $\schurto{\LL}{F} := \LL_{CC} - \LL_{CF} \LL_{FF}^{-1}\LL_{F C}$. The algorithm of \cite{peng2021sparsified} selects a $\rho$-RCDD (row-column-diagonally-dominant) block $\LL_{FF}$, and then computes the above decomposition, where the inverse of $\LL_{FF}$ is not explicitly computed. Given that it is easy to approximately invert $\rho$-RCDD blocks using iterative schemes, the main obstruction to apply the inverse of $\LL$ is to apply the inverse of $\schurto{\LL}{F}$. 

The augmented matrix view of partial block elimination introduced by \cite{peng2021sparsified} shows that $\schurto{\LL}{F}$ can be explicitly approximated using $O(\log \log n)$-approximate squaring steps with moderate accuracy $\epsilon \approx \frac{1}{\log \log n}$. Since $\schurto{\LL}{F}$ is another Eulerian Laplacian, and through careful patching its explicitly computed sparse approximation retains this property, the above decomposition can be iterated until $C$ has constant size. Since $|F| = \Omega(|F| + |C|)$, $\Theta(\log n)$ steps suffice. 

\subsection{Derandomizing the Sparsified-Cholesky Solver}

There are three randomized pieces in \cite{peng2021sparsified}. 

\begin{itemize}
    \item The \emph{sparsified squaring} routines from \cite{cohen2016almostlineartime} are used to approximate $\schurto{\LL}{F}$. We replace these calls with our deterministic sparsified squaring routine (see \Cref{sec:squaring}). 
    \item In \cite{peng2021sparsified} \emph{global sparsification} is invoked at the start, and after each squaring to avoid any build up of density. We cannot match this strategy using our global sparsification techniques. Therefore, we only occasionally globally sparsify and recurse (see \Cref{sec:globspars} and \Cref{sec:chains}).
    \item The routine for selecting a $\rho$-RCDD subset is randomized \cite{kyng2015sparsified, peng2021sparsified}. We show that this can be done deterministically in \Cref{sec:rho_rcdd}.
\end{itemize}

\paragraph{Schur complement chains. } Since we run our sparsified squaring algorithm without directly following up with global sparsification, the density increases by a factor of $\frac{1}{\epsilon^d}$ after $d$ block Cholesky decomposition steps. Therefore, as in the squaring framework, we cannot afford to go to full depth $\Theta(\log n)$, but have to limit the depth to say $\Theta((\log n)^{1/2})$, such that $\frac{1}{\epsilon^d} = n^{o(1)}$. Then, we invoke our global sparsification technique, and continue on the globally sparsified schur complement. 

\paragraph{The recursive algorithm.} We apply our decomposition just like \cite{peng2021sparsified}, but whenever we reach a global sparsification point, we have to recursively branch to rectify the error this induced. If we set the global sparsification error to $\Exp(O((\log n)^{1/10}))$ as in the squaring algorithm, we obtain a branching factor of $\Exp(O((\log n)^{1/10}))$ and $\Theta((\log n)^{1/2})$ depth. Therefore, the total branching is bounded by $n^{o(1)}$ and the algorithm runs in almost linear time, since all the matrices involved have an almost linear amount of entries if we globally sparsify at the start. 

\subsection{Deterministically Finding a $\rho$-RCDD Subset}
\label{sec:rho_rcdd}

Given a directed and Eulerian Graph $\dir{G} = (V, E, \omega)$ on $n$ vertices, we aim to find a set $S \subseteq V$ so that $|S| > \Omega(n)$ and for each vertex $s$ in $S$
\begin{align*}
    \sum_{(v, s) \in E, v \notin S} \omega(v, s) \geq \rho \deg_{\dir{G}}^-(s)
\end{align*}
and
\begin{align*}
    \sum_{(s, v) \in E, v \notin S} \omega(s, v) \geq \rho \deg_{\dir{G}}^+(s).
\end{align*}
Namely, a $\rho$-fraction of its (weighted) in-neighbours and a $\rho$-fraction of its (weighted) out-neighbours are not in $S$. Such a set $S$ is called a $\rho$-RCDD subset, since it corresponds to a $\rho$-RCDD block of the Eulerian Laplacian $\LL_{\dir{G}}$. See \Cref{fig:rcdd} for an illustration. 

\begin{figure}[h]
    \centering
    \includegraphics[width=4cm]{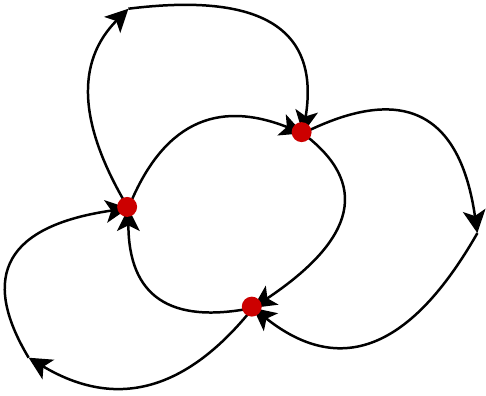}
    \caption{The dark red vertices form a $1/2$-RCDD subset $S$, since for each dark red vertex, half the out-edges leave $S$, and half the in-edges come from outside $S$. }
    \label{fig:rcdd}
\end{figure}

\paragraph{One Condition suffices. }

We first reduce the problem to finding an algorithm that satisfies one of the two conditions. 

\begin{lemma}
Given an algorithm for finding a subset $S$ of $V$ for a directed graph $\dir{G}$ such that
\begin{align*}
    \sum_{(v, s) \in E, v \notin S} \omega((v, s)) \geq \rho \deg_{\dir{G}}^-(s)
\end{align*}
and $|S| = \Omega(n)$ we can find an $\rho$-RCDD subset $S'$ with $|S'| = \Omega(n)$.
\label{lem:red}
\end{lemma}
\begin{proof}
Use the algorithm to find set $S$. Then look at the induced subgraph $\dir{G}' = \dir{G}[S]$. Reverse the direction of all edges in $\dir{G}'$, and use the algorithm again. 
\end{proof}

\paragraph{Satisfying the in-edge Condition. }

Given the previous section, we are left with the task of eliminating a constant weighted fraction of the in-edges. To do so, we define a potential for every set $S \subset V$.

\begin{definition}
Let $\Psi(S) = \sum_{v \in S}  \sum_{(u, v) \in E, u \in S} \frac{\omega((u, v))}{\deg_{\dir{G}}^-(v)}$.
\end{definition}

We first make a simple observation. 

\begin{lemma}
Given a set $S$ such that $\Psi(S)/|S| \leq \frac{1}{2}$, there is a subset $S' \subseteq S$ with $|S'| \geq \frac{1}{4} |S|$ and for each $s \in S'$
\begin{align*}
        \sum_{(v, s) \in E, v \notin S} \omega((v, s)) \geq \frac{1}{4} \deg_{\dir{G}}^-(s)
\end{align*}
\label{lem:pot_to_set}
\end{lemma}
\begin{proof}
We have 
\begin{align*}
    \Psi(S)  = \sum_{v \in S}  \sum_{(v, u) \in E, u \in S} \frac{\omega((u, v))}{\deg_{\dir{G}}^-(v)} &\leq \frac{1}{2}|S| \\
    \frac{1}{|S|} \sum_{v \in S}  \sum_{(v, u) \in E, u \in S} \frac{\omega((u, v))}{\deg_{\dir{G}}^-(v)} &\leq \frac{1}{2}.
\end{align*} 
and we aim to show that there exists a set $S' \subseteq S$ so that for $s \in S'$
\begin{align*}
    \sum_{(s, u) \in E, u \in S} \frac{\omega((u, s))}{\deg_{\dir{G}}^-(s)} &\leq \frac{3}{4}.
\end{align*}
But since the average value is below $\frac{1}{2}$ at least a $\frac{1}{4}$-fraction of the values 
\begin{align*}
    \sum_{(v, u) \in E, u \in S} \frac{\omega((u, v))}{\deg_{\dir{G}}^-(v)}
\end{align*}
are below $\frac{3}{4}$ which shows our claim. 
\end{proof}

We are left with having to construct a set $S$ as in \Cref{lem:pot_to_set}. We define an importance score for each vertex, which we will use as a greedy criterion. 

\begin{definition}
Let $\score(v, S) = \sum_{(v, u) \in E, u \in S} \frac{\omega((v, u))}{\deg_{\dir{G}}^-(u)} - \sum_{(u, v) \in E, u \notin S} \frac{\omega((u, v))}{\deg^-(v)}$ for $v \in S$.
\end{definition}

\begin{lemma}
For $S \subseteq V$ we have
\begin{align*}
    \sum_{v \in S} \score(v, S) = \Psi(S) 
\end{align*}
\end{lemma}
\begin{proof}
We prove the lemma by induction. Initially, we have 
\begin{align*}
    \Psi(V) &= \sum_{v \in V}  \sum_{(u, v) \in E} \frac{\omega((u, v))}{\deg_{\dir{G}}^-(v)} = n = \sum_{v \in V} \score(v, V).
\end{align*}
Now lets assume our lemma holds for $S$, and we show it holds for $S \setminus \{l\}$.
\begin{align*}
    \Psi(S \setminus \{l\}) &= \Psi(S) - \sum_{(u, l) \in E, u \in S} \frac{\omega(u, l)}{\deg^-(l)} - \sum_{(l, u) \in E, u \in S} \frac{\omega(l, u)}{\deg_{\dir{G}}^-(u)} \\
    &= \sum_{v \in S} \score(v, S) - \sum_{(u, l) \in E, u \in S} \frac{\omega(u, l)}{\deg^-(l)} + \sum_{(l, u) \in E, u \in S} \frac{\omega(l, u)}{\deg_{\dir{G}}^-(u)} \\
    &= \sum_{v \in S \setminus \{l\}} \score(v, S \setminus \{ l \}).
\end{align*}
Therefore the claim holds by induction on the size of $S$. 
\end{proof}

Next we introduce our greedy algorithm based on the scores. 

\begin{algorithm}
\caption{\textsc{FindDD}($\dir{G}$)}
\label{alg:finddd}
$S_0 = V$ \\
$i = 0$ \\
\While{$\Psi(S_i) > 0.5 |S_i|$}{
 $v_i = \arg \max_{v \in S_i} \score(v, S_i)$ \\
 $S_{i + 1} = S_i \setminus \{ v_i \}$ \\
 $i = i + 1$ \\
}
$S = S_{i}$; $S' = S$ \\
\For{$v \in S'$}{
    \If{$\sum_{(v, u) \in E, u \in S_{i}} \frac{\omega((u, v))}{\deg_{\dir{G}}^-(v)} > \frac{3}{4}$}{
        $S' = S' \setminus \{ v \}$
    }
}
\Return{$S'$}
\end{algorithm}

\begin{lemma}
If Algorithm \ref{alg:finddd} terminates it returns a set $S'$ so that for all $v \in S$
\begin{align*}
    \sum_{(v, u) \in E, u \in S_{i}} \frac{\omega((u, v))}{\deg_{\dir{G}}^-(v)} \leq \frac{3}{4}
\end{align*}
and $|S'| \geq \frac{1}{4} |S_i|$. 
\label{lem:size_sub}
\end{lemma}
\begin{proof}
Immediately follows from the description of Algorithm \ref{alg:finddd} and \Cref{lem:pot_to_set}. 
\end{proof}

\begin{lemma}
For all $i$: $\Psi(S_i) \leq \Psi(S_{i - 1}) - \frac{3}{2}$
\label{lem:sizered}
\end{lemma}
\begin{proof}
For any $S$ we have 
\begin{align*}
    \Psi(S \setminus \{l\}) &= \Psi(S) - \sum_{(u, l) \in E, u \in S} \frac{\omega(u, l)}{\deg_{\dir{G}}^-(l)} - \sum_{(l, u) \in E, u \in S} \frac{\omega(l, u)}{\deg_{\dir{G}}^-(u)} \\
    &= \Psi(S) - \score(l, S) - 1
\end{align*}
and thus the lemma follows from the maximum being at least the average. 
\end{proof}

\begin{lemma}
For the set $S$ as in Algorithm \ref{alg:finddd} we have $|S| \geq \frac{1}{2}|V|$.
\label{lem:size_s}
\end{lemma}
\begin{proof}
Assume the contrary. Then more than $n/2$ vertices got eliminated, but in iteration $n/2$
\begin{align*}
    \Psi(S_{n/2}) \leq \Psi(V) - \frac{3}{2} n < 0.
\end{align*}
So the while loop must have stopped then, which is a contradiction.
\end{proof}

\begin{lemma}
We can deterministically find a $\frac{1}{4}$-RCDD subset of an $n$-vertex $m$-edge Eulerian graph $\dir{G}$ with at least $\frac{1}{64}n$ vertices in $\tilde{O}(m)$ time. 
\end{lemma}
\begin{proof}
Follows from \Cref{lem:size_s}, \Cref{lem:size_sub} and the proof of \Cref{lem:red}. 
\end{proof}

\section{Preconditioning a Cycle with its Transpose Fails}
\label{sec:cycle}

In this section we show that preconditioning the Laplacian of a length $5$ cycle with its transpose (the directed Laplacian of the graph with edges reversed) does not lead to converging behaviour on some inputs. To this end, we consider the Laplacian
\begin{align*}
    \LL = \begin{pmatrix}
        1 & & & & -1 \\
        -1 & 1 & & & \\ 
        & -1 & 1 & & \\
        & & -1 & 1 & \\
        & & & -1 & 1
    \end{pmatrix}
\end{align*}
of said cycle. Computing the eigenvalues of $(\LL^T)^+ \LL$ yields one eigenvalue $\lambda' \approx -0.3 - 0.95i$ (See \Cref{fig:eigvals}). We conclude that $\II_{\im(\LL)} - \eta (\LL^T)^+ \LL$ has an eigenvalue that is approximately $1 + \eta \cdot 0.3 + \eta \cdot  0.95i$, which is strictly larger than $1$ in magnitude for any step size $\eta > 0$ \footnote{A similar argument can show that no complex step size $\eta$ leads to a converging behaviour by realizing that $(\LL^T)^+ \LL$ has an eigenvalue in each of the quadrants of the convex plane (See \Cref{fig:eigvals}).}. From this we conclude that $\rho(\II_{\im(\LL)} - \eta (\LL^T)^+ \LL) > 1$. For $\EE = \II_{\im(\LL)} - \eta (\LL^T)^+ \LL$ preconditioned Richardson approximates 
\begin{align*}
    \LL^{+} \bb = (\II_{im(\LL)} + \EE + \EE^2 + \EE^3 \dots)(\LL^T)^+ \bb
\end{align*}
by truncating the sum above. From our previous derivation we know that there is a complex vector $\vv = (\LL^T)^+ \hat{\bb}$ (note that $\LL$ and $\LL^{T}$ and their inverses all have the same image) for which the sum does not converge. To show that this behaviour is also exhibited by a real vector, we decompose $\vv = \vv_1 + i\vv_2$ into its real and imaginary part and conclude
\begin{align*}
    \EE \vv = \EE \vv_1 + i\EE \vv_2.
\end{align*}
Notice that $\EE$ is a real valued matrix, and thus $\EE \vv_1$ is the real part of $\EE \vv$. However, from this we can conclude that if $\EE^k \vv$ is large for some $k$, then either $\EE^k \vv_1$ or $\EE^k \vv_2$ must be large by the triangle inequality. 

\begin{figure}[h]
    \centering
    \includegraphics[width=7.5cm]{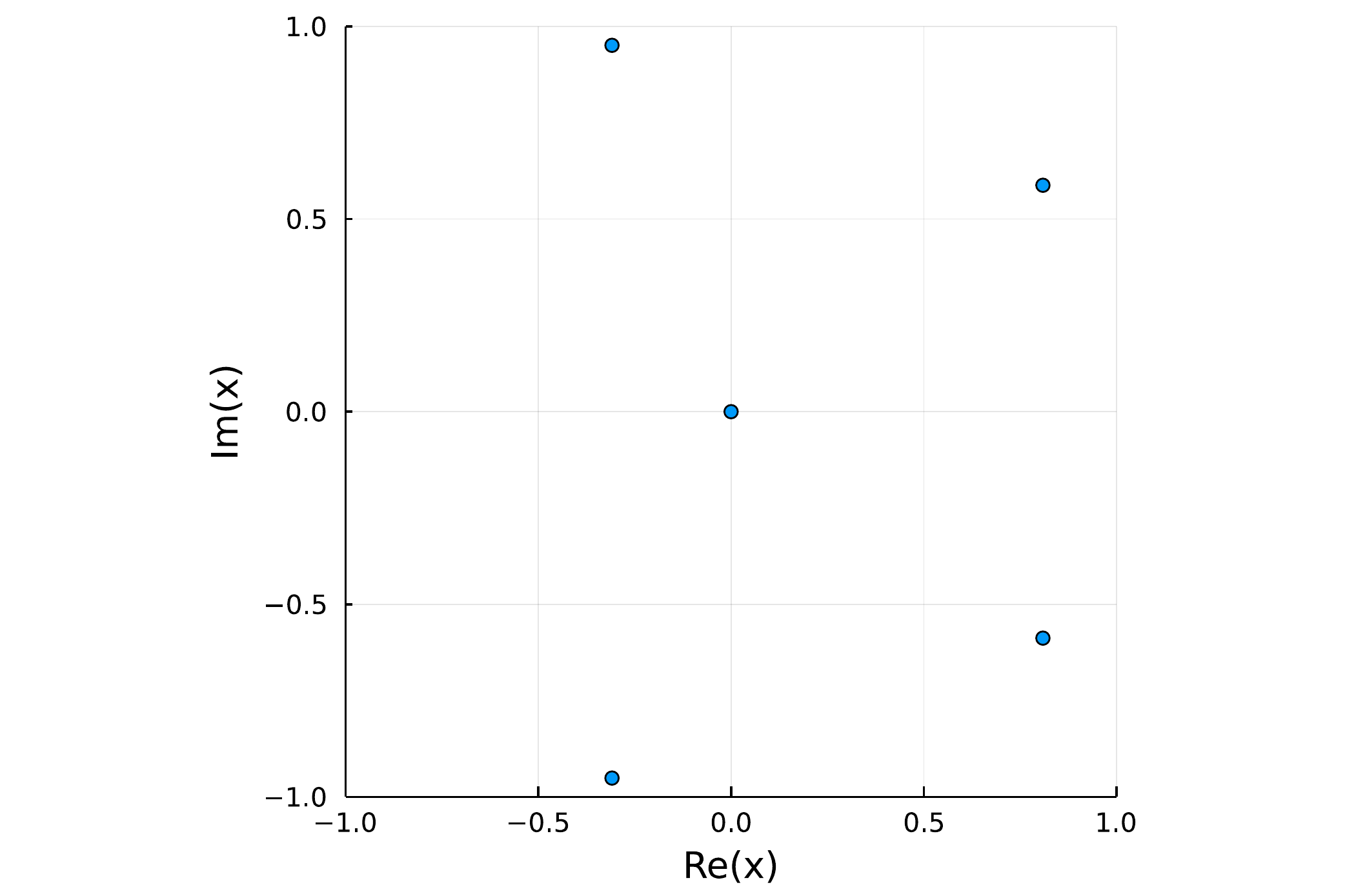}
    \caption{Plot of the eigenvalues of $(\LL^T)^+ \LL$, where $\LL$ is the Laplacian of a directed cycle of length $5$ with unit edge weight.}
    \label{fig:eigvals}
\end{figure}

Since $\LL^T$ is the Laplacian of the same cycle as $\LL$ with reversed edge directions this shows that there are eulerian Laplacians that cannot be used as preconditioners with any step size. This is a significant obstruction for developing a notion of high error sparsification for directed graphs.

\end{document}